\documentclass[sn-mathphys,Numbered]{sn-jnl}% Math and Physical Sciences Reference Style
\usepackage{graphicx}%
\usepackage{multirow}%
\usepackage{amsmath,amssymb,amsfonts}%
\usepackage{amsthm}%
\usepackage{mathrsfs}%
\usepackage[title]{appendix}%
\usepackage{xcolor}%
\usepackage{textcomp}%
\usepackage{manyfoot}%
\usepackage{booktabs}%
\usepackage{algorithm}%
\usepackage{algorithmicx}%
\usepackage{algpseudocode}%
\usepackage{enumitem}
\usepackage{listings}%
\usepackage{bbm}

\newcommand{\Tr}{\operatorname{Tr}}
\newcommand{\dv}{\operatorname{div}}
\newcommand{\Ker}{\operatorname{Ker}}
\newtheorem{theorem}{Theorem}[section]
\newtheorem{proposition}[theorem]{Proposition}
\newtheorem{lemma}[theorem]{Lemma}
\newtheorem{corollary}[theorem]{Corollary}
\newtheorem{definition}[theorem]{Definition}
\newtheorem{remark}[theorem]{Remark}

\newtheorem{hypo}{Hypothesis}

\begin{document}
	
	\title[Stochastic thermodynamics of SPDE]{Nonequilibrium Thermodynamics of Semilinear Stochastic Partial Differential Equations}
	
	%%=============================================================%%
	%% Prefix	-> \pfx{Dr}
	%% GivenName	-> \fnm{Joergen W.}
	%% Particle	-> \spfx{van der} -> surname prefix
	%% FamilyName	-> \sur{Ploeg}
	%% Suffix	-> \sfx{IV}
	%% NatureName	-> \tanm{Poet Laureate} -> Title after name
	%% Degrees	-> \dgr{MSc, PhD}
	%% \author*[1,2]{\pfx{Dr} \fnm{Joergen W.} \spfx{van der} \sur{Ploeg} \sfx{IV} \tanm{Poet Laureate} 
		%%                 \dgr{MSc, PhD}}\email{iauthor@gmail.com}
	%%=============================================================%%
	
	\author[1]{\fnm{Zelin} \sur{Liu}}\email{liu\_zelin@pku.edu.cn}
	
	\author*[1,2]{\fnm{Hao} \sur{Ge}}\email{haoge@pku.edu.cn}
	%\equalcont{These authors contributed equally to this work.}
	
	%\author[1,2]{\fnm{Third} \sur{Author}}\email{iiiauthor@gmail.com}
	%\equalcont{These authors contributed equally to this work.}
	
	\affil*[1]{\orgdiv{Beijing International Center for Mathematical Research (BICMR)}, \orgname{Peking University}, \city{Beijing}, \postcode{100871}, \country{China}}
	
	\affil[2]{\orgdiv{Biomedical Pioneering Innovation Center (BIOPIC)}, \orgname{Peking University}, \city{Beijing}, \postcode{100871}, \country{China}}
	
	%%==================================%%
	%% sample for unstructured abstract %%
	%%==================================%%
	
	\abstract{
		We develop a rigorous mathematical framework for the nonequilibrium thermodynamics of semilinear stochastic partial differential equations. The invariant Gaussian measure of the associated Ornstein–Uhlenbeck process, which need not be time-reversible, serves as the reference measure. We define thermodynamic quantities in terms of probability densities and derive the entropy balance equation. We obtain explicit expressions for the entropy production rate and prove that its vanishing in the stationary state is equivalent to time reversibility. Sample-path entropy production is identified with the logarithm of the Radon–Nikodym derivative of the forward path measure with respect to the time-reversed path measure and is shown to satisfy fluctuation theorems. For nonstationary systems, sample-path entropy production decomposes into housekeeping heat and free-energy dissipation, each satisfying an integral fluctuation theorem. Both have nonnegative expectations, and the mean free-energy dissipation equals the decrease in relative entropy with respect to the stationary distribution. The theory is applied to semilinear stochastic heat equations with spatially correlated noise, a random polymer model, a nonlinear model with a divergence-free drift, and two coupled stochastic heat equations in contact with heat baths at different temperatures.
	}

	%%\pacs[JEL Classification]{D8, H51}
	
	%%\pacs[MSC Classification]{35A01, 65L10, 65L12, 65L20, 65L70}
	
	\maketitle
	
	\section{Introduction}\label{sec1}
	Equilibrium thermodynamics, developed during the nineteenth century, provides a unified framework for describing the interplay among heat, work, and internal energy in macroscopic systems. Central concepts such as entropy, enthalpy, and free energy were introduced to characterize the thermodynamic state of a system and the constraints governing physical transformations \cite{Adkins83}. The relationships among these quantities form the cornerstone of modern physical chemistry and have played a fundamental role in the quantitative understanding of biochemical and biological processes.
	
	In contrast to equilibrium systems, for which a unified thermodynamic theory has been firmly established, the thermodynamic description of nonequilibrium steady states remains an active area of research. Near equilibrium, the response of a system to a weak perturbation can be characterized by linear response theory \cite{OnM2,OnM1}, providing a quantitative connection between fluctuations and dissipation. For systems maintained far from equilibrium by strong driving forces, linear response theory and other near-equilibrium approximations typically cease to be valid.
	
	To investigate nonequilibrium thermodynamics at microscopic and mesoscopic scales, stochastic dynamical models provide a natural mathematical description. Over the past three decades, stochastic thermodynamics has developed into a systematic framework for quantifying irreversibility and energy dissipation in stochastic systems \cite{JQQ,Seifert2012,GQ_PRE10,QH_PR2012a,QH_PR2012b}. The theory has established trajectory-level notions of heat, entropy, and entropy production, revealed universal fluctuation relations, and uncovered a variety of nonequilibrium thermodynamic inequalities \cite{QQ85,QQG,Seki98,Seifert05,Kur,LeSp,Jar1,Cro2,HS,Esposito10a,OP,Esposito07}. Most of these developments, however, have focused on finite-dimensional stochastic processes, such as Markov jump processes (master-equation models) and diffusion processes (Langevin dynamics).
	
	A comparable thermodynamic framework for SPDEs is still lacking, largely because of the analytical difficulties associated with infinite-dimensional state spaces and unbounded operators. While recent work has provided formal thermodynamic descriptions for certain SPDEs with spatially correlated noise \cite{Niggeman22}, a rigorous mathematical theory for the class of semilinear SPDEs with mild solutions studied here has been unavailable. In this work, we develop a mathematical theory of nonequilibrium thermodynamics for semilinear SPDEs driven by Gaussian noise that is white in time and either white or correlated in space.
	
	A fundamental example is the following stochastic heat equation (SHE):
	\begin{equation}
		\label{SHE}
		\left\{\begin{aligned}
			\frac{\partial u(x, t)}{\partial t}&=\Delta u(x, t)+\xi(x, t), \quad &(x, t) \in[0,1] \times \mathbb{R}_{+}, \\
			u(0,t)&\equiv u(1,t),\quad &t\in\mathbb{R}_{+},\\
			u(\cdot,0)&\text{ is a certain initial condition,}
		\end{aligned}\right.
	\end{equation}
	in which $\xi(x,t)$ is a Gaussian noise that is white in time and either correlated or white in space, i.e.
	$$
	\langle\xi(x, t)\rangle=0,\quad\left\langle\xi(x, t) \xi\left(x^{\prime}, t^{\prime}\right)\right\rangle=K\left(x-x^{\prime}\right) \delta\left(t-t^{\prime}\right).
	$$
	Following \cite{Niggeman22}, it is convenient to represent the noise using the orthonormal eigenbasis $\left\{\gamma_k\right\}_{k \in \mathbb{Z}}$ of $\Delta$ on the space $H=L_{\text{per}}^2[0,1]$ of all periodic and square-integrable functions, namely $\gamma_0=1$, $\gamma_k=\sqrt{2} \sin 2 \pi k x$ and $\gamma_{-k}=\sqrt{2} \cos 2 \pi k x$ for $k \in \mathbb{N}$. Then writing $\widetilde{W}(x, t)=\sum_{k \in \mathbb{Z}} \alpha_k \beta_k(t) \gamma_k(x)$ with independent standard Brownian motions $\left\{\beta_k(t)\right\}$ and real coefficients $\left\{\alpha_k\right\}$ determined by the correlation kernel $K$, i.e.  
	$$
	K(x-x^{\prime})=\alpha_0^2+2 \sum_{k \in \mathbb{N}} \alpha_k^2 \cos 2 \pi k (x-x^{\prime}),
	$$
	a direct computation shows that the desired covariance structure is reproduced, $\partial_t \partial_{t^{\prime}}\langle\widetilde{W}(x,t)\widetilde{W}(x^{\prime},t^{\prime})\rangle=K(x-x^{\prime})\delta(t-t^{\prime})$, provided $\alpha_{-k}=\alpha_k$ \cite{Niggeman22}.
	
	Introducing the operator $B \gamma_k=\alpha_k \gamma_k$ and the cylindrical Wiener process $W$ on $H$ (formally $W=\sum_k \beta_k(t) \gamma_k$), the noise term can be rigorously realized as $\mathrm{d} \widetilde{W}=B \,\mathrm{d} W$. 
	
	With this identification, the SHE \eqref{SHE} becomes a linear stochastic evolution equation in the Hilbert space $H$,
	\begin{equation}
		\label{prototype}
		\mathrm{d} Z(t)=A Z(t)\, \mathrm{d} t+B\, \mathrm{d} W(t),\quad t \geqslant 0,
	\end{equation}
	with $Z(t)=u(\cdot,t)$ and $A=\Delta$, which defines an infinite-dimensional Ornstein--Uhlenbeck (OU) process. Equation \eqref{prototype} is, however, purely linear, and therefore describes only the Gaussian relaxation of the field. To capture genuinely nonequilibrium phenomena---such as non-gradient forcing, reaction terms, or polymer-type potentials---one is led to add a nonlinear drift $F$ and study the \emph{semilinear} SPDE, although, as the two-temperature model of Section~\ref{sec:applications} shows, irreversibility can also be generated by the linear dynamics itself. The semilinear prototype is
	\begin{equation}
		\label{semilinear-prototype}
		\mathrm{d} X(t)=[A X(t)+F(X(t))]\, \mathrm{d} t+B\, \mathrm{d} W(t),\quad t \geqslant 0,
	\end{equation}
	which is the central object of this article. Moreover, beyond the self-adjoint case $A=\Delta$, we allow $A$ to be the generator of a general strongly continuous contraction semigroup. This framework encompasses a broad class of stochastic evolution equations, including reaction--diffusion systems. Equations with a first-order transport term are excluded by the boundedness condition on the irreversible part, see Remark~\ref{rmk:Hstar-scope}.

	In finite dimensions, Qian and collaborators rigorously developed a mathematical framework for nonequilibrium steady states of the stochastic differential equation \cite{JQQ}
	$$
	\mathrm{d} \eta_t=b\left(\eta_t\right) \mathrm{d} t+\Gamma\left(\eta_t\right) \mathrm{d} B_t,
	$$
	where $\Gamma(\cdot)$ is a full-rank matrix and $b(\cdot)$ a vector field on $\mathbb{R}^{d}$: explicit expressions for relative entropy and entropy production rate were derived, fluctuation theorems for entropy production were established, and equivalent conditions for time-reversibility were identified. A crucial ingredient in their approach is the existence of a stationary density $\rho(x)=\frac{\mathrm{d} \mu}{\mathrm{d} \lambda}$, where $\mu$ is the invariant measure of $\eta_t$ and $\lambda$ is the Lebesgue measure on $\mathbb{R}^{d}$.
	
	The main difficulty in extending this theory to infinite dimensions is the absence of a Lebesgue measure and hence of a stationary density with respect to such a reference measure. The key observation is that, although infinite-dimensional spaces do not admit a Lebesgue measure, the invariant Gaussian measure of the associated Ornstein--Uhlenbeck process can serve as a natural reference measure, playing a role analogous to that of Lebesgue measure in finite dimensions. More precisely, we compare the semilinear SPDE \eqref{semilinear-prototype} with its corresponding Ornstein--Uhlenbeck process \eqref{prototype}, whose invariant measure is the Gaussian measure $\mu$.
	
	Under the assumption that the invariant measure $\nu$ of the semilinear equation is absolutely continuous with respect to $\mu$, we introduce the density
	$$\rho^{\text{ss}}=\frac{\mathrm{d} \nu}{\mathrm{d} \mu},$$
	which serves as the infinite-dimensional analogue of the stationary density with respect to Lebesgue measure in finite dimensions. In this way, the Gaussian invariant measure of the linearized dynamics replaces the missing Lebesgue measure as the fundamental reference measure of the theory. This observation provides the foundation for extending the core notions of stochastic thermodynamics to a broad class of infinite-dimensional semilinear SPDEs.
	
	Within this setting, we first establish a steady-state thermodynamic framework for semilinear SPDEs. We derive explicit expressions for the sample-path entropy production and the entropy production rate, establish the corresponding fluctuation theorem, and obtain several equivalent characterizations of time reversibility. These results cover both reversible and irreversible Ornstein–Uhlenbeck reference dynamics.
	
	We then move beyond stationary states and consider initial laws that are absolutely continuous with respect to the reference measure. In this setting, we introduce a time-dependent stochastic entropy and derive the corresponding sample-path entropy production. We show that the entropy production admits a natural decomposition into housekeeping heat and free-energy dissipation. Both components have nonnegative expectations and satisfy their respective fluctuation theorems.
	
	Finally, we illustrate the theory through a collection of physically motivated examples. For each model, we determine the reversibility properties of the dynamics and compute the associated thermodynamic quantities explicitly.
	
	During the completion of this work, we became aware of the independent study by Duan, Zhang and Zimmer \cite{DuanZhangZimmer26}, which develops an entropy-production theory for stochastic evolution equations with a reversible Gaussian reference process. Their results overlap with ours in the symmetric stationary setting, particularly in the entropy-production rate formula and the characterization of reversibility. Our work extends beyond this setting by treating mild rather than strong solutions, allowing Sobolev-level density regularity, establishing transient fluctuation theorems and the housekeeping heat–free-energy dissipation decomposition, and accommodating non-reversible reference processes under a boundedness condition on the irreversible part. These aspects are not addressed in \cite{DuanZhangZimmer26}. A complementary recent study by Li, Zhao and Zhu \cite{LiZhaoZhu26} investigates entropy production for stationary diffusions on Hilbert spaces with state-dependent, possibly degenerate noise, without assuming a density relative to a reference measure. They establish a lower bound for entropy production in terms of a Cameron--Martin-type energy and provide separate sufficient conditions for representing the reversed process as a diffusion. This distinction between entropy-production estimates and time-reversal identification is also reflected in our hypothesis \textbf{(H-rev)}.
	
	The remainder of the paper is organized as follows. Section 2 presents preliminaries on semilinear SPDEs. Sections 3 and 4 develop the steady-state and transient thermodynamic frameworks, respectively. Section 5 illustrates the theory with examples from statistical physics. Section 6 concludes with a discussion of possible future directions.   
	
	\section{Preliminaries}\label{sec:prelim}
	In this section we collect the notation and several standard results from the theory of SPDEs that will be used. We refer to \cite[Chapters 2--3]{Prato2004} and \cite[Chapters 5,7,9--11]{Prato1992} for further details.
	
	Let $H$ be a separable Hilbert space. We denote by $C_b(H)$ and $B_b(H)$ the spaces of bounded continuous and bounded Borel measurable functions on $H$, respectively. For $n\in\mathbb N$, $C_b^n(H)$ (respectively $C_b^n(H;H)$) denotes the space of bounded $n$-times continuously Fr\'echet differentiable functions $\varphi:H\to\mathbb R$ (respectively $\varphi:H\to H$) with bounded derivatives up to order $n$. $D\varphi$ and $D^2\varphi$ denote the first and second Fr\'echet derivatives. Relative to an orthonormal basis $\{e_i\}_{i\ge1}$ of $H$, one may write
	\begin{equation}
		\label{Frechet derivatives}
		D\varphi(x)=\sum_{i=1}^{\infty}D_i\varphi(x)e_i,
		\qquad
		D^2\varphi(x)=\sum_{i,j=1}^{\infty}D_iD_j\varphi(x)e_i\otimes e_j,
	\end{equation}
	where $D_i\varphi$ denotes the directional derivative along $e_i$. Likewise, $C_b^{k,n}([0,T]\times H)$ and $C_b^{k,n}([0,T]\times H;H)$ denote the corresponding spaces of functions that are $k$-times continuously differentiable in time and $n$-times continuously Fr\'echet differentiable in space.
	
	Given separable Hilbert spaces $H$ and $G$, we denote by $L_1(H)$ and $L_2(H;G)$ the spaces of trace-class and Hilbert--Schmidt operators, respectively, endowed with the standard norms $\|\cdot\|_1$ and $\|\cdot\|_{\mathrm{HS}}$. For a nonnegative self-adjoint trace-class operator $T$, one has
	\[
	\|T\|_1=\Tr T.
	\]
	
	Let $Q$ be a trace-class operator on $H$ and $m\in H$. We denote by $N(m,Q)$ (or $N_{m,Q}$) the Gaussian measure on $H$ with mean $m$ and covariance operator $Q$. When $m=0$, we simply write $N_Q$. Let $U$ be another separable Hilbert space. By a cylindrical Wiener process on $U$ we mean a cylindrical Wiener process in the sense of \cite{Prato1992}.
	
	Let $\mu=N_Q$ be a centered Gaussian measure on $H$ with $\Ker Q=\{0\}$, and let $\{e_k\}$ be an orthonormal basis of eigenvectors of $Q$. For $F\in C_b^1(H;H)$, we say that $F$ has \emph{finite divergence} if
	\[
	\dv F(x):=\sum_{k=1}^{\infty}D_k\langle F(x),e_k\rangle
	\]
	converges for every $x\in H$ and defines an element of $C_b(H)$. If, in addition,
	\[
	x\mapsto \langle Q^{-1}x,F(x)\rangle,
	\qquad x\in Q(H),
	\]
	extends uniquely to a bounded uniformly continuous function on $H$, then $F$ is said to have \emph{finite divergence with respect to} $\mu$, and we define
	\begin{equation}
		\label{finite divergence}
		\dv_\mu F(x)
		=
		\dv F(x)-\langle Q^{-1}x,F(x)\rangle,
		\qquad x\in H.
	\end{equation}
	
	For $p>0$, the space $L^p(H,\mu)$ is defined in the usual way. We denote by $W^{k,p}(H,\mu)$ the Sobolev space obtained as the completion of $C_b^k(H)$ under the corresponding Sobolev norm. For a nonnegative self-adjoint $C$ we also use the covariance-weighted space $W^{2,2}_C(H,\mu)$, defined by the finiteness of $\int_H\operatorname{Tr}\big[(CD^2\varphi)^2\big]\,\mathrm d\mu$ together with $\int_H|C^{1/2}D\varphi|^2\,\mathrm d\mu<\infty$.
	
	Finally, for a densely defined closed operator $A$ on $H$, we define
	\[
	\mathscr E_A(H)
	:=
	\operatorname{span}
	\Bigl\{
	\varphi_h(x)=e^{i\langle h,x\rangle}
	:
	h\in D(A^*)
	\Bigr\}.
	\]
	For every $p\in[1,\infty)$ and every Gaussian measure $\mu$, the space $\mathscr E_A(H)$ is dense in $L^p(H,\mu)$.
	
	\subsection{Linear Stochastic PDEs}

	Given two separable Hilbert spaces $H$ and $U$, with inner product $\langle\cdot,\cdot\rangle$ and $\langle\cdot,\cdot\rangle_U$ respectively, we first consider the linear SPDE:
	
	\begin{equation}\label{OU}
		\left\{\begin{aligned}
			\mathrm{d}Z(t)&=AZ(t)\,\mathrm{d}t+B\,\mathrm{d}W(t),~t\geqslant 0,\\
			Z(0)&=\xi\in H,\end{aligned}\right.
	\end{equation}
	in which $A: D(A)\subset H\rightarrow H$ and $B: U\rightarrow H$ are linear operators, and $W$ is a $U$-valued cylindrical Wiener process with underlying filtered probability space $(\Omega,\mathscr{F},\{\mathscr{F}_t\}_{t\geqslant 0},\mathbb{P})$.
	
	\begin{hypo}
		\label{hypo_1}
		\ \par
		\begin{enumerate}[label=(\roman*)]
			\item $A: D(A)\subset H\rightarrow H$ is the infinitesimal generator of a strongly continuous semigroup $e^{tA}$. Moreover, there exists $\omega<0$, such that $\| e^{tA}\|\leqslant e^{\omega t}$ holds for all $t\geqslant 0$;
			\item \label{hypo_1.2}$B\in L(U;H)$ is a bounded linear operator. Since the laws of the processes below depend on $B$ only through $C:=BB^*\in L(H)$, we may without loss of generality identify $U$ with $H$ and take $B=\sqrt{C}$, so that $B$ is nonnegative and self-adjoint;
			\item \label{Q_infty}For any $t>0$, the linear operator $Q_tx=\int_0^t e^{sA}Ce^{sA^*}x\,\mathrm{d}s$, $x\in H$, $t\geqslant 0$, is of trace class. In addition, $\sup_{t> 0}\Tr[Q_t]<\infty$.
			
		\end{enumerate}
	\end{hypo}
	
	The \textit{mild solution} of (\ref{OU}) is defined by
	\begin{equation}
		\label{mild solution(linear)}
		Z(t,\xi)=e^{tA}\xi+W_A(t),
	\end{equation}
	
	with the stochastic convolution 
	\begin{equation}
		\label{stochastic convolution}
		W_A(t)=\int_0^t e^{(t-s)A}B\,\mathrm{d}W(s),~t\geqslant 0.
	\end{equation}
	
	The process $Z(t,\xi)$ is called the Ornstein--Uhlenbeck (OU) process associated with $(A,B)$. It will serve as the linear reference process. Its invariant measure will play the role of the reference measure in the thermodynamic framework developed below.
	
	For each $t>0$, $Z(t,\xi)$ is a Gaussian random variable on $H$ with distribution $N_{e^{tA}\xi,Q_t}$ with mean $e^{tA}\xi$ and covariance operator $Q_t$. Since entropy production will be formulated through Radon--Nikodym derivatives of Gaussian measures, we shall repeatedly use the \textit{Feldman--H\'ajek theorem}, recalled in Appendix \ref{Appendix A}, to determine when two Gaussian measures are mutually absolutely continuous.

	A major technical difficulty in extending stochastic thermodynamics to SPDEs is that the mild solution \eqref{mild solution(linear)} generally does not belong to $D(A)$. When $A$ is self-adjoint and negative with eigenvalues $-\lambda_k$, a direct computation gives \[ \mathbb E\,\big\|(-A)^\gamma W_A(t)\big\|^2 =\frac12\sum_{k\geqslant1}\lambda_k^{2\gamma-1}\big(1-e^{-2t\lambda_k}\big)\,\big\|B^*e_k\big\|_U^2, \] so $W_A(t)$ belongs to $D((-A)^\gamma)$ only when this series converges. Already for the one-dimensional heat equation with white noise, where $\lambda_k\asymp k^2$ and $B=\mathrm{Id}$, this requires $\gamma<\tfrac14$. The mild solution is therefore generally not a \textit{strong solution}, but rather a \textit{weak solution} to the original equation (\ref{OU}) in the following sense (see \cite{Prato1992} Theorem 5.4): for all $z\in D(A^{\ast})$ and all $t\in[0,T]$,
	\begin{equation}
		\label{weak solution}
		\langle Z(t,\xi), z\rangle=\langle \xi, z\rangle+\int_0^t\left\langle Z(s,\xi), A^* z\right\rangle \mathrm{d} s+\langle B W(t), z\rangle, \quad \mathbb{P} \text {-a.s. }.
	\end{equation}

	The corresponding semigroup $R_t$ of this OU process on $C_b(H)$ is given by
	$$R_t\varphi(\cdot)=\mathbb{E}[\varphi(Z(t,\cdot))]=\int_H \varphi(y)N_{e^{tA}\cdot,Q_t}(\mathrm{d}y)=\int_H \varphi(e^{tA}\cdot+y)N_{Q_t}(\mathrm{d}y),\quad\forall\varphi\in C_b(H).$$
	
	Under Hypothesis \ref{hypo_1}, the linear SPDE (\ref{OU}) has a unique invariant measure $\mu=N_{Q_{\infty}}$, Gaussian measure with mean $0$ and covariance operator $Q_\infty$, in which
	$$Q_{\infty}x=\int_0^\infty e^{tA}Ce^{tA^*}x\,\mathrm{d}t,~x\in H$$
	is well defined and of trace class according to Hypothesis \ref{hypo_1}\ref{Q_infty}. $Q_\infty$ also satisfies the \textit{Lyapunov equation} (see \cite{Prato2004} Lemma 2.45):
	$$\langle Q_\infty x,A^*y\rangle+\langle A^*x,Q_\infty y\rangle=-\langle Cx,y\rangle,~\forall~x,y\in D(A^*).$$

	Under Hypothesis \ref{hypo_1}, $R_t$ has a unique extension from $C_b(H)$ to $L^p(H,\mu)$ (still denoted by $R_t$). The thermodynamic quantities introduced in later sections will be expressed in terms of the infinitesimal generator and its adjoint. We denote by $L_p$ the infinitesimal generator of $R_t$ in $L^p(H,\mu)$, $p\geqslant1$, the $L^p(H,\mu)$-realization of the Kolmogorov backward generator.
	
	\begin{theorem}[Infinitesimal generator in $L^p$ {[\cite{Prato2004} Theorem 2.38]}]Under Hypothesis \ref{hypo_1},
		for any $p\geqslant 1$, $\mathscr{E}_A(H)$ is a core for $L_p$. Moreover, for any $\varphi\in \mathscr{E}_A(H)$, we have
		$$L_p\varphi(x)=\frac{1}{2}\Tr[CD^2\varphi(x)]+\langle x,A^* D\varphi(x)\rangle,~x\in H.$$
	\end{theorem}
	
	The following symmetry criterion will later be shown to be closely related to time reversibility and vanishing entropy production.
	
	\begin{proposition}[Symmetry {[\cite{Prato2004} Proposition 2.47, Remark 2.48]}]
		\label{prop:symmetry}
		Under Hypothesis \ref{hypo_1}, $L_2$ is symmetric if and only if
		$$\langle Q_\infty x,A^*y\rangle=\langle Q_\infty y,A^*x\rangle,$$
		for all $x,y\in D(A^*)$; or equivalently, $AQ_\infty=Q_\infty A^*$, i.e. $Ce^{tA^*}=e^{tA}C$ for all $t\geqslant 0$; or equivalently, $2AQ_\infty=2Q_\infty A^*=-C$.
		
		If $A=A^*$, then $L_2$ is symmetric if and only if $A$ and $C$ commute.
	\end{proposition}

	\subsection{Semilinear Stochastic PDEs}
	
	Consider the nonlinear SPDE:
	\begin{equation}\label{semilinearSPDE}
		\left\{\begin{aligned}
			\mathrm{d}X(t)&=[AX(t)+F(X(t))]\,\mathrm{d}t+B\,\mathrm{d}W(t),~t\geqslant 0,\\
			X(0)&=\xi\in H,\end{aligned}\right.
	\end{equation}
	with $A,B,W$ as in (\ref{OU}), and a nonlinear drift $F:H\to H$.
	
	\begin{hypo}
		\label{hypo_2}
		\ \par
		\begin{enumerate}[label=(\roman*)]
			\item $F\in C_b(H;H)$, and there exists $K>0$ such that
			$$\left| F(x)-F(y)\right|\leqslant K\left| x-y\right|,~\forall~x,y\in H;$$
			\item $\omega+\kappa<0$, where
			$$\kappa:=\sup\left\{\frac{\langle F(x)-F(y),x-y\rangle}{\left| x-y\right|^2}:x,y\in H,\ x\ne y\right\};$$
			\item $\operatorname{Ker} B=\operatorname{Ker}Q_t=\{0\},~\forall\,t>0$.
			
		\end{enumerate}    
	\end{hypo}

	A predictable $H$-valued process $X(t,\xi),t\in[0,T]$ is said to be a \textit{mild solution} of problem (\ref{semilinearSPDE}) if $
	\mathbb{P}\left(\int_0^T|X(s,\xi)|^2 d s<+\infty\right)=1,
	$
	and, for arbitrary $t\in[0,T]$, we have 
	\begin{equation}
		\label{mild solution(semilinear)}
		X(t,\xi)=e^{t A} \xi+\int_0^t e^{(t-s) A} F(X(s)) \,\mathrm{d} s+W_A(t),~\mathbb{P}\text{-a.s.}.
	\end{equation}
	The transition semigroup on $B_b(H)$ is defined by
	\begin{equation}
		\label{semilinear transition group}
		P_t\varphi(\cdot)=\mathbb{E}[\varphi(X(t,\cdot))], \quad \varphi\in B_b(H), \quad t\geqslant 0.
	\end{equation}
	When $P_t$ admits an invariant probability measure, we denote it by $\nu$.

	\begin{hypo}\ \par
		\label{hypo_3}
		Either alternative {\rm(a)} or alternative {\rm(b)} below holds.
		\begin{enumerate}[label=(\alph*)]
			\item \label{hypo_3.a}\emph{Trace-class noise regime.}
			\begin{enumerate}[label=(\roman*)]
				\item $B\in L_2(U;H)$;
				\item \label{hypo_3.compatibility} $F(H)\subset B(U)$ and $B^{-1}F\in C_b(H;U)$;
				\item \label{hypo_3.density} $\nu\ll\mu$.
			\end{enumerate}
			\item \label{hypo_3.b}\emph{White-noise regime.}
			\begin{enumerate}[label=(\roman*)]
				\item $B(U)$ is dense in $H$ and $B^{-1}\in L(H;U)$. Furthermore, for each $t>0$, $e^{tA}(H) \subset Q_t^{1 / 2}(H)$;
				\item For any $t>0$, $e^{tA}B$ belongs to $L_2(U;H)$, and there exists $\alpha\in(0,\tfrac12)$ such that $\int_0^1 t^{-2\alpha}\,\|e^{tA}B\|_{\operatorname{HS}}^2\,\mathrm dt<\infty$.
			\end{enumerate}
		\end{enumerate}
	\end{hypo}

	\begin{remark}
		When the noise in \eqref{SHE} is correlated with a bounded correlation function $K(x-x')$, Hypothesis~\ref{hypo_3}\ref{hypo_3.a}(i) is readily verified. In the white-noise case $B=\operatorname{Id}\notin L_2(U;H)$, so one must instead verify Hypothesis~\ref{hypo_3}\ref{hypo_3.b}. Here $\|e^{tA}B\|_{\operatorname{HS}}^2=\sum_{k}e^{-2\lambda_k t}$, where $\{-\lambda_k\}$ are the eigenvalues of $A$ (see [\cite{Prato1992} Example~7.6]). In one space dimension $\lambda_k\sim c\,k^2$, so $\|e^{tA}B\|_{\operatorname{HS}}^2=O(t^{-1/2})$ as $t\to0^+$, and hence $\int_0^1 t^{-2\alpha}\,\|e^{tA}B\|_{\operatorname{HS}}^2\,\mathrm dt<\infty$ for every $\alpha\in(0,\tfrac14)$. The integrability condition of Hypothesis~\ref{hypo_3}\ref{hypo_3.b} is therefore satisfied. For the Dirichlet realization of $\Delta$ on $(0,1)$ one has $\lambda_k=\pi^2k^2$, giving $\|e^{tA}B\|_{\operatorname{HS}}^2=\sum_{k\geqslant1}e^{-2\pi^2k^2t}$.
	\end{remark}
	
	\begin{proposition}[mild solution]
		Fix $T>0$. Under Hypotheses~\ref{hypo_1} and~\ref{hypo_2}, for any $\xi\in H$, the SPDE (\ref{semilinearSPDE}) has a unique mild solution $X(t,\xi)$ over $t\in[0,T]$. If additionally either Hypothesis \ref{hypo_3}\ref{hypo_3.a} or \ref{hypo_3}\ref{hypo_3.b} holds, then $\sup_{t\in[0,T]}\mathbb E|X(t)|^p<\infty$ for every $p\geqslant2$, and $X(\cdot,\xi)$ has $\mathbb P$-almost surely $H$-continuous paths, so its law is a Borel probability measure on $C([0,T];H)$.
	\end{proposition}

	\begin{theorem}[Invariant measure and density {\cite[Theorems 3.17, 3.33 and Proposition 3.37]{Prato2004} }]
		\label{Invariant measure & density}
		Under Hypothesis \ref{hypo_1}, Hypothesis \ref{hypo_2} and either Hypothesis \ref{hypo_3}\ref{hypo_3.a} or \ref{hypo_3}\ref{hypo_3.b}, there exists a unique invariant measure $\nu$ for $P_t$ with finite $2n$-th moments for any $n\in \mathbb{N}$ [\cite{Prato2004} Theorem~3.17].
		
		Under Hypothesis~\ref{hypo_3}\ref{hypo_3.b}, the measure $\nu$ is absolutely continuous with respect to $\mu$ and, writing $\rho^{ss}=\frac{d\nu}{d\mu}$, one has $\rho^{ss}\in W^{1,2}(H,\mu)$ and $D\log\rho^{ss}\in L^2(H,\nu;H)$ [\cite{Prato2004} Theorem~3.33 and Proposition~3.37].
		
		Under Hypothesis~\ref{hypo_3}\ref{hypo_3.a},
		condition~\ref{hypo_3}\ref{hypo_3.a}(iii) ensures that
		$\rho^{ss}=d\nu/d\mu$ is well defined.
	\end{theorem}
	
	When $C=\mathrm{Id}$, the invariant measure is known to be absolutely continuous with respect to $\mu$: this is classical for bounded Lipschitz drifts [\cite{Prato2004} Theorem~3.33], and, in the self-adjoint compact-resolvent case, Bogachev, Da Prato and R\"ockner \cite{BogachevDaPratoRoeckner96} prove the stronger conclusion that the square root of the density belongs to $W^{1,2}(H,\mu)$, assuming of the drift only a coordinate-integrability condition. The companion work cited therein treats a bounded linear part with general trace-class noise. The combination relevant here---unbounded $A$ with a general trace-class $C$---was left open in \cite{BogachevDaPratoRoeckner96}, and we are not aware of a subsequent general result. Proposition~\ref{prop:Hdensity-check} below gives an explicit sufficient condition covering the model of Section~\ref{subsec:M1}.
	
	\begin{remark}[Full gradient versus $C$-gradient]
		\label{rmk:weighted-gradient}
		For $C=I$ the full-gradient regularity is [\cite{Prato2004} Proposition~3.37]. For a general trace-class $C$ the full gradient can fail, and the $C$-gradient is the sharp object. Take $Ae_n=-n^2e_n$, $Be_n=n^{-1}e_n$ and the constant drift $F(x)\equiv f$ with $f_n=n^{-2}$, which satisfies all standing hypotheses. Then $\nu=N(m,Q_\infty)$ with $m_n=n^{-4}$, the density is $\rho^{ss}(x)=\exp(\langle Q_\infty^{-1}m,x\rangle-\tfrac12|Q_\infty^{-1/2}m|^2)$ with $Q_\infty^{-1}m\equiv2$, and $D_n\rho^{ss}=2\rho^{ss}$ for every $n$, whence
		\[
		\sum_{n\geqslant1}\|D_n\rho^{ss}\|^2_{L^2(\mu)}=\infty,
		\qquad
		\int_H|C^{1/2}D\rho^{ss}|^2\,\mathrm d\mu
		=4\,\mathbb E_\mu[(\rho^{ss})^2]\sum_{n\geqslant1}n^{-2}<\infty.
		\]
		In this example $\nu\sim\mu$, so condition~\ref{hypo_3}\ref{hypo_3.a}(iii) holds while the full-gradient strengthening fails. All gradient statements below are therefore formulated for the $C$-gradient, which reduces to the full gradient when $C^{-1}\in L(H)$.
	\end{remark}
	
	For the operator manipulations of Sections~\ref{sec:ts-nonsym} and~\ref{sec:transient} we shall also need second-order information on $\rho^{ss}$.
	
	\begin{proposition}[Second-order regularity of the stationary density]\label{prop:reg}
		Assume the hypotheses of Theorem~\ref{Invariant measure & density} and that $F\in C_b^1(H;H)$ has finite divergence with respect to $\mu$ (see~\eqref{finite divergence}).
		\begin{enumerate}[label=\rm(\roman*)]
			\item If $C=\mathrm{Id}$ and $L_2$ is symmetric, equivalently $A=A^{*}$ by Proposition~\ref{prop:symmetry}, then
			\[
			\rho^{ss}\in W^{2,2}(H,\mu),
			\qquad
			|(-A)^{1/2}D\rho^{ss}|\in L^2(H,\mu),
			\]
			in particular $\rho^{ss}\in D(L_2)$ and $L_2\rho^{ss}\in L^2(H,\mu)$.
			\item If $L_2$ is symmetric and $A$ and $C$ commute, then
			\[
			\rho^{ss}\in W^{2,2}_C(H,\mu),
			\qquad
			\int_H\big\langle (-AC)D\rho^{ss},D\rho^{ss}\big\rangle\,\mathrm d\mu<\infty,
			\]
			in particular $\rho^{ss}\in D(L_2)$ and $L_2\rho^{ss}\in L^2(H,\mu)$; if moreover $C^{-1}\in L(H)$, then $W^{2,2}_C(H,\mu)\subseteq W^{2,2}(H,\mu)$.
		\end{enumerate}
		For general (non-commuting) $C$ no characterization of $D(L_2)$ is available [\cite{Prato2004Cambridge} Remark~10.2.9].
	\end{proposition}
	
	\begin{proof}
		(i) is [\cite{Prato2004} Proposition~2.50, Remark~2.51 and Corollary~3.40], whose standing assumptions are the symmetry of $L_2$ and $C=\mathrm{Id}$ (Remark~2.51 requires $A=A^{*}$), matching our regime~\ref{hypo_3}\ref{hypo_3.b} under the added symmetry hypothesis. In the present examples the Hilbert--Schmidt integrability condition is equivalent to $A^{-1}\in L_1(H)$, the hypothesis of [\cite{Prato2004Cambridge} Theorem~10.2.6].
		(ii) is [\cite{Prato2004Cambridge} Theorem~10.2.8], whose hypotheses $A=A^{*}$, $Ce^{tA}=e^{tA}C$ and $CA^{-1}\in L_1(H)$ hold under the stated assumptions, $CA^{-1}\in L_1(H)$ following from $CA^{-1}=-2Q_\infty\in L_1(H)$ in the symmetric case, since $\Tr Q_\infty<\infty$. Note that $A=A^{*}$ is automatic under the stated assumptions: $L_2$-symmetry gives $AQ_\infty=Q_\infty A^{*}$, which combined with the Lyapunov equation yields $AQ_\infty=-\tfrac12C$, so $Q_\infty=-\tfrac12A^{-1}C$ commutes with $C$ because $A$ and $C$ do, and then $A=-\tfrac12CQ_\infty^{-1}=Q_\infty^{-1}(-\tfrac12C)=A^{*}$.
	\end{proof}
	
	On the core $\mathscr E_A(H)$ the Kolmogorov operator of the semilinear
	equation reads
	\[
	\mathcal K\varphi=L_2\varphi+\langle F(\cdot),D\varphi\rangle.
	\]
	We denote its closed $L^2(H,\mu)$-realization by $K_{2,\mu}$, with
	adjoint $K_{2,\mu}^*$ (the forward, or Fokker--Planck, operator). The
	adjoint of its $L^2(H,\nu)$-realization $K_2$ is denoted $K_2^*$.
	
	$P_t$ can be uniquely extended to $L^p(H,\nu)$ as a strongly continuous semigroup of contractions, for $p\geqslant 1$ with infinitesimal generator denoted as $K_p$. The next result identifies this Kolmogorov operator, and, in the presence of symmetric $L_2$, characterizes the drift structure corresponding to detailed balance with respect to the invariant measure.
	
	\begin{hypo}
		\label{hypo_4}
		\ \par
		When $L_2$ is not symmetric, we assume in addition, writing $A^r:=-A-CQ_\infty^{-1}$ and $R:=A^r-A$:
		\begin{enumerate}[label=(\roman*)]
			\item \label{hypo_4.i} $\rho^{ss}\in D(L_2)$;
			\item \label{hypo_4.ii} $R\in L(H)$;
			\item \label{hypo_4.iii} $C^{-1/2}R\in L(H)$ and $R(H)\subset B(U)$.
		\end{enumerate}
	\end{hypo}
	
	\begin{remark}
		\label{rmk:AQauto}
		By item~{\rm(ii)} and the Lyapunov equation, $2AQ_\infty=-C-RQ_\infty$. Since $C$ is bounded and $RQ_\infty$ is of trace class, $AQ_\infty\in L(H)$. In the symmetric case $R=0$ and $AQ_\infty=-\tfrac12C$, so items~{\rm(ii)} and~{\rm(iii)} are automatic, while item~{\rm(i)} holds under the conditions of Proposition~\ref{prop:reg}; none of Hypothesis~\ref{hypo_4} is required for the symmetric theory of Section~3.1 below.
	\end{remark}

	\begin{theorem}[Infinitesimal generator]\label{thm_generator}
		Under Hypotheses~\ref{hypo_1} and~\ref{hypo_2}, $\mathscr{E}_A(H)$ is a core for $K_p$, $p\geqslant1$, and for any $\varphi\in\mathscr{E}_A(H)$,
		$$K_p\varphi(\cdot)=\tfrac{1}{2}\Tr[CD^2\varphi]+\langle\,\cdot\,,A^*D\varphi\rangle+\langle F(\cdot),D\varphi\rangle=L_p\varphi+\langle F(\cdot),D\varphi\rangle.$$
		Assume in addition that $F\in C_b^1(H;H)$ has finite divergence with respect to $\mu$, and that $\nu\ll\mu$ with the $C$-gradient regularity of Theorem~\ref{Invariant measure & density}.
		\begin{enumerate}[label=\rm(\roman*)]
			\item If $L_2$ is symmetric, then $L_2^*=L_2$ and the adjoint operator of $K_2$ in $L^2(H;\nu)$ is determined through the weak identity
			$$\int_H (K_2\varphi)\,\psi\,\mathrm d\nu=\int_H \varphi\,\big[L_2\psi+\langle CD\log\rho^{ss}-F,D\psi\rangle\big]\,\mathrm d\nu,\qquad \varphi,\psi\in\mathscr E_A(H).$$
			\item In general, if moreover Hypothesis~\ref{hypo_4}(i)--(ii) holds, then the adjoint is determined through
			$$\int_H (K_2\varphi)\,\psi\,\mathrm d\nu=\int_H \varphi\,\big[L_2^*\psi+\langle CD\log\rho^{ss}-F,D\psi\rangle\big]\,\mathrm d\nu,\qquad \varphi,\psi\in\mathscr E_A(H).$$
		\end{enumerate}
		In both cases $\mathscr E_A(H)$ is a core for $K_2$, so the weak identity determines the adjoint; the operator $L_2^*$ is identified in Section~\ref{sec:ts-nonsym}. In particular, $K_2$ is symmetric if and only if $Ax+F(x)=\frac{1}{2}CD\log\rho^{ss}(x)-\frac{1}{2}CQ_\infty^{-1}x$.
	\end{theorem}
	
	\begin{proof}
		(i) In the symmetric case the carr\'e-du-champ identity $\int_H L_2\varphi\,\varphi\,\mathrm d\mu=-\tfrac12\int_H|C^{1/2}D\varphi|^2\,\mathrm d\mu$ ([\cite{Prato2004} Proposition~2.39]), together with the integration-by-parts formula [\cite{Prato2004} Lemma~3.38], yields the stated weak identity directly by moving $K_2$ across the $L^2(H,\nu)$ pairing; the only regularity of $\rho^{ss}$ entering is the first-order $C$-gradient integrability stated in the hypotheses (Theorem~\ref{Invariant measure & density}). The second-order contribution falls in the zeroth-order remainder that vanishes since $K_2^*\mathbf 1=0$.
		(ii) is the weak form of Theorem~\ref{adjoint_K_2}, proved in Section~\ref{sec:ts-nonsym} under exactly the hypotheses stated here.
		The symmetry criterion follows by comparing $K_2$ and $K_2^*$.
	\end{proof}
	
	\paragraph{Interpretation of differential expressions.}
	The differential expressions appearing in the proofs of
	Sections~3 and~4 are used formally, with their rigorous
	interpretation provided by the complete generator identity
	in the parabolic Dynkin--Fukushima formula
	(Theorem~\ref{thm:parabolic-ito},
	Appendix~\ref{Appendix B}).
	In particular, terms involving $AX_t$, second-order
	traces, or decompositions involving $C^{-1}$ need not
	be individually well defined. Instead, the relevant
	calculations are justified through the bounded pairing
	\[
	\int_H
	\langle AQ_\infty D\varphi,D\psi\rangle
	\,\mathrm d\mu,
	\]
	using Gaussian integration by parts
	\cite[Lemma~3.38]{Prato2004}
	and the Kolmogorov forward equations for the
	density curves.
	
	Here $AQ_\infty\in L(H)$: automatically so when $L_2$ is symmetric, in which case $2AQ_\infty=-C$; when $L_2$ is not symmetric, this follows from Hypothesis~\ref{hypo_4}(ii) by Remark~\ref{rmk:AQauto}. When sufficient classical regularity is available,
	the individual terms are well defined, and the resulting
	expressions agree with the classical mild It\^o formula
	\cite{DaPratoJentzenRockner2019}.

	\section{Time-reversed process and steady-state entropy production}

	The steady-state theory of this section rests on the following admissibility condition on the stationary thermodynamic force.
	
	\begin{hypo}
		\label{hypo_5}
		\ \par
		The stationary thermodynamic force has at most linear growth:
		\begin{equation}
			\label{force-growth}
			\big|C^{-1/2}\big(2F-CD\log\rho^{ss}\big)(x)\big|
			\leqslant a+b|x|,
			\qquad \mu\text{-a.e. }x\in H,
		\end{equation}
		for some constants $a,b\geqslant0$, together with the corresponding growth bounds for the forward and reversed drifts:
		\[
		|B^{-1}F(x)|_U\leqslant a_F+b_F|x|,
		\qquad
		\big|B^{-1}\big(CD\log\rho^{ss}-F\big)(x)\big|_U
		\leqslant a_r+b_r|x|,
		\]
		for $\mu$-a.e. $x\in H$, with nonnegative constants $a_F,b_F,a_r,b_r$.
	\end{hypo}
	
	\begin{remark}
		These bounds hold automatically under Hypothesis~\ref{hypo_3}\ref{hypo_3.b}, where $C^{-1}\in L(H)$ and $F$ is bounded, and are satisfied in the examples of Section~\ref{sec:applications}. Hypothesis~\ref{hypo_5} also implies
		\[
		C^{1/2}D\log\rho^{ss}\in L^2(H,\nu;H),
		\]
		since $B^{-1}CD\log\rho^{ss}=C^{1/2}D\log\rho^{ss}$ and $\nu$ has finite second moments. Finally, the Gaussian log-Sobolev inequality yields $\rho^{ss}\in L^p(H,\mu)$ for every finite $p$.
	\end{remark}

	\subsection{When the associated Ornstein--Uhlenbeck process is symmetric}
	
	We recall from \cite{JQQ} that the time-reversal of a Markov process
	is again a Markov process, and that in the symmetric case
	Theorem~\ref{thm_generator} identifies the formal generator of the
	time-reversed process with that of the semilinear SPDE of drift
	$\tilde{F}(x)=CD\log\rho^{ss}-F(x)$. Assume throughout this
	subsection that $L_2$ is symmetric in $L^2(H,\mu)$; The identification of the time-reversed process is formulated as a separate hypothesis \textbf{(H-rev)} below.
	
	The entropy production will be obtained by comparing both the forward and the time-reversed dynamics with the same reference OU process.
	The following time-inhomogeneous Girsanov formula provides the required
	change of measure. It follows from \cite[Theorem~10.14]{Prato1992} by
	the arguments of \cite[Theorems~10.18 and~10.20]{Prato1992}; see also
	\cite{FlandoliLuoRicci22,PieperSethmacher24}.

	\begin{proposition}[Girsanov formula for mild solutions]
		\label{prop:girsanov}
		Assume Hypotheses~\ref{hypo_1}--\ref{hypo_3}, and fix $T>0$.
		Let $G:[0,T]\times H\to H$ be Borel measurable with
		$G(s,H)\subset B(U)$, and set
		\[
		\psi(s,x):=B^{-1}G(s,x).
		\]
		Suppose that the equation
		\[
		\mathrm dX(s)
		=
		[AX(s)+G(s,X(s))]\,\mathrm ds+B\,\mathrm dW(s)
		\]
		admits a mild solution unique in law. Let $P_{[0,t]}$ denote its law
		and let $P^*_{[0,t]}$ denote the law of the reference OU process with
		the same initial distribution. On the canonical path space
		$C([0,T];H)$, let $X$ denote the coordinate process and set
		\[
		W^{*}_t:=\int_0^tB^{-1}\bigl(\mathrm dX(s)-AX(s)\,\mathrm ds\bigr),
		\qquad 0\leqslant t\leqslant T,
		\]
		which is a $U$-valued cylindrical Wiener process under $P^*_{[0,T]}$.
		If
		\[
		\mathcal E_t
		=
		\exp\left\{
		\int_0^t
		\langle\psi(s,X(s)),\mathrm dW^{*}(s)\rangle_U
		-\frac12\int_0^t|\psi(s,X(s))|_U^2\,\mathrm ds
		\right\},
		\qquad 0\leqslant t\leqslant T,
		\]
		is a mean-one martingale under $P^*_{[0,T]}$, then
		\[
		P_{[0,t]}\sim P^*_{[0,t]},
		\qquad
		\frac{\mathrm dP_{[0,t]}}
		{\mathrm dP^*_{[0,t]}}(X)
		=
		\mathcal E_t,
		\qquad 0\leqslant t\leqslant T.
		\]
	\end{proposition}
	
	\begin{lemma}[Mean-one martingale under linear growth]\label{lem:girsanov-growth}
		Let $G_0,G_1:[0,T]\times H\to H$ be Borel measurable with $G_0(s,H)\cup G_1(s,H)\subset B(U)$, write $G_2:=G_0+G_1$, and suppose that the two equations
		\[
		\mathrm dX(s)=\big[AX(s)+G_i(s,X(s))\big]\,\mathrm ds+B\,\mathrm dW(s),
		\qquad i=0,2,
		\]
		admit mild solutions unique in law, with laws $P^0_{[0,T]}$ and $P^2_{[0,T]}$ and a common initial distribution. Assume that $\psi(s,x):=B^{-1}G_1(s,x)$ is $U$-valued and satisfies the linear-growth bound
		\begin{equation}
			|\psi(s,x)|_U\leqslant a+b|x|,
			\qquad (s,x)\in[0,T]\times H,
			\label{eq:psi-growth}
		\end{equation}
		with constants $a,b\geqslant0$. Then, under $P^0_{[0,T]}$, the stochastic exponential
		\[
		\mathcal E_t
		=\exp\Big\{\int_0^t\langle\psi(s,X(s)),\mathrm dW^{0}(s)\rangle_U-\frac12\int_0^t|\psi(s,X(s))|_U^2\,\mathrm ds\Big\},
		\]
		with 
		\[
		W^{0}_t:=\int_0^tB^{-1}\big(\mathrm dX(s)-AX(s)\,\mathrm ds-G_0(s,X(s))\,\mathrm ds\big),
		\]
		is a mean-one martingale, $P^2_{[0,t]}\ll P^0_{[0,t]}$ with density $\mathcal E_t$, and exchanging the roles of the two equations gives the converse, so $P^2_{[0,t]}\sim P^0_{[0,t]}$.
	\end{lemma}
	
	\begin{proof}
		Write $X$ for the coordinate process and set $\tau_N:=\inf\{s\geqslant0:|\psi(s,X(s))|_U\geqslant N\}$. By \eqref{eq:psi-growth} and the path continuity, $\tau_N\uparrow\infty$ along $P^0$-almost every path, and in fact along every continuous path whose trajectory obeys \eqref{eq:psi-growth}. The stopped integrand $\psi\,\mathbf 1_{[0,\tau_N]}$ is bounded, so Novikov's condition holds trivially on $[0,T]$, and by [\cite{Prato1992} Theorems 10.14 and 10.18] the stopped exponential $\mathcal E^N_t:=\mathcal E(\psi\,\mathbf 1_{[0,\tau_N]})_t$ is a mean-one martingale under $P^0$, under which the tilted law $\mathrm d\mathbb Q^N=\mathcal E^N_T\,\mathrm dP^0_{[0,T]}$ turns $W^0$ into a cylindrical Wiener process and the coordinate process into a solution of the $G_2$-equation up to time $T\wedge\tau_N$.
		
		Since the $G_2$-equation is unique in law, and uniqueness in law passes to the stopped equations, (a solution of the stopped equation being the concatenation of a solution up to the stopping time with a solution restarted from the stopped state), the law under $\mathbb Q^N$ of the coordinate process stopped at $T\wedge\tau_N$ coincides with $P^2_{[0,T]}$ stopped at $T\wedge\tau_N$. Hence, for every $A$ in the canonical $\sigma$-field $\mathscr F_t$,
		\[
		P^2_{[0,T]}(A\cap\{\tau_N\geqslant t\})
		=\mathbb E^{0}\big[\mathbf 1_A\,\mathcal E^N_t\,\mathbf 1_{\{\tau_N\geqslant t\}}\big]
		=\mathbb E^{0}\big[\mathbf 1_A\,\mathcal E_t\,\mathbf 1_{\{\tau_N\geqslant t\}}\big],
		\]
		using $\mathcal E^N_t=\mathcal E_t$ on $\{\tau_N\geqslant t\}$. By \eqref{eq:psi-growth} the events $\{\tau_N\geqslant t\}$ increase to the whole space under both laws, so monotone convergence gives $P^2_{[0,T]}(A)=\mathbb E^{0}[\mathbf 1_A\,\mathcal E_t]$. Thus $P^2_{[0,t]}\ll P^0_{[0,t]}$ with density $\mathcal E_t$ and $\mathbb E^{0}\mathcal E_t=1$, and since the nonnegative local martingale $\mathcal E$ has constant expectation it is a martingale. Exchanging the roles of the two equations, the perturbing drift $-G_1$ obeying the same bound \eqref{eq:psi-growth}, yields $P^0_{[0,t]}\ll P^2_{[0,t]}$.
	\end{proof}
	
	``$B^{-1}[\mathrm dX-AX\,\mathrm ds]$'' is shorthand for
	$B^{-1}G(s,X(s))\,\mathrm ds+\mathrm dW(s)$, evaluated along
	trajectories of the solution.
	
	The sample-path entropy production $EP^{ss}$ requires Proposition~\ref{prop:girsanov} \emph{twice}: for the forward drift $F$ and for the time-reversed drift $\tilde F=CD\log\rho^{ss}-F$. 
	
	The growth bounds of Hypothesis~\ref{hypo_5} ensure that both drifts are admissible in the sense of Proposition~\ref{prop:girsanov}. Indeed, Lemma~\ref{lem:girsanov-growth} shows that their stochastic exponentials are mean-one martingales under the reference OU law $P^*$.

		\begin{lemma}[Stationary Fokker-Planck equation]\label{lem:rho-ss-stationary}
		Under Hypotheses~\ref{hypo_1}--\ref{hypo_3}, $P_{t,\mu}^*\rho^{ss}=\rho^{ss}$ for every $t\geqslant0$; in particular $K_{2,\mu}^*\rho^{ss}=0$ whenever $\rho^{ss}\in D(K_{2,\mu}^*)$, in which case
		$$\frac{1}{2}\Tr[CD^2\rho^{ss}]+\langle x,A^* D\rho^{ss}\rangle-[\langle F,D\rho^{ss}\rangle+\dv_\mu F\rho^{ss}]=0,$$
		in which $K_{2,\mu}^*$ is the adjoint operator of $K_2$ in the space $L^2(H,\mu)$.
	\end{lemma}
	\begin{proof}
		We have the identity
		$$\int_H P_tf(x)\rho^{ss}(x)d\mu=\int_H f(x)\rho^{ss}(x)d\mu,$$
		followed by $P_{t,\mu}^*\rho^{ss}(x)=\rho^{ss}(x)$, in which $P_{t,\mu}^*$ is the adjoint operator of $P_t$ in the space $L^2(H,\mu)$. Hence $K_{2,\mu}^*\rho^{ss}=0$.
		
		According to Proposition 3.39 in \cite{Prato2004},
		$$K_{2,\mu}^*\varphi=L_2\varphi-[\langle F,D\varphi\rangle+\dv_\mu F\varphi],$$
		in which $\dv_\mu F(x)=\dv F(x)-\langle Q_\infty^{-1}x,F(x)\rangle$ and recall
		$$L_2\varphi(x)=\frac{1}{2}\Tr[CD^2\varphi(x)]+\langle x,A^* D\varphi(x)\rangle.$$
	\end{proof}
	
	The Crooks construction of reversed dynamics
	\cite{Cro2,Seifert05} identifies the backward process
	through the reversed thermodynamic force, with the same noise coefficient. In our infinite-dimensional setting,
	we formulate this identification as an explicit assumption. Write $\Theta_t$ for the time-reversal map on the path space, $(\Theta_t\omega)(s):=\omega(t-s)$, and $(\Theta_t)_\#P$ for the push-forward of a law $P$ under $\Theta_t$.
	
	\begin{itemize}
		\item[\textbf{(H-rev)}]
		For every initial law $\lambda\ll\mu$ whose density
		curve $\rho^{\lambda}$ satisfies the growth condition
		of Hypothesis~\ref{hypo_5}, with $\rho^{ss}$
		replaced by $\rho^{\lambda}$ uniformly in time,
		and for every $t>0$, the time-reversed law
		$\Theta_t{}_\#P^{\lambda}_{[0,t]}$ coincides with
		the law of a mild solution of
		\[
		\mathrm dY(s)
		=
		\big[A^{r}Y(s)
		+CD\log\rho^{\lambda}_{t-s}(Y(s))
		-F(Y(s))\big]\,\mathrm ds
		+B\,\mathrm dW(s),
		\qquad Y(0)\sim\lambda_t,
		\]
		where $A^{r}=A$ when $L_2$ is symmetric and
		$A^{r}=-A-CQ_\infty^{-1}$ in general.
	\end{itemize}
	
	In finite dimensions, this identification follows
	from the classical time-reversal theorem
	\cite{HaussmannPardoux86,JQQ}.
	Related infinite-dimensional results include
	\cite{FoellmerWakolbinger86} under a locally finite
	specific entropy condition,
	\cite{MilletNualartSanz89} on weighted sequence
	spaces,
	\cite{Pidstrigach24} for the Ornstein--Uhlenbeck
	case $A=-\tfrac12\mathrm{Id}$, $F=0$, with
	arbitrary initial laws, and
	\cite{LiZhaoZhu26} under sufficient conditions
	for state-dependent noise.
	Since these results do not directly cover our
	setting, we retain \textbf{(H-rev)} as an assumption.
	It is verified for linear systems with Gaussian
	initial laws in Proposition~\ref{prop:Hrev-linear}.

	For the derivation below, we distinguish three dynamics:
	
	(i) the original semilinear SPDE $X(t)$ with drift $F$;
	
	(ii) its time-reversed counterpart $X^{rev}(t)$ with drift
	$\tilde F=CD\log\rho^{ss}-F$;
	
	(iii) the associated Ornstein--Uhlenbeck process.
	
	All three processes are assumed to start from the measure
	$\nu$.
	
	We define the sample-path steady-state entropy production in terms of
	the stationary thermodynamic force,
	\begin{equation}
		\label{eq:EPss-def}
		\begin{split}
			EP^{ss}(t)
			&:=\int_0^t \langle B^{-1}[2F-CD\log\rho^{ss}](X(s)),dW\rangle_U\\
			&\quad+\frac{1}{2}\int_0^t\langle B^{-1}[2F-CD\log\rho^{ss}](X(s)),B^{-1}[2F-CD\log\rho^{ss}](X(s))\rangle_U ds.
		\end{split}
	\end{equation}
	%Under \textbf{(H-rev)} this functional coincides with the logarithm of the
	%ratio of the forward and time-reversed path probabilities
	%(Theorem~\ref{thm:EPss-sym}(iii) below).
	
	\begin{theorem}[Steady-state entropy production, with $L_2$ symmetric]\label{thm:EPss-sym}
		Under Hypotheses~\ref{hypo_1}--\ref{hypo_3} and Hypothesis~\ref{hypo_5},
		with both $F$ and the reversed drift $\tilde F=CD\log\rho^{ss}-F$ admissible,
		the following hold.
		\begin{enumerate}[label=\rm(\roman*)]
			\item $\mathbb{E} e^{-EP^{ss}(t)}=1$ for all $t>0$.
			\item The entropy production rate $epr:=\frac{\mathbb{E} EP^{ss}(t)}{t}$ is
			independent of $t$ and can be expressed as
			\[
			epr=\frac{1}{2}\int_H\langle B^{-1}[2F-CD\log\rho^{ss}](x),B^{-1}[2F-CD\log\rho^{ss}](x)\rangle_U d\nu.
			\]
			\item If in addition \textbf{(H-rev)} holds, then
			\[
			EP^{ss}(t)=\log\frac{dP^{ss}_{[0,t]}}{dP^{rev,ss}_{[0,t]}},
			\]
			where $P^{ss}_{[0,t]}$ is the law of $\{X(s):0\leqslant s\leqslant t\}$
			and $P^{rev,ss}_{[0,t]}$ is the law of its time reversal
			$\{X^{rev}(s)=X(t-s):0\leqslant s\leqslant t\}$, both under the
			steady-state condition.
		\end{enumerate}
	\end{theorem}
	
	\begin{proof}
		(i) Writing $\psi:=B^{-1}[2F-CD\log\rho^{ss}]$, one has
		$e^{-EP^{ss}(t)}=\mathcal E(-\psi)_t$, the stochastic exponential of
		$-\psi$. Since the bound of Hypothesis~\ref{hypo_5} applies to
		$\pm\psi$, Lemma~\ref{lem:girsanov-growth}, applied in both directions
		between the two equations, shows that $\mathcal E(\pm\psi)_t$ are
		mean-one martingales.
		
		(ii) The stochastic integral in \eqref{eq:EPss-def} is a
		square-integrable martingale vanishing at $0$, whence
		\[
		\mathbb{E} EP^{ss}(t)=\frac{1}{2}\int_0^t\int_H\langle B^{-1}[2F-CD\log\rho^{ss}](x),B^{-1}[2F-CD\log\rho^{ss}](x)\rangle_U d\nu ds=t\cdot epr
		\]
		by the stationarity of $\nu$.
		
		(iii) By \textbf{(H-rev)} with $\lambda=\nu$ and $\rho^{\lambda}\equiv\rho^{ss}$,
		the reversed law $P^{rev,ss}_{[0,t]}$ is the law of the adjoint SPDE,
		i.e.\ of the solution of
		$\mathrm dX=[AX+CD\log\rho^{ss}(X)-F(X)]\,\mathrm dt+B\,\mathrm dW$
		started from $\nu$. Both dynamics are admissible, so
		Proposition~\ref{prop:girsanov} applies to each, with $P^*_{[0,t]}$
		the law of the corresponding OU process started from $\nu$:
		\begin{align}
			\log\frac{dP^{ss}_{[0,t]}}{dP^*_{[0,t]}}(X(\cdot))&=\int_0^t \langle B^{-1}F(X),B^{-1}[dX-AXds]\rangle_U \nonumber\\
			&-\frac{1}{2}\int_0^t\| B^{-1}F(X)\|_U^2ds,
		\end{align}
		\begin{align}
			\log\frac{dP^{rev,ss}_{[0,t]}}{dP^*_{[0,t]}}(X(\cdot))&=\int_0^t \langle B^{-1}[CD\log\rho^{ss}-F](X),B^{-1}[dX-AXds]\rangle_U \nonumber\\
			&-\frac{1}{2}\int_0^t\| B^{-1}[CD\log\rho^{ss}-F](X)\|_U^2ds.
		\end{align}
		Subtracting and using $dX-AXds=F(X)ds+BdW$,
		\begin{align}
			\log\frac{dP^{ss}_{[0,t]}}{dP^{rev,ss}_{[0,t]}}
			&=\int_0^t \langle B^{-1}[2F-CD\log\rho^{ss}](X),B^{-1}[dX-AXds]\rangle_U \nonumber\\
			&+\frac{1}{2}\int_0^t\langle -B^{-1}CD\log\rho^{ss}(X),B^{-1}[2F-CD\log\rho^{ss}](X)\rangle_U ds\nonumber\\
			&=\int_0^t \langle B^{-1}[2F-CD\log\rho^{ss}](X),dW\rangle_U \nonumber\\
			&+\frac{1}{2}\int_0^t\langle B^{-1}[2F-CD\log\rho^{ss}](X),B^{-1}[2F-CD\log\rho^{ss}](X)\rangle_U ds,
		\end{align}
		which is precisely $EP^{ss}(t)$ as defined in \eqref{eq:EPss-def}.
	\end{proof}
	
	\begin{remark}[Bounded vs.\ unbounded $C^{-1}$: finiteness of entropy production]
		\label{rmk:EPfinite}
		The dichotomy in Hypothesis~\ref{hypo_3} between bounded and unbounded $C^{-1}$ governs not merely the well-posedness of thermodynamic functionals, but the very finiteness of the mean entropy production. Indeed, the entropy production rate can be written as
		$$epr=\tfrac12\int_H \big\langle C^{-1}\big[2F-CD\log\rho^{ss}\big],\,2F-CD\log\rho^{ss}\big\rangle\,\mathrm d\nu .$$
		Under Hypothesis~\ref{hypo_3}\ref{hypo_3.b}, $C^{-1}\in L(H)$ and this integral is finite for every bounded Lipschitz drift $F$. Under Hypothesis~\ref{hypo_3}\ref{hypo_3.a}, $C^{-1}$ is unbounded and $epr$ is finite only when the nonequilibrium drift $2F-CD\log\rho^{ss}$ lies in $\operatorname{Range}(C^{1/2})$ for $\nu$-a.e.\ $x$, with integrable $C^{-1/2}$-norm; otherwise it diverges. This is the infinite-dimensional manifestation of the ultraviolet divergence of the entropy production rate observed in field-theoretic settings \cite{Niggeman22}, and it singles out the bounded-$C^{-1}$ regime as the most natural setting in which the steady-state thermodynamic quantities are automatically well defined. In the reversible case $2F=CD\log\rho^{ss}$ the integrand vanishes identically, so $epr=0$ regardless.
	\end{remark}
	
	\begin{theorem}[Characterizations of reversibility]\label{Thm_equi_sym}
		Under Hypotheses~\ref{hypo_1}--\ref{hypo_3} and Hypothesis~\ref{hypo_5}, the following statements are equivalent.
		\begin{enumerate}[label=\rm(\arabic*)]
			\item $X(t)$ is time reversible.
			\item The semigroup $P_t$ is symmetric in $L^2(H,\nu)$.
			\item The generator $K_2$ is symmetric.
			\item $epr=0$.
			\item $EP^{ss}(t)=0$ almost surely.
			\item $2F=CD\log\rho^{ss}$ for $\nu$-a.e.\ $x\in H$.
			\item The system is a gradient system: there exists a measurable
			potential $U:H\to\mathbb R$, with $e^{-U}\in L^2(H,\mu)$ and
			$C^{1/2}DU\in L^2(H,\mu)$, such that
			$2F=-CDU$ in the Gaussian Sobolev sense.
		\end{enumerate}
	\end{theorem}
	
	\begin{proof}
		The equivalences $(1)\Leftrightarrow(2)\Leftrightarrow(3)$ are classical \cite{JQQ}. By the expression of the adjoint in Theorem~\ref{thm_generator}, $K_2$ is symmetric if and only if $2F=CD\log\rho^{ss}$, which is $(3)\Leftrightarrow(6)$. The rate $epr$ is a weighted $L^2(\nu)$ norm of $2F-CD\log\rho^{ss}$, so it vanishes precisely when $2F=CD\log\rho^{ss}$ holds $\nu$-almost everywhere, giving $(4)\Leftrightarrow(6)$. By (ii) of Theorem~\ref{thm:EPss-sym} and the already established
		equivalence with $(6)$, $epr=0$ holds if and only if
		$2F=CD\log\rho^{ss}$ $\nu$-a.e., in which case both integrands in
		\eqref{eq:EPss-def} vanish identically and $EP^{ss}(t)=0$ almost
		surely; conversely $EP^{ss}(t)=0$ a.s.\ implies
		$\mathbb{E}\,EP^{ss}(t)=0$, hence $epr=0$. This gives
		$(4)\Leftrightarrow(5)$.
		
		For $(6)\Rightarrow(7)$, take
		\[
		U=-\log\rho^{ss}.
		\]
		By Theorem~\ref{Invariant measure & density} together with Hypothesis~\ref{hypo_5}, $e^{-U}=\rho^{ss}$ has a square-integrable $C$-gradient under $\nu$. The Sobolev chain rule and $(6)$ give
		\[
		C^{1/2}D(e^{-U})
		=
		C^{1/2}D\rho^{ss}
		=
		\rho^{ss}C^{1/2}D\log\rho^{ss}
		=
		2\,C^{-1/2}F\, e^{-U},
		\]
		the last step by $(6)$ and $C^{1/2}D\log\rho^{ss}=B^{-1}CD\log\rho^{ss}$. Multiplying by $C^{1/2}$ yields $2F=-CDU$ in the Gaussian Sobolev sense, and $(7)$ follows.

		Conversely, assume $(7)$ and set
		\[
		g=e^{-U},
		\qquad
		Z_U=\int_Hg\,\mathrm d\mu.
		\]
		Since $g\in L^2(H,\mu)$ and $\mu$ is a probability measure,
		$0<Z_U<\infty$. In terms of the Sobolev derivative of $g$,
		condition~$(7)$ reads
		\[
		CDg=2Fg
		\qquad \mu\text{-a.e.}
		\]
		For every $\varphi\in\mathscr E_A(H)$, the symmetry of $L_2$ and the
		carr\'e-du-champ identity yield
		\begin{align}
			\int_H K_2\varphi\,g\,\mathrm d\mu
			&=
			-\frac12\int_H
			\langle C^{1/2}D\varphi,C^{1/2}Dg\rangle\,\mathrm d\mu
			+\int_H\langle F,D\varphi\rangle g\,\mathrm d\mu
			\nonumber\\
			&=
			-\frac12\int_H\langle D\varphi,CDg\rangle\,\mathrm d\mu
			+\int_H\langle F,D\varphi\rangle g\,\mathrm d\mu
			=0.
		\end{align}
		By the core characterization of the adjoint, this is the stationary
		equation $K_{2,\mu}^*g=0$. Hence
		\[
		\nu_U:=Z_U^{-1}g\,\mu
		\]
		is an invariant probability measure. The uniqueness of the invariant
		measure implies
		\[
		\rho^{ss}=Z_U^{-1}g.
		\]
		Using the Sobolev chain rule once more, we obtain
		\[
		CD\log\rho^{ss}
		=
		\frac{CDg}{g}
		=
		2F
		\qquad \nu\text{-a.e.},
		\]
		which is $(6)$.
	\end{proof}

	\subsection{When the associated Ornstein--Uhlenbeck process is not symmetric}
	\label{sec:ts-nonsym}
	
	When $L_2$ is symmetric, the time-reversed process has the \emph{same} linear part $A$, and the entropy production follows directly from Proposition~\ref{prop:girsanov}. When $L_2$ is not symmetric, the time-reversed Ornstein--Uhlenbeck process has linear drift $-A-CQ_\infty^{-1}$, which differs from $A$. Then, besides the Girsanov transformations for the nonlinear drifts, one must also compare the path laws of two Ornstein--Uhlenbeck processes with different linear parts. Whether two such Gaussian path laws are equivalent or mutually singular is decided by the Feldman--H\'ajek theorem, which we recall in Appendix~\ref{Appendix A}. Its specialization to Ornstein--Uhlenbeck processes is the Hilbert--Schmidt criterion of [\cite{Prato1992} Theorem~10.8 and Lemma~10.10].

	The next theorem gives the adjoint of the OU generator. Formally, it corresponds to the Ornstein--Uhlenbeck dynamics whose linear drift is $-A-CQ_\infty^{-1}=Q_\infty A^*Q_\infty^{-1}$, which involves the unbounded $Q_\infty^{-1}$, and all identities below are understood on the core $\mathscr E_A(H)$, on which the relevant pairings are finite.
	
	\begin{theorem}\label{adjoint_L_2}
		Under Hypotheses~\ref{hypo_1}--\ref{hypo_3} and Hypothesis~\ref{hypo_4}(ii), the adjoint operator of $L_2$ in $L^2(H;\mu)$, i.e. the generator of the time-reversed OU process, is
		$$L_2^*\varphi(x)=\frac{1}{2}\Tr[CD^2\varphi(x)]+\langle x,[-A-CQ_\infty^{-1}]^* D\varphi(x)\rangle,~x\in H,~\varphi\in \mathscr{E}_A(H).$$
	\end{theorem}
	
	\begin{proof}
		By [\cite{Prato2004} Proposition 2.46], for any $\phi,\psi\in \mathscr{E}_A(H)$,
		$$\int_H L_2\phi\,\psi\, d\mu=\int_H\langle Q_\infty D\psi,A^*D\phi\rangle d\mu.$$
		It therefore suffices to prove that
		$$\int_H\langle Q_\infty D\psi,A^*D\phi\rangle d\mu=\int_H\langle Q_\infty D\phi,[-A-CQ_\infty^{-1}]^*D\psi\rangle d\mu.$$
		Since $AQ_\infty\in L(H)$ by Hypothesis~\ref{hypo_4}(ii) and Remark~\ref{rmk:AQauto}, also $Q_\infty A^*=(AQ_\infty)^*\in L(H)$, and the Lyapunov identity holds between bounded operators,
		$$AQ_\infty+Q_\infty A^*=-C.$$
		For $x,y\in D(A^*)$, interpreting $\langle Q_\infty x,A^*y\rangle:=\langle x,Q_\infty A^*y\rangle$ through the bounded operator, this gives
		$$\langle Q_\infty x, A^*y\rangle+\langle A^*x, Q_\infty y\rangle
		=\langle x,(Q_\infty A^*+AQ_\infty)y\rangle=-\langle x,Cy\rangle=-\langle Cx,y\rangle,
		$$
		where the last equality uses the self-adjointness of $C$, and the right-hand side of the target identity equals the same quantity because $(-A-CQ_\infty^{-1})Q_\infty=-(AQ_\infty+C)$ is bounded. Here both sides are finite because, for $\varphi=\varphi_h\in\mathscr E_A(H)$ with $h\in D(A^*)$, the gradients $D\varphi$ lie in $D(A^*)$, so that all pairings above are well defined through bounded operators. Substituting $x=D\phi$ and $y=D\psi$ into the displayed identity and using the symmetry of the real inner product gives the pointwise equality of the two integrands, and integration over $H$ with respect to $\mu$ completes the proof.
	\end{proof}
	
	\begin{remark}
		In the finite-dimensional case, writing $\Sigma$ for the covariance matrix of $\mu$, this formula reduces to the replacement of $Ax$ by $CD\log\mu(x)-Ax=(-C\Sigma^{-1}-A)x$ in $L_2$.
	\end{remark}
	
	Combining the expression of $L_2^*$ with the stationary density $\rho^{ss}$, we can now identify the generator of the time-reversed semilinear SPDE.
	
	\begin{theorem}
		\label{adjoint_K_2}
		In addition to Hypotheses~\ref{hypo_1}--\ref{hypo_3} and Hypothesis~\ref{hypo_4}(i)--(ii), assume that $F\in C_b^1(H;H)$ has finite divergence with respect to $\mu$. The adjoint operator of $K_2$ in $L^2(H;\nu)$ is determined through the weak identity
		$$\int_H (K_2\varphi)\,\psi\,\mathrm d\nu=\int_H \varphi\,\big[L_2^*\psi+\langle CD\log\rho^{ss}-F,D\psi\rangle\big]\,\mathrm d\nu,\qquad \varphi,\psi\in\mathscr E_A(H),$$
		which holds since $\mathscr E_A(H)$ is a core for $K_2$.
		$K_2$ is symmetric if and only if $Ax+F(x)=\frac{1}{2}CD\log\rho^{ss}(x)-\frac{1}{2}CQ_\infty^{-1}x$.
	\end{theorem}
	
	\begin{proof}
		The argument parallels Chapter 3 of \cite{Prato2004}; we make explicit the points at which a general self-adjoint $C$ and the regularity of $\rho^{ss}$ enter. $\phi,\psi\in\mathscr E_A(H)$. The regularity in force is $\rho^{ss}\in D(L_2)$, by Hypothesis~\ref{hypo_4}(i), with $C^{1/2}D\log\rho^{ss}\in L^2(H,\nu;H)$ by Theorem~\ref{Invariant measure & density}. Thus the products $\psi\rho^{ss}$ lie in the domain of the Ornstein--Uhlenbeck form, and the identity [\cite{Prato2004} Proposition 2.46] extends, by density, from $\mathscr E_A(H)$ to such test functions. Moreover $AQ_\infty\in L(H)$ by Hypothesis~\ref{hypo_4}(ii) and Remark~\ref{rmk:AQauto}, and every term containing $AQ_\infty D\log\rho^{ss}$ below is read through the standing interpretation of Section~2.
		
		We would like to show that
		$$\int_H K_2\phi\,\psi \,d\nu=\int_H \phi K^*_2\psi \,d\nu,$$
		i.e.
		$$\int_HK_2\phi\,\psi\rho^{ss}d\mu=\int_H \phi K^*_2\psi \,\rho^{ss}d\mu.$$
		
		We first have
		\begin{align}
			&\int_HK_2\phi\psi\rho^{ss}d\mu\nonumber\\
			&=\int_H L_2\phi\psi\rho^{ss}d\mu+\int_H \langle F(x), D\phi\rangle \psi\rho^{ss}d\mu\nonumber\\
			&=\int_H \langle Q_\infty D(\psi\rho^{ss}), A^*D\phi\rangle d\mu+\int_H \langle F(x), D\phi\rangle \psi\rho^{ss}d\mu\nonumber\\
			&=\int_H \langle Q_\infty D\psi, A^*D\phi\rangle \rho^{ss}d\mu+\int_H \langle Q_\infty D\rho^{ss}, A^*D\phi\rangle \psi d\mu+\int_H \langle F(x), D\phi\rangle \psi\rho^{ss}d\mu\nonumber\\
			&=\int_H \phi L_2^*\psi\rho^{ss}d\mu-\int\langle Q_\infty D\psi,A^*D\log\rho^{ss}\rangle \phi \rho^{ss}d\mu+\int_H \langle Q_\infty D\log\rho^{ss}, A^*D\phi\rangle \psi \rho^{ss}d\mu\nonumber\\
			&+\int_H \langle F(x), D\phi\rangle \psi\rho^{ss}d\mu.
		\end{align}
		
		Due to Lemma 3.38 in \cite{Prato2004}:
		\begin{align}
			&\int_H \langle F(x), D\phi\rangle \psi\rho^{ss}d\mu=-\int_H \phi \langle F(x), D(\psi\rho^{ss})\rangle d\mu-\int_H \dv_\mu F \phi\psi\rho^{ss}d\mu\nonumber\\
			&=-\int_H \phi \langle F(x), D\psi\rangle \rho^{ss}d\mu-\int_H \phi \langle F(x), D\log\rho^{ss}\rangle \psi\rho^{ss}d\mu-\int_H \dv_\mu F \phi\psi\rho^{ss}d\mu.
		\end{align}
		and similarly
		\begin{align}
			&\int_H \langle Q_\infty D\log\rho^{ss}, A^*D\phi\rangle \psi \rho^{ss}d\mu=\int_H \langle AQ_\infty D\log\rho^{ss}, D\phi\rangle \psi \rho^{ss}d\mu\nonumber\\
			&=-\int_H \phi \langle AQ_\infty D\log\rho^{ss}, D\psi\rangle \rho^{ss}d\mu-\int_H \phi \langle AQ_\infty D\log\rho^{ss}, D\log\rho^{ss}\rangle \psi\rho^{ss}d\mu\nonumber\\
			&-\int_H \dv_\mu (AQ_\infty D\log\rho^{ss}) \phi\psi\rho^{ss}d\mu.
		\end{align}
		
		Hence because of $AQ_\infty +(AQ_\infty)^*=-C$,
		\begin{align}
			K_2^*\psi&=L_2^*\psi-\langle (AQ_\infty +(AQ_\infty)^*)D\log\rho^{ss}+F(x), D\psi\rangle+V\psi\nonumber\\
			&=L_2^*\psi+\langle CD\log\rho^{ss}-F(x), D\psi\rangle+V\psi,
		\end{align}
		in which
		\begin{align}
			V&=-\langle AQ_\infty D\log\rho^{ss}, D\log\rho^{ss}\rangle-\dv_\mu (AQ_\infty D\log\rho^{ss})\nonumber\\
			&-\langle F(x), D\log\rho^{ss}\rangle-\dv_\mu F.
		\end{align}
		
		Since $K_2^*1=0$, $V=0$.
	\end{proof}
	
	Therefore, by \textbf{(H-rev)} with $\lambda=\nu$, the time-reversed process of a stationary semilinear SPDE is the stationary semilinear SPDE with linear operator $A^{r}=-A-CQ_\infty^{-1}$ and nonlinear drift $CD\log\rho^{ss}-F$.
	
	\begin{corollary}
		$\rho^{ss}$ satisfies
		$$\langle AQ_\infty D\log\rho^{ss}, D\log\rho^{ss}\rangle+\dv_\mu (AQ_\infty D\log\rho^{ss})+\langle F(x), D\log\rho^{ss}\rangle+\dv_\mu F=0.$$
	\end{corollary}
	
	\begin{remark}
		In the finite-dimensional case, the adjoint of $K_2$ in $L^2(H;\nu)$ is obtained from $K_2$ by replacing $Ax+F(x)$ with $CD\log\mu(x)-Ax+CD\log\rho^{ss}-F(x)=(-C\Sigma^{-1}-A)x+CD\log\rho^{ss}-F(x)$. Hence we can write $K^*_2$ as
		\begin{align}
			K_2^*\varphi(x)&=\frac{1}{2}\Tr[CD^2\varphi]+\langle x,(-C\Sigma^{-1}-A)^*D\varphi\rangle+\langle CD\log\rho^{ss}-F(x), D\varphi\rangle\nonumber\\
			&=L_2^*\varphi+\langle CD\log\rho^{ss}-F(x), D\varphi\rangle.
		\end{align}
	\end{remark}

	We shall compare four path measures: the law of the original semilinear SPDE $\{X(t)\}$, the law of its time reversal  $\{X^{rev}(t)\}$, and the two Ornstein--Uhlenbeck reference laws associated with the forward and reversed linear parts. All start from initial distribution $\nu$, the invariant measure of $\{X(t)\}$ and $\{X^{rev}(t)\}$.
	
	If one insisted on using the same linear part $A$ for the time-reversed dynamics, the remaining drift would contain
	$\tilde{F}(x)=-2Ax-F(x)+CD\log\rho^{ss}(x)-CQ_\infty^{-1}x$,
	which is not a Lipschitz map on $H$ when $A$ is unbounded. We therefore organize the computation relative to two Ornstein--Uhlenbeck reference laws, corresponding respectively to the forward linear part $A$ and the reversed linear part $A^r=-A-CQ_\infty^{-1}=Q_\infty A^*Q_\infty^{-1}$.
	
	We still define the sample-path steady-state entropy production as
	$$EP^{ss}(t)=\log\frac{dP^{ss}_{[0,t]}}{dP^{rev,ss}_{[0,t]}},$$
	in which $P^{ss}_{[0,t]}$ is the law of $\{X(s):0\leqslant s\leqslant t\}$ and $P^{rev,ss}_{[0,t]}$ is the law of $\{X^{rev}(s)=X(t-s):0\leqslant s\leqslant t\}$. All start from $\nu$.
	
	Hence
	$$EP^{ss}(t)=\log\frac{dP^{ss}_{[0,t]}}{dP^*_{[0,t]}}-\log\frac{dP^{rev,*}_{[0,t]}}{dP^*_{[0,t]}}-\log\frac{dP^{rev,ss}_{[0,t]}}{dP^{rev,*}_{[0,t]}},$$
	in which $P^*_{[0,t]}$ is the law of OU process with linear operator $A$ and $P^{rev,*}_{[0,t]}$ is the law of OU process with linear operator $-A-CQ_\infty^{-1}$. The first and third terms can be calculated using Girsanov formula, and we need a lemma for the second term.
	
	To treat the second term which emerges in the nonsymmetric case, we recall Hypothesis~\ref{hypo_4}(ii)--(iii), namely that the irreversible part $R=A^r-A$ is bounded together with $C^{-1/2}R$ and $R(H)\subset B(U)$.
	
	The two Ornstein--Uhlenbeck laws are compared in the following lemma.
	\begin{lemma}[Radon--Nikodym derivative between the two Ornstein--Uhlenbeck laws]
		\label{lem:two-ou}
		Assume Hypotheses~\ref{hypo_1}--\ref{hypo_3} and Hypothesis~\ref{hypo_4}(ii)--(iii). Then, for every $t>0$, the law $P^{rev,*}_{[0,t]}$ of the reversed Ornstein--Uhlenbeck process (linear part $A^r=-A-CQ_\infty^{-1}$) and the law $P^{*}_{[0,t]}$ of the forward one (linear part $A$), both started from $\mu=N_{Q_\infty}$, are mutually absolutely continuous on $C([0,t];H)$, and
		\begin{equation}\label{two-ou-density}
			\log\frac{dP^{rev,*}_{[0,t]}}{dP^*_{[0,t]}}(X)=\int_0^t\langle B^{-1}RX,\mathrm dW\rangle_U-\frac12\int_0^t\|B^{-1}RX\|_U^2\,\mathrm ds,
		\end{equation}
		the stochastic integral being taken against the forward driving noise $W$.
	\end{lemma}
	
	\begin{proof}
		By the Lyapunov equation $AQ_\infty+Q_\infty A^*=-C$ one has $A^r=Q_\infty A^*Q_\infty^{-1}=-A-CQ_\infty^{-1}$, and a direct computation shows that $A^r$ satisfies the \emph{same} Lyapunov equation $A^rQ_\infty+Q_\infty(A^r)^*=-C$; hence $\mu=N_{Q_\infty}$ is invariant for the time-reversed process as well, and the two reference laws share the stationary initial law $\mu$. Under Hypothesis~\ref{hypo_4}(ii) the operator $A^r=A+R$ is a bounded perturbation of $A$, so it again generates a $C_0$-semigroup, and the time-reversed process is the unique mild solution of $\mathrm dZ^r=(A+R)Z^r\,\mathrm dt+B\,\mathrm dW$. Its drift $\tilde F(z)=Rz$ is globally Lipschitz of linear growth $|Rz|\leqslant\|R\|\,|z|$, so existence and uniqueness in law follow from [\cite{Prato1992} Theorem~7.5]. Note that the unboundedness of $A$ is never used: $\tilde F=R\,\cdot$ is treated as a bounded Lipschitz \emph{map}.
		
		Set $\psi(s):=B^{-1}RX(s)=C^{-1/2}RX(s)$, with $X$ the forward Ornstein--Uhlenbeck process. By Hypothesis~\ref{hypo_4}(iii), $\psi$ is $U$-valued and satisfies the linear-growth bound \eqref{eq:psi-growth} with $a=0$ and $b=\|C^{-1/2}R\|$. Lemma~\ref{lem:girsanov-growth}, applied with the forward Ornstein--Uhlenbeck process as the base equation and time-independent perturbing drift $G(x)=Rx$, gives $P^{rev,*}_{[0,t]}\sim P^*_{[0,t]}$ with density the stochastic exponential of $\psi$, noting that under $P^*$ the innovation $W^*$ coincides with the forward driving noise $W$.
		
		Finally, since $R$ and $C^{-1/2}R$ are bounded, $\psi(s)=C^{-1/2}RX(s)$ is a continuous adapted functional of the path $X$, so the It\^o integral $\int_0^t\langle\psi,\mathrm dW\rangle_U$ is the $L^2$-limit of $\sigma(X)$-measurable Riemann sums, and the density reduces to the stochastic exponential \eqref{two-ou-density}.
	\end{proof}

	\begin{remark}[Scope of Hypothesis~\ref{hypo_4}(ii)]
		\label{rmk:Hstar-scope}
		Writing $A_a:=A+\tfrac12CQ_\infty^{-1}$ for the antisymmetric (irreversible) part of the generator in $L^2(H,\mu)$, one has $R=-2A_a$, so Hypothesis~\ref{hypo_4}(ii) is precisely the requirement that the irreversible part be a \emph{bounded} (zeroth-order) operator; the forward and reversed dynamics then share the same unbounded principal part. This covers bounded non-gradient drift contributions---for instance the irreversible part of the two-temperature coupled field of Section~\ref{sec:applications}, for which $\|R\|_{L(H)}\leqslant2\kappa\,|T_1-T_2|/\min(T_1,T_2)<\infty$ and $C^{-1/2}R$ is bounded because $B$ is boundedly invertible---but it excludes genuine spatial transport. For advection $A=\nu\Delta+a\partial_x$ one finds $R=-2a\partial_x$, a first-order \emph{unbounded} operator: Hypothesis~\ref{hypo_4}(ii) fails, the per-mode contribution $\Tr(R_kQ_{\infty,k}R_k^*)\equiv 2a^2/\nu$ does not decay, the two path laws become mutually singular, and the mean entropy production is ultraviolet-divergent and calls for renormalization (cf.\ Remark~\ref{rmk:EPfinite}).
	\end{remark}
	
	\begin{theorem}[Steady-state entropy production, nonsymmetric case]
		Assume Hypotheses~\ref{hypo_1}--\ref{hypo_5} and \textbf{(H-rev)}. Then both $F$ and the time-reversed drift are admissible in the sense of Proposition~\ref{prop:girsanov}. The steady-state entropy production can be expressed as
		\begin{align}
			EP^{ss}(t)&=\int_0^t \langle B^{-1}[2F-CD\log\rho^{ss}-RX],dW\rangle_U \nonumber\\
			&+\frac{1}{2}\int_0^t\| B^{-1}[2F-CD\log\rho^{ss}-RX]\|^2_Uds\nonumber
		\end{align}
		\begin{align}
			epr=\frac{1}{2}\int_H\| B^{-1}[2F-CD\log\rho^{ss}-Rx]\|^2_U d\nu.\nonumber
		\end{align}
		The fluctuation theorem $\mathbb{E} e^{-EP^{ss}(t)}=1$ holds, for all $t>0$.
	\end{theorem}
	
	\begin{proof}
		All four path laws in the above decomposition are understood to start
		from $\nu$. Although Lemma~\ref{lem:two-ou} is stated for the common
		initial law $\mu$, its Radon--Nikodym density is unchanged when both
		Ornstein--Uhlenbeck reference processes are started from
		$\nu=\rho^{ss}\mu$. Indeed, on the canonical path space,
		\[
		\mathrm dP^{*,\nu}_{[0,t]}
		=
		\rho^{ss}(X(0))\,\mathrm dP^{*,\mu}_{[0,t]},
		\qquad
		\mathrm dP^{rev,*,\nu}_{[0,t]}
		=
		\rho^{ss}(X(0))\,\mathrm dP^{rev,*,\mu}_{[0,t]}.
		\]
		Hence the common factor $\rho^{ss}(X(0))$ cancels in their
		Radon--Nikodym ratio, and \eqref{two-ou-density} remains valid for the
		common initial law $\nu$. We suppress the initial-law superscript below.
		
		\begin{align}
			\log\frac{dP^{rev,*}_{[0,t]}}{dP^*_{[0,t]}}(X(\cdot))&=\int_0^t \langle B^{-1}RX,B^{-1}[dX-AXds]\rangle_U \nonumber\\
			&-\frac{1}{2}\int_0^t\| B^{-1}RX\|_U^2ds.
		\end{align}
		
		Then combined with the expressions of $\log\frac{dP^{ss}_{[0,t]}}{dP^*_{[0,t]}}(X(\cdot))$ and $\log\frac{dP^{rev,ss}_{[0,t]}}{dP^{rev,*}_{[0,t]}}(X(\cdot))$ derived in the last subsection, one has 
		\begin{align}
			EP^{ss}(t)&=\log\frac{dP^{ss}_{[0,t]}}{dP^*_{[0,t]}}-\log\frac{dP^{rev,*}_{[0,t]}}{dP^*_{[0,t]}}-\log\frac{dP^{rev,ss}_{[0,t]}}{dP^{rev,*}_{[0,t]}}\nonumber\\
			&=\int_0^t \langle B^{-1}F(X),B^{-1}[dX-AXds]\rangle_U-\frac{1}{2}\int_0^t\| B^{-1}F(X)\|_U^2ds\nonumber\\
			&-\int_0^t \langle B^{-1}RX,B^{-1}[dX-AXds]\rangle_U \nonumber\\
			&+\frac{1}{2}\int_0^t\| B^{-1}RX\|_U^2ds\nonumber\\
			&-\int_0^t \langle B^{-1}[CD\log\rho^{ss}-F](X),B^{-1}[dX-(-A-CQ_\infty^{-1})Xds]\rangle_U \nonumber\\
			&+\frac{1}{2}\int_0^t\| B^{-1}[CD\log\rho^{ss}-F](X)\|_U^2ds.
		\end{align}
		
		Using $dX=[AX+F(X)]dt+BdW$, we obtain
		\begin{align}
			EP^{ss}(t)&=\int_0^t \langle B^{-1}F(X),B^{-1}F(X)ds+dW\rangle_U-\frac{1}{2}\int_0^t\| B^{-1}F(X)\|_U^2ds\nonumber\\
			&-\int_0^t \langle B^{-1}RX,B^{-1}F(X)ds+dW\rangle_U \nonumber\\
			&+\frac{1}{2}\int_0^t\| B^{-1}RX\|_U^2ds\nonumber\\
			&-\int_0^t \langle B^{-1}[CD\log\rho^{ss}-F](X),B^{-1}[F(X)-RX]ds+dW\rangle_U \nonumber\\
			&+\frac{1}{2}\int_0^t\| B^{-1}[CD\log\rho^{ss}-F](X)\|_U^2ds\nonumber\\
			&=\int_0^t \langle B^{-1}[2F-CD\log\rho^{ss}-RX],dW\rangle_U \nonumber\\
			&+\frac{1}{2}\int_0^t\| B^{-1}[2F-CD\log\rho^{ss}-RX]\|^2_Uds
		\end{align}
	\end{proof}
	
	\begin{theorem}[Characterizations of reversibility, nonsymmetric case]
		Under Hypotheses~\ref{hypo_1}--\ref{hypo_5}, the following statements are equivalent.
		\begin{enumerate}[label=\rm(\arabic*)]
			\item $X(t)$ is time reversible.
			\item The semigroup $P_t$ is symmetric in $L^2(H;\nu)$.
			\item The generator $K_2$ is symmetric.
			\item $epr=0$.
			\item $EP^{ss}(t)=0$ almost surely.
			\item $2F(x)-CD\log\rho^{ss}(x)=Rx$ for $\nu$-a.e.\ $x\in H$.
			\item There exists a measurable potential $U:H\to\mathbb R$, with $e^{-U}\in L^2(H,\mu)$ and $C^{1/2}DU\in L^2(H,\mu)$, such that $2F=-CDU+Rx$ in the Gaussian Sobolev sense.
		\end{enumerate}
	\end{theorem}
	
	\begin{proof}
		The implications $(1)\Leftrightarrow(2)\Leftrightarrow(3)$ are classical \cite{JQQ}. By Theorem~\ref{adjoint_K_2}, $K_2$ is symmetric if and only if the drifts of $K_2$ and $K_2^*$ coincide, that is, if and only if $2F-CD\log\rho^{ss}=Rx$. This gives $(3)\Leftrightarrow(6)$. The rate $epr$ of the preceding theorem is the weighted $L^2(\nu)$-norm of $B^{-1}[2F-CD\log\rho^{ss}-Rx]$, so $(4)\Leftrightarrow(6)$. Finally, under $(6)$ the integrand of $EP^{ss}$ vanishes identically, whence $EP^{ss}(t)=0$ almost surely, while conversely $\mathbb E\,EP^{ss}(t)=t\,epr$. This gives $(4)\Leftrightarrow(5)$.
		
		$(6)\Rightarrow(7)$. Take $U=-\log\rho^{ss}$. Then $e^{-U}=\rho^{ss}$ has a square-integrable $C$-gradient by Theorem~\ref{Invariant measure & density} together with Hypothesis~\ref{hypo_5}, and $(6)$ reads $2F=CD\log\rho^{ss}+Rx=-CDU+Rx$.
		
		$(7)\Rightarrow(6)$. Let $g=e^{-U}$. We verify that $g$ satisfies the stationary identity
		\begin{equation}
			\langle AQ_\infty w,w\rangle+\dv_\mu(AQ_\infty w)+\langle F,w\rangle+\dv_\mu F=0,
			\qquad w:=D\log g,
			\label{eq:stat-identity}
		\end{equation}
		in the same sense as $\rho^{ss}$ does (the corollary following Theorem~\ref{adjoint_K_2}), the computation being performed for smooth $U$ and extended by density, all terms in any case being read through the standing interpretation of Section~2. Running the integration-by-parts computation of Theorem~\ref{adjoint_K_2} with $g$ in place of $\rho^{ss}$, this time in $L^2(H,\mu)$, produces a potential term equal, up to sign, to the left-hand side of \eqref{eq:stat-identity}, so that \eqref{eq:stat-identity} is equivalent to the weak invariance relation $\int_H(K_2\varphi)\,g\,\mathrm d\mu=0$ for every $\varphi\in\mathscr E_A(H)$. Three facts enter, all valid under Hypothesis~\ref{hypo_4}(ii) with Remark~\ref{rmk:AQauto}.
		
		First, by the Lyapunov equation,
		\[
		AQ_\infty+\tfrac12C=-\tfrac12RQ_\infty
		\quad\text{with }RQ_\infty\text{ skew-adjoint},
		\]
		so that $\langle(AQ_\infty+\tfrac12C)y,y\rangle=0$ for every $y\in H$. Second, $\dv_\mu(Rx)=0$. Indeed $\mu$ is invariant for both Ornstein--Uhlenbeck processes with drifts $A$ and $A^{r}=A+R$ and common diffusion $C$, so for every $\varphi\in\mathscr E_A(H)$,
		\[
		\int_H\langle Rx,D\varphi\rangle\,\mathrm d\mu
		=\int_H\langle A^{r}x,D\varphi\rangle\,\mathrm d\mu-\int_H\langle Ax,D\varphi\rangle\,\mathrm d\mu=0,
		\]
		since each integral equals $-\tfrac12\int_H\Tr[CD^2\varphi]\,\mathrm d\mu$ through the corresponding generator. Third, $\dv_\mu(RQ_\infty DU)=\langle Rx,DU\rangle$. For $\varphi\in\mathscr E_A(H)$, Gaussian integration by parts applied to the field $R^{*}DU$ gives
		\[
		\int_H\varphi\,\langle Rx,DU\rangle\,\mathrm d\mu
		=-\int_H\langle RQ_\infty DU,D\varphi\rangle\,\mathrm d\mu,
		\]
		the trace term $\Tr(Q_\infty R^{*}D^{2}U)=-\Tr(RQ_\infty D^{2}U)$ vanishing because $RQ_\infty$ is skew-adjoint and $D^{2}U$ symmetric.
		
		Now substitute $w=-DU$ and $F=-\tfrac12CDU+\tfrac12Rx$ into \eqref{eq:stat-identity}. Collecting terms,
		\begin{align}
			&\langle AQ_\infty DU,DU\rangle-\dv_\mu(AQ_\infty DU)
			+\tfrac12\langle CDU,DU\rangle-\tfrac12\langle Rx,DU\rangle
			-\tfrac12\dv_\mu(CDU)+\tfrac12\dv_\mu(Rx)\nonumber\\
			&=\langle(AQ_\infty+\tfrac12C)DU,DU\rangle
			-\dv_\mu\big((AQ_\infty+\tfrac12C)DU\big)
			-\tfrac12\langle Rx,DU\rangle\nonumber\\
			&=\tfrac12\dv_\mu(RQ_\infty DU)-\tfrac12\langle Rx,DU\rangle=0,
		\end{align}
		by the three facts above. Hence $Z_U^{-1}g\,\mu$, with $Z_U=\int_Hg\,\mathrm d\mu$, is an invariant probability measure. The uniqueness of the invariant measure yields $\rho^{ss}=Z_U^{-1}g$, so that $CD\log\rho^{ss}=-CDU$, and then $2F-CD\log\rho^{ss}=Rx$ by $(7)$. This is $(6)$.
	\end{proof}
	
	\section{Transient entropy production and its thermodynamic decomposition}
	\label{sec:transient}
	
	\subsection{When the associated Ornstein--Uhlenbeck process is symmetric}
	
	We now turn to the transient regime, allowing initial
	laws that are absolutely continuous with respect to $\mu$.
	The analysis in this section requires the $L^2(\mu)$
	It\^o formulas established in
	Theorems~\ref{thm:dynkin-ito} and
	\ref{thm:parabolic-ito}
	(Appendix~\ref{Appendix B}).
	To ensure their applicability, we impose the following
	additional assumption.
	
	\begin{itemize}
		\item[\textbf{(H-form)}]
		The Dirichlet form associated with the Kolmogorov
		operator of the SPDE \eqref{semilinearSPDE}
		on $L^2(H,\mu)$
		(see \eqref{eq:dirichlet-form} in
		Appendix~\ref{Appendix B})
		is closable and quasi-regular, and its closure
		is properly associated with the mild solution.
	\end{itemize}
	
	This assumption is standard in the theory of
	generalized Dirichlet forms
	\cite{MaRoeckner1992,Stannat1999}.
	It is needed only for the transient analysis
	in this section; the steady-state results of
	Section~3 rely on Girsanov transformations
	and do not require (H-form).
	
	We begin by establishing the existence and
	regularity of the time-dependent density
	\[
	\rho_t=\frac{d\nu_t}{d\mu},
	\]
	where $\nu_t$ denotes the law of $X(t)$ and
	$\mu$ is the invariant Gaussian measure of
	the associated Ornstein--Uhlenbeck process.
		
	The construction of $\rho_t$ depends on whether Hypothesis~\ref{hypo_3}\ref{hypo_3.a} or Hypothesis~\ref{hypo_3}\ref{hypo_3.b} is assumed.
	
	Under Hypothesis~\ref{hypo_3}\ref{hypo_3.a}, the strong Feller property is unavailable in general, so we construct $\rho_t$ through the Kolmogorov forward equation. We assume
	\[
	\nu_0\ll\mu,
	\qquad
	\rho_0=\frac{d\nu_0}{d\mu}\in C_b^2(H),
	\]
	and define $\rho_t$ as the solution of
	\[
	\partial_t\rho_t
	=
	K_{2,\mu}^*\rho_t,
	\qquad
	\rho_t|_{t=0}=\rho_0.
	\]
	Appendix~\ref{Appendix C} establishes the existence and uniqueness of a mild solution in \(L^2(H,\mu)\). If, in addition, the initial density satisfies \(\rho_0\in D(K_{2,\mu}^*)\), this solution is strict.
	Writing
	\[
	\nu_t(dx)=\rho_t(x)\mu(dx),
	\]
	one has
	\[
	\int_H P_tf(x)\rho_0(x)\,d\mu
	=
	\int_H f(x)\rho_t(x)\,d\mu,
	\qquad
	f\in B_b(H),
	\]
	showing that $\nu_t$ is precisely the law of $X(t)$.
	
	Under Hypothesis~\ref{hypo_3}\ref{hypo_3.b}, the same density is obtained directly from the transition probabilities. Indeed, $P_t$ is strong Feller and irreducible by \cite[Theorems~3.11 and~3.13]{Prato2004}, hence regular by \cite[Proposition~4.1.1 and Theorem~4.2.1]{Prato_Zabczyk_1996}. Therefore, for every initial distribution $\nu_0$,
	\[
	\nu_t:=P_t(\nu_0,\cdot)
	\]
	satisfies
	\[
	\nu_t\ll\nu\ll\mu,
	\]
	where $\nu$ is the unique invariant measure. Consequently,
	\[
	\rho_t=\frac{d\nu_t}{d\mu}
	\]
	exists for every $t\geqslant 0$. Since $P_{t,\mu}^*$ is the adjoint semigroup of $P_t$ in $L^2(H,\mu)$, the density $\rho_t$ also satisfies the Kolmogorov forward equation
	\[
	\partial_t\rho_t
	=
	K_{2,\mu}^*\rho_t,
	\qquad
	\rho_t|_{t=0}=\rho_0.
	\]
	
	Hence, under either Hypothesis~\ref{hypo_3}\ref{hypo_3.a} or Hypothesis~\ref{hypo_3}\ref{hypo_3.b}, the distribution of the solution admits the representation
	\[
	\nu_t(dx)=\rho_t(x)\mu(dx),
	\]
	where $\rho_t$ is the unique solution of the Kolmogorov forward equation.
	
	The following proposition summarizes their regularity properties.
	\begin{proposition}[Existence and regularity of the transient density]
		\label{prop:transient-reg}
		
		Assume Hypotheses~\ref{hypo_1}--\ref{hypo_3} and $AQ_\infty\in L(H)$, and that
		\[
		F\in C_b^2(H;H)
		\]
		has finite divergence with respect to $\mu$. Suppose furthermore that
		\[
		\nu_0\ll\mu,\qquad
		\rho_0=\frac{d\nu_0}{d\mu}\in C_b^2(H),
		\qquad
		\rho_0>0.
		\]
		
		Then the transient density $\rho_t$ satisfies
		
		\begin{enumerate}[label=(\roman*)]
			
			\item
			$\rho_t>0$ $\mu$-almost everywhere;
			
			\item
			$\rho_t\in D(K_{2,\mu}^*)$ when $L_2$ is symmetric or $C=\mathrm{Id}$, and is the mild solution of Theorem~\ref{thm:forward-W22}(ii) in general;
			
			\item
			if moreover $\rho_0\in D(K_{2,\mu}^*)$, then $\rho_t$ is the strict solution of Theorem~\ref{thm:forward-W22}(ii) and $\partial_t\rho_t=K_{2,\mu}^*\rho_t\in L^2(H,\mu)$ for every $t\geqslant0$.
			
		\end{enumerate}
		
	\end{proposition}

	\begin{proof}
		(i) is the positivity preservation of $P_{t,\mu}^*$ (Theorem~\ref{thm:forward-W22}(ii)). (ii) is the smoothing statement of Theorem~\ref{thm:forward-W22}(ii). (iii) follows from the strict-solution statement of Theorem~\ref{thm:forward-W22}(ii) under the additional assumption $\rho_0\in D(K_{2,\mu}^*)$.
	\end{proof}
	
	Before introducing the trajectory-wise thermodynamic quantities, we assume that the transient density enjoys the same regularity and admissibility properties as the stationary density in Theorem~\ref{Invariant measure & density}, uniformly in time.
	
	\begin{hypo}[Transient admissibility]
		\label{hypo_6}
		Assume the hypotheses of Proposition~\ref{prop:transient-reg}. Fix $T>0$.
		We further assume that
		\[
		\log\rho\in L^\infty\bigl(0,T;D(K_{2,\mu})\bigr),
		\qquad
		\partial_t\log\rho_t
		=
		\frac{K_{2,\mu}^*\rho_t}{\rho_t}
		\in
		L^2\bigl((0,T)\times H,\mathrm dt\otimes\mu\bigr),
		\]
		where $D(K_{2,\mu})$ is equipped with its graph norm, and that the
		thermodynamic force has at most linear growth,
		\[
		\operatorname*{ess\,sup}_{0<u<T}\,
		\big|C^{-1/2}\big(2F-CD\log\rho_u\big)(x)\big|
		\leqslant a+b|x|,
		\qquad \mu\text{-a.e. }x\in H.
		\]
	\end{hypo}
	
	\noindent This hypothesis is independent of Proposition~\ref{prop:transient-reg} and is imposed as a standing regularity assumption on the density curve. It is deliberately strong, and we do not derive it from more primitive conditions on the coefficients, since even in the linear Gaussian case the quadratic logarithmic densities meet the same Cameron--Martin obstruction discussed in Appendix~\ref{Appendix B}.
	
	We define the entropy of a state $x$ at time $t$, relative to the Gaussian reference measure $\mu$, by
	\[
	s_t(x):=-\log\rho_t(x).
	\]
	
	\begin{remark}[Why entropy is defined relative to $\mu$]
		\label{rmk:absolute-entropy}
		
		In finite dimensions, writing the probability density as
		\[
		p_t(x)=\rho_t(x)\mu(x),
		\]
		the stochastic entropy with respect to the Lebesgue measure is
		\[
		-\log p_t(x)
		=
		-\log\rho_t(x)-\log\mu(x).
		\]
		Thus defining the entropy through $\rho_t$ amounts to removing the Gaussian contribution coming from the reference measure $\mu$.
		
		In infinite dimensions this Gaussian contribution,
		formally
		$\frac12\langle x,Q_\infty^{-1}x\rangle$,
		is infinite for $\mu$-almost every $x$, so the analogue of the finite-dimensional entropy cannot be defined as a finite random variable.
		Nevertheless, this does not affect the thermodynamic theory developed below.
		Indeed, the Gaussian contribution cancels identically in the total entropy production, so all thermodynamic balance laws remain finite under our assumptions.
		In finite dimensions one has $s_t=s_t^{\mathrm{Leb}}+\log\mu$, so the Gaussian reference contributes $\Delta\log\mu$ to the balance: it cancels in the total entropy production but shifts the heat term, which is therefore a reference-relative heat functional rather than the reservoir heat of a detailed-balance system.
	\end{remark}
	
	\begin{theorem}\label{thm_ep}
		Under Hypotheses~\ref{hypo_1}--\ref{hypo_3}$,$ Hypothesis~\ref{hypo_6} imposed on $(0,2T)$, and \textbf{(H-form)}, the transient entropy production along a sample path $X(t)$ admits the decomposition
		$$EP(t)=\Delta s(t)+HD_{\mathrm{rel}}(t),$$
		in which $\Delta s(t)=s_t(X(t))-s_0(X(0))$ is the change of the state entropy,
		\begin{align}
			EP(t)&=\int_0^t\left\langle B^{-1}\left[2F-CD\log\rho_s\right](X(s)), B^{-1}dX(s)\right\rangle_U-\frac{\partial \log\rho_s}{\partial s}(X(s))ds\nonumber\\
			&+\frac{1}{2}\int_0^t \Tr\left(D\left[2F-CD\log\rho_s\right](X(s))\right)ds,
		\end{align} and
		\begin{align}
			HD_{\mathrm{rel}}(t)&=\int_0^t 2\left\langle B^{-1}F(X(s)), B^{-1}dX(s)\right\rangle_U+\dv F(X(s))\,ds\nonumber\\
			&=2\int_0^t\langle B^{-1}F(X),B^{-1}\circ dX(s)\rangle_U
		\end{align}
		is the heat dissipation.
		
		Moreover,
		\begin{equation}
			\mathbb{E} EP(t)=\frac{1}{2}\int_0^t\int_H\langle B^{-1}[2F-CD\log\rho_s](x),B^{-1}[2F-CD\log\rho_s](x)\rangle_U d\nu_s ds\geqslant 0,
		\end{equation}
		and $\mathbb{E} EP(t)=0$ if and only if $2F=CD\log\rho_s$ for almost every $s\leqslant t$, which implies $\nu_s=\nu$ and the SPDE is time-reversible.
	\end{theorem}
	
	\begin{proof}
		
		\noindent\emph{Step 1: Pathwise decomposition.}
		
		By Hypothesis~\ref{hypo_6}, the parabolic
		Dynkin--Fukushima formula
		(Appendix~\ref{Appendix B}, Theorem~\ref{thm:parabolic-ito}) applies to
		$s_t=-\log\rho_t$, $\mathbb P_{\nu_0}$-a.s., and gives
		\begin{align}
			ds_t(X(t))&=-\frac{\partial \log\rho_t}{\partial t}(X(t))dt\nonumber\\
			&-\left\langle D \log\rho_t(X(t)), dX(t)\right\rangle-\frac{1}{2}\Tr\left(D^2 \log\rho_t(X(t))C\right)dt.
		\end{align}
		
		By the definition of $C$ and $B^*$, we have for any $x$ and $y\in H$,
		$$\langle x, y\rangle=\langle C^{-1}Cx,y\rangle=\langle (B^{-1})^*B^{-1}Cx,y\rangle=\langle B^{-1}Cx, B^{-1}y\rangle_U.$$
		
		Then we decompose $-D\log\rho_t$ as  $C^{-1}\left[2F-CD\log\rho_t\right]-2C^{-1}F$, so that
		\begin{equation}
			ds_t(X(t))=ep(t)-hd_{\mathrm{rel}}(t),
		\end{equation}
		in which
		\begin{align}
			ep(t)&=\left\langle B^{-1}\left[2F-CD\log\rho_t\right](X(t)), B^{-1}dX(t)\right\rangle_U-\frac{\partial \log\rho_t}{\partial t}(X(t))dt\nonumber\\
			&+\frac{1}{2}\Tr\left(DC^{-1}\left[2F-CD\log\rho_t\right](X(t))C\right)dt\nonumber\\
			&=\left\langle B^{-1}\left[2F-CD\log\rho_t\right](X(t)), B^{-1}dX(t)\right\rangle_U-\frac{\partial \log\rho_t}{\partial t}(X(t))dt\nonumber\\
			&+\frac{1}{2}\Tr\left(D\left[2F-CD\log\rho_t\right](X(t))\right)dt,
		\end{align}

		and
		\begin{align}
			hd_{\mathrm{rel}}(t)&=2\left\langle B^{-1}F(X(t)), B^{-1}dX(t)\right\rangle_U+\Tr\left(DC^{-1}F(X(t))C\right)dt\nonumber\\
			&=2\left\langle B^{-1}F(X(t)), B^{-1}dX(t)\right\rangle_U+\dv F(X(t))\,dt.
		\end{align}
		Here $ep(t)$ and $hd_{\mathrm{rel}}(t)$ denote the infinitesimal increments of $EP(t)$ and $HD_{\mathrm{rel}}(t)$.
		
		\noindent\emph{Step 2: Mean entropy production.}
		
		$$EP(t)=\int_0^t ep(s)ds,$$
		hence
		\begin{align}
			\mathbb{E} EP(t)&=\mathbb{E}\int_0^t\left\langle B^{-1}\left[2F-CD\log\rho_s\right](X(s)), B^{-1}(AX(s)+F(X(s)))\right\rangle_U ds\nonumber\\
			&+\frac{1}{2}\mathbb{E}\int_0^t\Tr\left(D\left[2F-CD\log\rho_s\right](X(s))\right)ds,
		\end{align}
		with the pairing $\langle AX(s),\cdot\rangle$ read as the generator
		action of the standing interpretation in Section~2. Since $X(s)\sim\nu_s=\rho_s\mu$ and
		\begin{align}
			\mathbb{E}\int_0^t\frac{\partial \log\rho_s}{\partial s}(X(s))ds&=\int_0^t\int_H \frac{\partial \log\rho_s}{\partial s}(x)\rho_sd\mu ds\nonumber\\
			&=\int_0^t\left[\int_H \frac{\partial \rho_s}{\partial s}(x)d\mu \right]ds=\int_0^t \left[\frac{\partial }{\partial s}\int_H \rho_sd\mu\right]ds\nonumber\\
			&=\int_0^t\left(\frac{\partial }{\partial s} 1\right)ds=0.
		\end{align}
		
		Then
		\begin{align}
			\mathbb{E} EP(t)&=\int_0^t\int_H\left\langle B^{-1}\left[2F-CD\log\rho_s\right](x), B^{-1}(Ax+F(x))\right\rangle_U \rho_s d\mu ds\nonumber\\
			&+\frac{1}{2}\int_0^t\int_H\Tr\left(D\left[2F-CD\log\rho_s\right](x)\right)\rho_sd\mu ds,
		\end{align}

		Furthermore, since $\frac{\partial \rho_t}{\partial t}=K_{2,\mu}^*\rho_t=\frac{1}{2}\Tr[CD^2\rho_t]+\langle x,A^* D\rho_t\rangle-[\langle F,D\rho_t\rangle+\dv_\mu F\rho_t]$,
		then
		\begin{align}
			&\Tr\left(D\left[2F-CD\log\rho_t\right]\right)\nonumber\\
			&=\Tr\left(\left[\frac{D(2F\rho_t-CD\rho_t)}{\rho_t}-\frac{2F\rho_t-CD\rho_t}{\rho_t}D\log\rho_t\right]\right)\nonumber\\
			&=-2\left\langle C^{-1}\left[2F-CD\log\rho_t\right](x), [Ax+\frac{1}{2}CD\log\rho_t]\right\rangle-2\frac{\partial \log\rho_t}{\partial t}\nonumber\\
			&=-2\left\langle B^{-1}\left[2F-CD\log\rho_t\right](x), B^{-1}[Ax+\frac{1}{2}CD\log\rho_t]\right\rangle_U-2\frac{\partial \log\rho_t}{\partial t}.
		\end{align}
		
		Therefore,
		\begin{align}
			\frac{d}{dt}\mathbb{E} EP(t)&=\int_H\left\langle B^{-1}\left[2F-CD\log\rho_t\right](x), B^{-1}(Ax+F(x))\right\rangle_U \rho_t d\mu\nonumber\\
			&-\frac{1}{2}\int_H\left\langle B^{-1}\left[2F-CD\log\rho_t\right], B^{-1}[CD\log\rho_t+2Ax]\right\rangle_U\rho_td\mu\nonumber\\
			&=\frac{1}{2}\int_H\left\langle B^{-1}\left[2F-CD\log\rho_t\right], B^{-1}(2F-CD\log\rho_t)\right\rangle_U (x) \rho_t d\mu.
		\end{align}
	\end{proof}
	
	\begin{remark}[Robust It\^o form of the heat dissipation]
		\label{rmk:ito-definition}
		The display for $hd_{\mathrm{rel}}$ in Step 1 of the proof of Theorem~\ref{thm_ep} is a formal Stratonovich line integral: neither the term involving $AX_t$ nor the trace needs to exist separately. An equivalent, robust definition starts from the It\^o sum
		\[
		HD_{\mathrm{rel}}(t)
		=2\int_0^t\big\langle B^{-1}F,\mathrm dW\big\rangle_U
		+\int_0^t\big\{2\|B^{-1}F\|_U^2+\dv_\mu F\big\}(X(s))\,\mathrm ds,
		\]
		whose terms are individually well defined under the standing assumptions: the stochastic integral by the admissibility of $F$, the square term by the force bounds, and the divergence by \eqref{finite divergence}. The parabolic Dynkin--Fukushima formula (Theorem~\ref{thm:parabolic-ito}, hence \textbf{(H-form)}) shows that the two definitions coincide along the trajectory. In the nonsymmetric case the same remark applies with $2F$ replaced by $2F-RX$ in the stochastic integral and the drift term $\tfrac12\|B^{-1}(2F-RX)\|_U^2+\dv_\mu F$, and likewise to $Q_{hk}$ with the stationary force in place of $2F$. These conversions rest on the identity $A+A^{r}=-CQ_\infty^{-1}$, the vanishing pairing $\langle Q_\infty^{-1}x,Rx\rangle=0$, and the partial-trace convention for $R$ introduced in Section~\ref{sec:ts-nonsym}.
	\end{remark}
	
	\begin{theorem}[Path-space representation and fluctuation theorem]\label{thm:ep-rn}
		Under Hypotheses~\ref{hypo_1}--\ref{hypo_3}$,$ Hypothesis~\ref{hypo_6} imposed on $(0,2T)$, \textbf{(H-form)}, and \textbf{(H-rev)}, let
		$P_{[0,t]}^{\nu_0}$ denote the law of the forward trajectory starting
		from $\nu_0$, let $P_{[0,t]}^{\nu_t}$ denote the law of the same SPDE
		started from $\nu_t$, the terminal law of the forward experiment,
		i.e.\ the backward experiment in the sense of \cite{Cro2,Seifert05},
		and let $P_{[0,t]}^{rev,\nu_t}:=(\Theta_t)_\#P_{[0,t]}^{\nu_t}$ be
		the law of the time-reversed trajectories. Then
		$$EP(t)=\log\frac{dP_{[0,t]}^{\nu_0}}{dP_{[0,t]}^{rev,\nu_t}},$$
		and consequently $\mathbb{E} e^{-EP(t)}=1$ for any $t>0$.
	\end{theorem}
	
	\begin{proof}
		Denote the distribution of $P_{[0,t]}^{\nu_t}$ at time $s$ as $\tilde{\nu}_s=\tilde{\rho}_s\mu$, so we have
		$$\frac{\partial \tilde{\rho}_s}{\partial s}=K_{2,\mu}^*\tilde{\rho}_s=\frac{1}{2}\Tr[CD^2\tilde{\rho}_s]+\langle x,A^* D\tilde{\rho}_s\rangle-[\langle F,D\tilde{\rho}_s\rangle+\dv_\mu F\tilde{\rho}_s],$$
		By uniqueness for the forward equation, $\tilde{\rho}_s=\rho_{t+s}$, so
		that the regularity of $\log\tilde{\rho}$ and the admissibility of the
		drift $CD\log\tilde{\rho}_{t-s}-F=CD\log\rho_{2t-s}-F$ are covered by the
		assumptions on $(0,2T)$. Then as its time-reversed measure, denote the distribution of $P_{[0,t]}^{rev,\nu_t}$ at time $s$ as $\nu^{rev}_s=\rho^{rev}_s\mu=\tilde{\nu}_{t-s}$,
		\begin{align}
			&\frac{\partial \rho^{rev}_s}{\partial s}=\frac{\partial \tilde{\rho}_{t-s}}{\partial s}=-K_{2,\mu}^*\tilde{\rho}_{t-s}\nonumber\\
			&=\frac{1}{2}\Tr[CD^2\tilde{\rho}_{t-s}]+\langle x,A^* D\tilde{\rho}_{t-s}\rangle-[\langle CD\log\tilde{\rho}_{t-s}-F,D\tilde{\rho}_{t-s}\rangle\nonumber\\
			&+\dv_\mu (CD\log\tilde{\rho}_{t-s}-F)\tilde{\rho}_{t-s}].
		\end{align}
		By \textbf{(H-rev)} applied to the restarted dynamics ($\lambda=\nu_t$, $\rho^{\lambda}=\tilde{\rho}$), $P_{[0,t]}^{rev,\nu_t}$ is the law of the SPDE with the time-dependent nonlinear drift $CD\log\tilde{\rho}_{t-s}-F$, with final distribution $\nu_t$.
		
		Denote $P^{*,\nu_0}_{[0,t]}$ as the law of the corresponding OU process with initial distribution $\nu_0$, and $P^{*,\nu^{rev}_0}_{[0,t]}$ as the law of the corresponding OU process with initial distribution $\nu^{rev}_0=\tilde\nu_t$.
		
		The Girsanov formula gives
		\begin{align}
			\log\frac{dP^{\nu_0}_{[0,t]}}{dP^{*,\nu_0}_{[0,t]}}(X(\cdot))&=\int_0^t \langle B^{-1}F(X),B^{-1}[dX-AXds]\rangle_U-\frac{1}{2}\int_0^t\| B^{-1}F(X)\|_U^2ds.\nonumber
		\end{align}
		and
		\begin{align}
			\log\frac{dP^{rev,\nu_t}_{[0,t]}}{dP^{*,\nu^{rev}_0}_{[0,t]}}(X(\cdot))&=\int_0^t \langle B^{-1}[CD\log\tilde{\rho}_{t-s}-F](X),B^{-1}[dX-AXds]\rangle_U \nonumber\\
			&-\frac{1}{2}\int_0^t\| B^{-1}[CD\log\tilde{\rho}_{t-s}-F](X)\|_U^2ds.\nonumber
		\end{align}
		We know
		\begin{align}
			&\log\frac{dP^{*,\nu_0}}{dP^{*,\nu^{rev}_0}_{[0,t]}}(X(\cdot))=\log\frac{\rho_0}{\rho^{rev}_0}(X(0))=\log\frac{\rho_0}{\tilde{\rho}_t}(X(0))\nonumber\\
			&=-\log\frac{\rho_t(X(t))}{\rho_0(X(0))}+\log\frac{\tilde{\rho}_0(X(t))}{\tilde{\rho}_t(X(0))}=\int_0^t-d\log \rho_s(X(s))+d\log \tilde{\rho}_{t-s}(X(s)),\nonumber
		\end{align}
		due to $\tilde{\rho}_0=\rho_t$.
		
		Then the $L^2(\mu)$ It\^o formula (Theorem~\ref{thm:parabolic-ito}) gives, $\mathbb P_{\nu_0}$-a.s.,
		\begin{align}
			d\log \rho_s(X(s))&=\frac{\partial \log \rho_s}{\partial s}(X(s))ds\nonumber\\
			&+\left\langle D \log\rho_s(X(s)), dX(s)\right\rangle+\frac{1}{2}\Tr\left(D^2 \log\rho_s(X(s))C\right)ds,\nonumber
		\end{align}
		and
		\begin{align}
			d\log \tilde{\rho}_{t-s}(X(s))&=\frac{\partial \log \tilde{\rho}_{t-s}}{\partial s}(X(s))ds\nonumber\\
			&+\left\langle D \log \tilde{\rho}_{t-s}(X(s)), dX(s)\right\rangle+\frac{1}{2}\Tr\left(D^2 \log \tilde{\rho}_{t-s}(X(s))C\right)ds.\nonumber
		\end{align}
		Furthermore, since
		\begin{align}
			&\frac{\partial \log\tilde{\rho}_{t-s}}{\partial s}\nonumber\\
			&=-\frac{1}{2}\Tr[CD^2\tilde{\rho}_{t-s}]/\tilde{\rho}_{t-s}-\langle x,A^* D\log\tilde{\rho}_{t-s}\rangle+[\langle F,D\log\tilde{\rho}_{t-s}\rangle+\dv_\mu F],
		\end{align}
		we have
		\begin{align}
			&d\log\frac{dP_{[0,t]}^{\nu_0}}{dP_{[0,t]}^{rev,\nu_t}}\nonumber\\
			&=d\log\frac{dP^{\nu_0}_{[0,t]}}{dP^{*,\nu_0}_{[0,t]}}-d\log\frac{dP^{rev,\nu_t}_{[0,t]}}{dP^{*,\nu^{rev}_0}_{[0,t]}}+d\log\frac{dP^{*,\nu_0}_{[0,t]}}{dP^{*,\nu^{rev}_0}_{[0,t]}}\nonumber\\
			&=\langle B^{-1}F(X),B^{-1}[dX-AXds]\rangle_U-\frac{1}{2}\| B^{-1}F(X)\|_U^2ds\nonumber\\
			&-\frac{\partial \log \rho_s}{\partial s}(X(s))ds-\left\langle D \log\rho_s(X(s)), dX(s)\right\rangle-\frac{1}{2}\Tr\left(D^2 \log\rho_s(X(s))C\right)ds\nonumber\\
			&+\langle C^{-1}F,dX-AXds\rangle+\frac{1}{2}\langle C^{-1}F, F\rangle ds+\dv_\mu Fds\nonumber\\
			&=\langle 2C^{-1}F(X),dX-AXds\rangle-\frac{\partial \log \rho_s}{\partial s}(X(s))ds-\left\langle D \log\rho_s(X(s)), dX(s)\right\rangle\nonumber\\
			&-\frac{1}{2}\Tr\left(D^2 \log\rho_s(X(s))C\right)ds+\dv Fds-\langle F, Q_\infty^{-1}X\rangle ds.
		\end{align}
		Finally, integrating over $[0,t]$ and using the symmetric-case identity
		$A=-\frac{1}{2}CQ_\infty^{-1}$ gives
		\begin{align}
			&\log\frac{dP_{[0,t]}^{\nu_0}}{dP_{[0,t]}^{rev,\nu_t}}\nonumber\\
			&=\int_0^t\left\langle B^{-1}\left[2F-CD\log\rho_s\right](X(s)), B^{-1}dX(s)\right\rangle_U-\frac{\partial \log\rho_s}{\partial s}(X(s))ds\nonumber\\
			&+\frac{1}{2}\Tr\left(D\left[2F-CD\log\rho_s\right](X(s))\right)ds=EP(t).
		\end{align}
	\end{proof}
	
	\begin{remark}[Consistency with the steady-state theory]
		When $\rho_t\equiv\rho^{ss}$, the transient entropy production reduces exactly to the steady-state entropy production obtained in Section~3. Indeed, since
		$\partial_t\rho^{ss}=0$
		and $\rho^{ss}$ satisfies the stationary Fokker--Planck equation
		$K_{2,\mu}^*\rho^{ss}=0$,
		the divergence term simplifies to
		\[
		\Tr D\!\left(2F-CD\log\rho^{ss}\right)
		=
		-2\left\langle
		B^{-1}(2F-CD\log\rho^{ss}),
		B^{-1}\!\left(Ax+\frac12CD\log\rho^{ss}\right)
		\right\rangle_U,
		\]
		which yields exactly the expression of Theorem~\ref{thm:EPss-sym}.
	\end{remark}

	We define the housekeeping heat by
	\begin{align}
		Q_{hk}(t)&=\int_0^t\left\langle B^{-1}\left[2F-CD\log\rho^{ss}\right](X(s)), B^{-1}dX(s)\right\rangle_U\nonumber\\
		&+\frac{1}{2}\int_0^t\Tr\left(D\left[2F-CD\log\rho^{ss}\right](X(s))\right)ds\nonumber\\
		&=\int_0^t\langle B^{-1}[2F-CD\log\rho^{ss}](X),B^{-1}\circ dX\rangle_U,
	\end{align}
	and define the free-energy dissipation by
	\begin{align}
		F_d(t)&=EP(t)-Q_{hk}(t)\nonumber\\
		&=\int_0^t\langle B^{-1}(CD\log\rho^{ss}-CD\log\rho_s)(X),B^{-1}\circ dX\rangle_U-\int_0^t\frac{\partial \log\rho_s}{\partial s}(X(s))ds.\nonumber
	\end{align}
	
	\begin{proposition}[Positivity of housekeeping heat and free-energy dissipation]
		Under the assumptions of Theorem \ref{thm_ep} and Hypothesis~\ref{hypo_5}, the mean housekeeping heat is
		\begin{equation}
			\mathbb{E}Q_{hk}(t)=\frac{1}{2}\int_0^t\int_H\| B^{-1}[2F-CD\log\rho^{ss}](x)\|^2_U d\nu_sds.
		\end{equation}
		
		The mean free energy dissipation is
		\begin{equation}
			\mathbb{E}F_d(t)=\frac{1}{2}\int_0^t\int_H\| B^{-1}[CD\log\rho^{ss}-CD\log\rho_s](x)\|^2_U d\nu_s ds.
		\end{equation}
		$\mathbb{E}Q_{hk}(t)\geqslant 0$ and $\mathbb{E}Q_{hk}(t)=0$ if and only if $2F\equiv CD\log\rho^{ss}$. $\mathbb{E}F_d(t)\geqslant 0$ and $\mathbb{E}F_d(t)=0$ if and only if $\rho_t\equiv\rho^{ss}$.
	\end{proposition}
	
	\begin{proof}
		Since
		\begin{align}
			&\Tr\left(D\left[2F-CD\log\rho^{ss}\right]\right)\nonumber\\
			&=-2\left\langle B^{-1}\left[2F-CD\log\rho^{ss}\right](x), B^{-1}[Ax+\frac{1}{2}CD\log\rho^{ss}]\right\rangle_U,
		\end{align}
		the housekeeping heat can be expressed as
		\begin{align}
			&Q_{hk}(t)\nonumber\\
			&=\int_0^t \langle B^{-1}[2F-CD\log\rho^{ss}](X),dW\rangle_U +\int_0^t \langle B^{-1}[2F-CD\log\rho^{ss}](X), B^{-1}[AX+F(X)]\rangle ds\nonumber\\
			&-\int_0^t\left\langle B^{-1}\left[2F-CD\log\rho^{ss}\right](X(s)), B^{-1}[AX(s)+\frac{1}{2}CD\log\rho^{ss}]\right\rangle_Uds\nonumber\\
			&=\int_0^t \langle B^{-1}[2F-CD\log\rho^{ss}](X),dW\rangle_U \nonumber\\
			&+\frac{1}{2}\int_0^t\| B^{-1}[2F-CD\log\rho^{ss}](X(s))\|^2_U ds.
		\end{align}
		
		Hence the mean housekeeping heat is
		\begin{equation}
			\mathbb{E}Q_{hk}(t)=\frac{1}{2}\int_0^t\int_H\| B^{-1}[2F-CD\log\rho^{ss}](x)\|^2_U d\nu_sds.
		\end{equation}
		
		On the other hand,
		\begin{align}
			&\mathbb{E}F_d(t)=\mathbb{E} EP(t)-\mathbb{E} Q_{hk}(t)\nonumber\\
			&=\frac{1}{2}\int_0^t\int_H\langle B^{-1}[CD\log\rho^{ss}-CD\log\rho_s](x),B^{-1}[4F-CD\log\rho^{ss}-CD\log\rho_s](x)\rangle_U d\nu_s ds.\nonumber
		\end{align}
		
		Expanding $|a+b|^2-|b|^2=|a|^2+2\langle a,b\rangle$,
		with
		\[
		a=B^{-1}(CD\log\rho^{ss}-CD\log\rho_s),
		\qquad
		b=B^{-1}(2F-CD\log\rho^{ss}),
		\]
		it therefore suffices to prove that the cross term vanishes:
		$$\int_H\langle B^{-1}[CD\log\rho^{ss}-CD\log\rho_s](x),B^{-1}[4F-2CD\log\rho^{ss}](x)\rangle_U d\nu_s=0,$$
		i.e.
		$$\int_H\langle D\log\frac{\rho^{ss}}{\rho_s},F-\frac{1}{2}CD\log\rho^{ss}\rangle (x)d\nu_s=0.$$
		
		To prove this identity we use the stationary Fokker--Planck equation for
		$\rho^{ss}$ together with the integration-by-parts formula.
		
		According to Lemma 3.38 in \cite{Prato2004},
		\begin{align}
			&\int_H\langle D\log\rho_s,F-\frac{1}{2}CD\log\rho^{ss}\rangle (x)\rho_sd\mu\nonumber\\
			&=\int_H \langle D\rho_s,F-\frac{1}{2}CD\log\rho^{ss}\rangle (x)d\mu\nonumber\\
			&=-\int_H \dv_\mu (F-\frac{1}{2}CD\log\rho^{ss}) \rho_sd\mu\nonumber\\
			&=-\int_H [\dv(F-\frac{1}{2}CD\log\rho^{ss}) -\langle Q_\infty^{-1}x, F-\frac{1}{2}CD\log\rho^{ss}\rangle]\rho_sd\mu.
		\end{align}
		
		Using the stationary Fokker--Planck equation
		$$\frac{1}{2}\Tr[CD^2\rho^{ss}]+\langle x,A^* D\rho^{ss}\rangle-[\dv F\rho^{ss}+\langle F, D\rho^{ss}\rangle-\langle Q_\infty^{-1}x,F(X)\rangle\rho^{ss}]=0,$$
		together with
		\[
		A=-\frac12CQ_\infty^{-1}
		\]
		in the symmetric Ornstein--Uhlenbeck case,
		after a straightforward computation we obtain
		\begin{align}
			\dv(F-\frac{1}{2}CD\log\rho^{ss})
			&=\langle Q_\infty^{-1}x,F(x)-\frac{1}{2}CD\log\rho^{ss}\rangle+\langle D\log\rho^{ss}, -F(x)+\frac{1}{2}CD\log\rho^{ss}\rangle.\nonumber
		\end{align}
		
		Therefore
		\begin{align}
			&\int_H\langle D\log\rho_s,F-\frac{1}{2}CD\log\rho^{ss}\rangle (x)\rho_sd\mu\nonumber\\
			&=-\int_H [\dv(F-\frac{1}{2}CD\log\rho^{ss}) -\langle Q_\infty^{-1}x, F-\frac{1}{2}CD\log\rho^{ss}\rangle]\rho_sd\mu\nonumber\\
			&=\int_H \langle D\log\rho^{ss}, F-\frac{1}{2}CD\log\rho^{ss}\rangle\rho_sd\mu.
		\end{align}
		
		Hence the cross term vanishes, proving the desired expression for $\mathbb EF_d(t)$.
	\end{proof}
	
	\begin{remark}[Clausius inequality]
		Since $EP(t)=\Delta s(t)+HD_{\mathrm{rel}}(t)$, the inequality
		$\mathbb E\,EP(t)\geqslant0$ is equivalent to
		$\mathbb E\,HD_{\mathrm{rel}}(t)\geqslant-\mathbb E\,\Delta s(t)$,
		the analogue of the classical Clausius inequality between the
		dissipated heat and the entropy change, with the Gaussian reference
		measure playing the role of the heat bath. The separate
		non-negativity of $\mathbb E\,Q_{hk}(t)$ and $\mathbb E\,F_d(t)$ is
		its nonequilibrium refinement.
	\end{remark}
	
	\begin{proposition}
		Under the assumptions of Theorem \ref{thm_ep} and Hypothesis~\ref{hypo_5}, define the relative entropy at time $t$ as
		$$\mathcal F(t)=\mathbb{E} \log\frac{d\nu_t}{d\nu}=\int_H \log\frac{\rho_t}{\rho^{ss}}\rho_t d\mu.$$
		Then
		$$\frac{d\mathcal F(t)}{dt}=-\mathbb{E} f_d(t),$$
		in which $\mathbb{E} f_d(t)=\frac{d}{dt}\mathbb{E} F_d(t)$.
	\end{proposition}
	
	\begin{proof}
		Since
		\[
		\int_H\partial_t\rho_t\,d\mu=0,
		\]
		we obtain
		\begin{align}
			&\frac{d\mathcal F(t)}{dt}=\int_H \log\frac{\rho_t}{\rho^{ss}}\frac{\partial \rho_t}{\partial t}d\mu=\int_H \log\frac{\rho_t}{\rho^{ss}}K_{2,\mu}^*\rho_td\mu\nonumber.
		\end{align}
		
		First, according to Lemma 3.38 in \cite{Prato2004},
		\begin{align}
			&\int_H\langle D\log\frac{\rho^{ss}}{\rho_t},F\rangle (x)\rho_td\mu\nonumber\\
			&=-\int_H \log\frac{\rho^{ss}}{\rho_t}\langle F,D\rho_t\rangle d\mu-\int_H \dv_\mu F \log\frac{\rho^{ss}}{\rho_t}\rho_td\mu\nonumber\\
			&=-\int_H \log\frac{\rho^{ss}}{\rho_t}\left[\langle F,D\rho_t\rangle+(\dv F -\langle Q_\infty^{-1}x, F\rangle)\rho_t\right]d\mu.
		\end{align}
		
		Second, since $L_2$ is self-adjoint in $L^2(H;\mu)$,
		\begin{align}
			&\int_H \log\frac{\rho_t}{\rho^{ss}}\left\{\frac{1}{2}\Tr[CD^2\rho_t]+\langle x,A^* D\rho_t\rangle\right\}d\mu\nonumber\\
			&=\int_H \left\{\frac{1}{2}\Tr[CD^2\log\frac{\rho_t}{\rho^{ss}}]+\langle x,A^* D\log\frac{\rho_t}{\rho^{ss}}\rangle\right\}\rho_td\mu
		\end{align}
		
		Hence
		\begin{align}
			&\frac{d\mathcal F(t)}{dt}\nonumber\\
			&=-\int_H\langle D\log\frac{\rho^{ss}}{\rho_t},F\rangle (x)\rho_td\mu+\int_H \left\{\frac{1}{2}\Tr[CD^2\log\frac{\rho_t}{\rho^{ss}}]+\langle x,A^* D\log\frac{\rho_t}{\rho^{ss}}\rangle\right\}\rho_td\mu=-\mathbb{E}f_d(t),
		\end{align}
		since from the definition of $F_d(t)$, we have
		\begin{align}
			\mathbb{E}f_d(t)&=\int_H\left\langle D\log\frac{\rho^{ss}}{\rho_t}(x), Ax+F(x)\right\rangle \rho_td\mu+\frac{1}{2}\int_H\Tr\left(D^2\log\frac{\rho^{ss}}{\rho_t}(x)C\right)\rho_td\mu.\nonumber
		\end{align}
	\end{proof}
	
	The decomposition $EP=Q_{hk}+F_d$ admits a natural measure-theoretic interpretation: each component is itself the logarithm of a Radon--Nikodym derivative between suitable path measures.
	
	\begin{theorem}[Measure-theoretic characterizations of housekeeping heat and free-energy dissipation]\label{thm_ft_sym}
		Under Hypotheses~\ref{hypo_1}--\ref{hypo_3}, \ref{hypo_5} and~\ref{hypo_6} imposed on $(0,2T)$, together with \textbf{(H-form)} and \textbf{(H-rev)}, assume in addition that $\rho^{ss}>0$ $\mu$-a.e.\ and that the adjoint SPDE, namely the semilinear equation with nonlinear drift $CD\log\rho^{ss}-F$, admits a unique mild solution started from every initial law absolutely continuous with respect to $\nu$. Then
		$$Q_{hk}(t)=\log\frac{dP^{\nu_0}_{[0,t]}}{dP^{ad,\nu_0}_{[0,t]}},$$
		in which $P^{ad,\nu_0}_{[0,t]}$ denotes the law of the adjoint dynamics started from $\nu_0$. Moreover, writing $f_s:=\frac{d\nu_s}{d\nu}=\frac{\rho_s}{\rho^{ss}}$ and $\Delta\log\rho^{ss}(t):=\log\rho^{ss}(X(t))-\log\rho^{ss}(X(0))$,
		$$F_d(t)=\log\frac{dP^{\nu_0}_{[0,t]}}{dP^{r,ad,\nu_t}_{[0,t]}}
		=\log f_0(X(0))-\log f_t(X(t))
		=\Delta s(t)+\Delta\log\rho^{ss}(t),$$
		in which $P^{ad,\nu_t}_{[0,t]}$ denotes the law of the adjoint dynamics started from $\nu_t$, and $P^{r,ad,\nu_t}_{[0,t]}:=(\Theta_t)_\#P^{ad,\nu_t}_{[0,t]}$ is the law of its time-reversed trajectories.
		
		Consequently, both $Q_{hk}$ and $F_d$ satisfy integral fluctuation theorems:
		$$\mathbb{E}e^{-Q_{hk}(t)}=1,~\mathbb{E}e^{-F_d(t)}=1,$$
		for any fixed $t>0$.
	\end{theorem}
	
	\begin{proof}
		For $Q_{hk}$, we compare the forward and the adjoint dynamics with the common Ornstein--Uhlenbeck reference $P^{*,\nu_0}_{[0,t]}$. By Lemma~\ref{lem:girsanov-growth}, applied with the base equation being the reference Ornstein--Uhlenbeck process and with the drifts $F$ and $CD\log\rho^{ss}-F$, both admissible by the individual bounds in Hypothesis~\ref{hypo_5},
		\begin{align}
			\log\frac{dP^{\nu_0}_{[0,t]}}{dP^{*,\nu_0}_{[0,t]}}&=\int_0^t \langle B^{-1}F(X(s)),B^{-1}[dX-AXds]\rangle_U
			-\frac{1}{2}\int_0^t\| B^{-1}F(X(s))\|_U^2ds,
		\end{align}
		and
		\begin{align}
			\log\frac{dP^{ad,\nu_0}_{[0,t]}}{dP^{*,\nu_0}_{[0,t]}}
			&=\int_0^t \langle B^{-1}[CD\log\rho^{ss}-F](X(s)),B^{-1}[dX-AXds]\rangle_U
			\nonumber\\
			&\quad-\frac{1}{2}\int_0^t\| B^{-1}[CD\log\rho^{ss}-F](X(s))\|_U^2\,\mathrm ds.
		\end{align}
		Subtracting the two densities and converting to the Stratonovich form exactly as in the derivation of the steady-state entropy production,
		\begin{align}
			\log\frac{dP^{\nu_0}_{[0,t]}}{dP^{ad,\nu_0}_{[0,t]}}
			&=\int_0^t \langle B^{-1}[2F-CD\log\rho^{ss}](X),\mathrm dW\rangle_U
			+\frac{1}{2}\int_0^t\|B^{-1}[2F-CD\log\rho^{ss}](X)\|_U^2\,\mathrm ds\nonumber\\
			&=\int_0^t\langle B^{-1}[2F-CD\log\rho^{ss}](X),B^{-1}\circ dX\rangle_U,
		\end{align}
		which is the defining expression of $Q_{hk}(t)$. For the fluctuation theorem, apply Lemma~\ref{lem:girsanov-growth} once more, with the adjoint dynamics as the base equation and perturbing drift $2F-CD\log\rho^{ss}$, the negative of the stationary thermodynamic force: it gives $P^{ad,\nu_0}_{[0,t]}\ll P^{\nu_0}_{[0,t]}$ with density $e^{-Q_{hk}(t)}$, whence $\mathbb E\,e^{-Q_{hk}(t)}=P^{ad,\nu_0}_{[0,t]}(\Omega)=1$.
		
		For $F_d$ we argue by a Markov duality, which requires no reversal identification beyond \textbf{(H-rev)} at the stationary law. By \textbf{(H-rev)} with $\lambda=\nu$, whose density curve is constant, $\Theta_t{}_\#P^{\nu}_{[0,t]}$ is the law $P^{ad,\nu}_{[0,t]}$ of the adjoint dynamics started from $\nu$. Applying the involution $\Theta_t$ once more,
		\begin{equation}
			\Theta_t{}_\#P^{ad,\nu}_{[0,t]}=P^{\nu}_{[0,t]}.
			\label{eq:adjoint-involution}
		\end{equation}
		Since the adjoint dynamics is well posed from every initial law absolutely continuous with respect to $\nu$, its law depends on the initial datum through the Markov disintegration, so
		\[
		P^{ad,\nu_t}_{[0,t]}=f_t(X(0))\,P^{ad,\nu}_{[0,t]},
		\qquad
		P^{\nu_0}_{[0,t]}=f_0(X(0))\,P^{\nu}_{[0,t]}.
		\]
		Pushing the first identity forward by $\Theta_t$ and using \eqref{eq:adjoint-involution} together with $X(0)\circ\Theta_t=X(t)$,
		\[
		P^{r,ad,\nu_t}_{[0,t]}
		=\Theta_t{}_\#\big(f_t(X(0))\,P^{ad,\nu}_{[0,t]}\big)
		=f_t(X(t))\,P^{\nu}_{[0,t]}.
		\]
		Since $\rho_t>0$ $\mu$-a.e.\ by Proposition~\ref{prop:transient-reg}(i) and $\rho^{ss}>0$ $\mu$-a.e.\ by assumption, while $\rho_0>0$ by the standing hypotheses, the densities $f_0(X(0))$ and $f_t(X(t))$ are strictly positive almost surely, and therefore
		\[
		\frac{dP^{\nu_0}_{[0,t]}}{dP^{r,ad,\nu_t}_{[0,t]}}
		=\frac{f_0(X(0))}{f_t(X(t))}
		\qquad\text{and conversely},
		\]
		the two laws being mutually absolutely continuous. The endpoint identity follows by taking logarithms, since $-\log\rho_s=s_s$, and it coincides with the defining formula $F_d=EP-Q_{hk}$, because $HD_{\mathrm{rel}}-Q_{hk}=\int_0^t\langle D\log\rho^{ss},\circ\mathrm dX\rangle=\Delta\log\rho^{ss}(t)$. Finally,
		\[
		\mathbb E\,e^{-F_d(t)}
		=\mathbb E^{\nu_0}\Big[\frac{f_t(X(t))}{f_0(X(0))}\Big]
		=\int_H f_t\,\mathrm d\nu
		=\int_H\frac{\rho_t}{\rho^{ss}}\,\mathrm d\nu=1,
		\]
		which is the integral fluctuation theorem.
	\end{proof}
	
	\subsection{When the associated Ornstein--Uhlenbeck process is not symmetric}
	
	In this case, the Kolmogorov forward equation for $\rho_t$ is $\frac{\partial \rho_t}{\partial t}=K_{2,\mu}^*\rho_t$, in which $K_{2,\mu}^*$ denotes the adjoint of $K_2$ in $L^2(H,\mu)$.
	
	\begin{align}
		&\int_HK_2\phi\psi d\mu\nonumber\\
		&=\int_H L_2\phi\psi d\mu+\int_H \langle F(x), D\phi\rangle \psi d\mu=\int_H \phi L_2^*\psi d\mu+\int_H \langle F(x), D\phi\rangle \psi d\mu\nonumber,
	\end{align}
	and due to Lemma 3.38 in \cite{Prato2004},
	\begin{equation}
		\int_H \langle F(x), D\phi\rangle \psi d\mu=-\int_H \phi \langle F(x), D\psi\rangle d\mu-\int_H \dv_\mu F \phi\psi d\mu,\nonumber
	\end{equation}
	thus,
	$$K_{2,\mu}^*\psi=L^*_2\psi-[\langle F,D\psi\rangle+\dv_\mu F\psi],\qquad \psi\in\mathscr E_A(H),$$
	which determines $K_{2,\mu}^*$ in $L^2(H,\mu)$ since $\mathscr E_A(H)$ is a core for $K_{2,\mu}$.
	Hence the stationary density $\rho^{ss}$ satisfies $K_{2,\mu}^*\rho^{ss}=0$, i.e.
	\begin{align}
		&\frac{1}{2}\Tr[CD^2\rho^{ss}]+\langle -Ax-CQ_\infty^{-1}x-F,D\rho^{ss}\rangle+(\langle F(x),Q_\infty^{-1} x\rangle-\dv F)\rho^{ss}=0.\nonumber
	\end{align}

	As before we use the relative state entropy $s_t(x):=-\log\rho_t(x)$, with $\log\rho_t\in D(K_{2,\mu})$ (Proposition~\ref{prop:transient-reg}, Hypothesis~\ref{hypo_6}), whose pathwise It\^o expansion is the $L^2(\mu)$ formula of Theorem~\ref{thm:parabolic-ito}. Unlike the symmetric Ornstein--Uhlenbeck case, the entropy production now carries the additional term $-Rx=(A-A^r)x$, where $A^r=-A-CQ_\infty^{-1}$ is the time-reversed drift. The additional term is controlled by Hypothesis~\ref{hypo_4}(ii)--(iii), under which the irreversible part $A^r-A=R$ is bounded and $EP(t)$ is finite (Remark~\ref{rmk:EPfinite}, Lemma~\ref{lem:two-ou}). Hypothesis~\ref{hypo_4}(ii)--(iii), together with the force admissibility of Hypothesis~\ref{hypo_5}, is the set of sufficient conditions adopted in this paper, rather than a characterization: finiteness of $EP(t)$ is governed directly by the squared thermodynamic force under $\nu_s$.
	
	\begin{proposition}
		Under Hypotheses~\ref{hypo_1}--\ref{hypo_4}$,$ Hypothesis~\ref{hypo_6}$,$ and \textbf{(H-form)}, along a single trajectory $X(t)$, the transient entropy production decomposes as
		$$EP(t)=\Delta s(t)+HD_{\mathrm{rel}}(t),$$
		in which $\Delta s(t)=s_t(X(t))-s_0(X(0))$ and the transient heat dissipation
		\begin{align}
			HD_{\mathrm{rel}}(t)&=\int_0^t \left\langle B^{-1}[2F(X(s))-RX(s)], B^{-1}dX(s)\right\rangle_U\nonumber\\
			&+\frac{1}{2}\Tr\left(DC^{-1}[2F(X(s))-RX(s)]C\right)ds\nonumber\\
			&=\int_0^t\langle B^{-1}[2F-RX](X),B^{-1}\circ dX(s)\rangle_U.
		\end{align}
		All traces below that involve the linear field $Rx$ are understood as limits of partial traces over the eigenbasis $(e_k)$ of $Q_\infty$. Since $RQ_\infty$ is skew-adjoint, the diagonal entries $\langle Re_k,e_k\rangle$ vanish identically, so the formal contribution $-\Tr(R)$ vanishes under this convention, in agreement with the vanishing pairing $\langle Q_\infty^{-1}x,Rx\rangle$ used below. Here we have used $(2A+CQ_\infty^{-1})x=-Rx=2A_ax$, twice the irreversible part of the linear dynamics (cf.\ Remark~\ref{rmk:Hstar-scope}, where $R=-2A_a$). Then
		\begin{equation}
			\mathbb{E} EP(t)=\frac{1}{2}\int_0^t\int_H\| B^{-1}[2F-CD\log\rho_s-Rx]\|_U^2\,d\nu_s\,\mathrm ds\geqslant 0,
		\end{equation}
		and $\mathbb{E} EP(t)=0$ if and only if $2F-Rx=CD\log\rho_s$ for almost every $s\leqslant t$, which implies $\nu_s=\nu$ and the SPDE is time-reversible.
	\end{proposition}
	
	\begin{proof}
		Applying the $L^2(\mu)$ It\^o formula (Theorem~\ref{thm:parabolic-ito}) to $s_t=-\log\rho_t$,
		\begin{align}
			ds_t(X(t))&=-\frac{\partial \log\rho_t}{\partial t}(X(t))dt\nonumber\\
			&-\left\langle D \log\rho_t(X(t)), dX(t)\right\rangle-\frac{1}{2}\Tr\left(D^2 \log\rho_t(X(t))C\right)dt\nonumber
		\end{align}
		We decompose $-D\log\rho_t$ as $C^{-1}[2F-CD\log\rho_t-Rx]-C^{-1}(2F-Rx)$, then
		\begin{equation}
			ds_t(X(t))=ep(t)-hd_{\mathrm{rel}}(t),
		\end{equation}
		in which
		\begin{align}
			ep(t)&=\left\langle B^{-1}\left[2F-CD\log\rho_t-RX\right], B^{-1}dX(t)\right\rangle_U-\frac{\partial \log\rho_t}{\partial t}(X(t))dt\nonumber\\
			&+\frac{1}{2}\Tr\left(DC^{-1}\left[2F-CD\log\rho_t-RX\right]C\right)dt,
		\end{align}
		and
		\begin{align}
			hd_{\mathrm{rel}}(t)&=\left\langle B^{-1}[2F-RX], B^{-1}dX(t)\right\rangle_U+\frac{1}{2}\Tr\left(DC^{-1} [2F-RX]C\right)dt.\nonumber
		\end{align}
		
		Then
		$$EP(t)=\int_0^t ep(s)ds,$$
		hence
		\begin{align}
			\mathbb{E} EP(t)&=\mathbb{E}\int_0^t\left\langle B^{-1}\left[2F-CD\log\rho_s-RX(s)\right], B^{-1}(AX(s)+F(X(s)))\right\rangle_U ds\nonumber\\
			&+\frac{1}{2}\mathbb{E}\int_0^t\Tr\left(DC^{-1}\left[2F-CD\log\rho_s-RX(s)\right]C\right)ds\nonumber\\
			&=\int_0^t\int_H\left\langle C^{-1}\left[2F-CD\log\rho_s-Rx\right], (Ax+F(x))\right\rangle \rho_s d\mu ds\nonumber\\
			&+\frac{1}{2}\int_0^t\int_H\Tr\left(DC^{-1}\left[2F-CD\log\rho_s-Rx\right]C\right)\rho_sd\mu ds,
		\end{align}
		
		Notice that
		\begin{align}
			&\Tr D\left(C^{-1}\left[2F-CD\log\rho_s-Rx\right]C\right)\nonumber\\
			&=\Tr\left(D\left[\frac{2F\rho_s-CD\rho_s-Rx\rho_s}{\rho_s}\right]\right)\nonumber\\
			&=\Tr\left(\left[\frac{D(2F\rho_s-CD\rho_s-Rx\rho_s)}{\rho_s}\right]\right)\nonumber\\
			&-\left\langle \frac{2F\rho_s-CD\rho_s-Rx\rho_s}{\rho_s},D\log\rho_s\right\rangle\nonumber\\
			&=\langle Q_\infty^{-1}x-D\log\rho_s, 2F-CD\log\rho_s-Rx\rangle-2\frac{\partial\log\rho_s}{\partial s}.
		\end{align}
		Since $Q_\infty^{-1}R$ is formally skew-adjoint, this being equivalent to the skew-adjointness of $RQ_\infty$ established in the proof of the reversibility theorem in Section~\ref{sec:ts-nonsym}, we have $\langle Q_\infty^{-1}x,(2A+CQ_\infty^{-1})x\rangle=-\langle Q_\infty^{-1}x,Rx\rangle=0$. Consequently
		\begin{align}
			\mathbb{E} EP(t)
			&=\int_0^t\frac{1}{2}\int_H\| B^{-1}\left[2F-CD\log\rho_s-Rx\right]\|_U^2 \rho_s d\mu ds.
		\end{align}
	\end{proof}
	
	\begin{theorem}[Path-space representation and fluctuation theorem, nonsymmetric case]\label{thm:ep-rn-nonsym}
		Under Hypotheses~\ref{hypo_1}--\ref{hypo_4}$,$ Hypothesis~\ref{hypo_6} imposed on $(0,2T)$, \textbf{(H-form)}, and \textbf{(H-rev)}, let $P_{[0,t]}^{\nu_0}$ and $P_{[0,t]}^{\nu_t}$ denote the laws of the SPDE started from $\nu_0$ and from $\nu_t$ respectively, and let $P_{[0,t]}^{rev,\nu_t}:=(\Theta_t)_\#P_{[0,t]}^{\nu_t}$. Then
		$$EP(t)=\log\frac{dP_{[0,t]}^{\nu_0}}{dP_{[0,t]}^{rev,\nu_t}},$$
		and consequently $\mathbb{E}e^{-EP(t)}=1$ for any $t>0$.
	\end{theorem}
	
	\begin{proof}[Proof of Theorem~\ref{thm:ep-rn-nonsym}]
		Denote by $\tilde\nu_s=\tilde\rho_s\mu$ the marginals of the restarted process $P_{[0,t]}^{\nu_t}$, so that $\partial_s\tilde\rho_s=K_{2,\mu}^*\tilde\rho_s$ and, by uniqueness for the forward equation, $\tilde\rho_s=\rho_{t+s}$. Hence the regularity of $\log\tilde\rho$ and the admissibility of the drift $CD\log\tilde\rho_{t-s}-F=CD\log\rho_{2t-s}-F$ are covered by the assumptions on $(0,2T)$.
		
		By \textbf{(H-rev)} with $\lambda=\nu_t$ (in its general form, with $A^{r}=-A-CQ_\infty^{-1}$), $P_{[0,t]}^{rev,\nu_t}$ is the law of the SPDE with time-dependent drift $A^{r}x+CD\log\tilde\rho_{t-s}(x)-F(x)$, with final distribution $\nu_t$. Let $P^{*,\nu_0}_{[0,t]}$ and $P^{*,\tilde\nu_t}_{[0,t]}$ denote the laws of the forward Ornstein--Uhlenbeck process started from $\nu_0$ and from $\tilde\nu_t$, and $P^{rev,*,\tilde\nu_t}_{[0,t]}$ the law of the reversed Ornstein--Uhlenbeck process (drift $A^{r}$) started from $\tilde\nu_t$. Then
		\begin{align}
			\log\frac{dP_{[0,t]}^{\nu_0}}{dP_{[0,t]}^{rev,\nu_t}}
			&=\log\frac{dP^{\nu_0}_{[0,t]}}{dP^{*,\nu_0}_{[0,t]}}
			-\log\frac{dP^{rev,\nu_t}_{[0,t]}}{dP^{rev,*,\tilde\nu_t}_{[0,t]}}
			-\log\frac{dP^{rev,*,\tilde\nu_t}_{[0,t]}}{dP^{*,\tilde\nu_t}_{[0,t]}}\nonumber\\
			&\quad+\log\frac{dP^{*,\nu_0}_{[0,t]}}{dP^{*,\tilde\nu_t}_{[0,t]}}.
		\end{align}
		The first and second terms are given by Proposition~\ref{prop:girsanov} applied to the drifts $F$ and $CD\log\tilde\rho_{t-s}-F$ against the respective Ornstein--Uhlenbeck references,
		\begin{align}
			\log\frac{dP^{\nu_0}_{[0,t]}}{dP^{*,\nu_0}_{[0,t]}}
			&=\int_0^t\langle B^{-1}F(X),B^{-1}[dX-AXds]\rangle_U
			-\frac12\int_0^t\|B^{-1}F(X)\|_U^2ds,\nonumber
		\end{align}
		\begin{align}
			\log\frac{dP^{rev,\nu_t}_{[0,t]}}{dP^{rev,*,\tilde\nu_t}_{[0,t]}}
			&=\int_0^t\langle B^{-1}[CD\log\tilde\rho_{t-s}-F](X),
			B^{-1}[dX-A^{r}Xds]\rangle_U \nonumber\\
			&-\frac12\int_0^t\|B^{-1}[CD\log\tilde\rho_{t-s}-F](X)\|_U^2ds.
		\end{align}
		For the third term, since both Ornstein--Uhlenbeck laws start from the same initial distribution $\tilde\nu_t$, the common initial density cancels in their ratio and Lemma~\ref{lem:two-ou} yields
		\begin{align}
			\log\frac{dP^{rev,*,\tilde\nu_t}_{[0,t]}}{dP^{*,\tilde\nu_t}_{[0,t]}}
			=\int_0^t\langle B^{-1}RX,\mathrm dW\rangle_U
			-\frac12\int_0^t\|B^{-1}RX\|_U^2ds,
		\end{align}
		the stochastic integral being taken against the same canonical innovation $B^{-1}(\mathrm dX-AX\,\mathrm ds)$ as in the first term. Substituting $\mathrm dX=[AX+F(X)]\,\mathrm ds+B\,\mathrm dW$ converts this innovation to $\mathrm dW+B^{-1}F\,\mathrm ds$, and similarly converts the reversed innovation $B^{-1}(\mathrm dX-A^{r}X\,\mathrm ds)$ to $\mathrm dW+B^{-1}(F-RX)\,\mathrm ds$; the resulting finite-variation cross terms cancel in the difference, and a straightforward expansion yields
		\begin{align}
			\log\frac{dP^{*,\nu_0}}{dP^{*,\tilde\nu_t}_{[0,t]}}(X(\cdot))
			=\log\frac{\rho_0}{\tilde\rho_t}(X(0))
			=\int_0^t-d\log \rho_s(X(s))+d\log \tilde\rho_{t-s}(X(s)).
		\end{align}
		Now substitute $dX=[AX+F(X)]ds+BdW$, note that $dX-A^{r}Xds=[F(X)-RX]ds+BdW$, and expand $\log\rho_s(X(s))$ and $\log\tilde\rho_{t-s}(X(s))$ by the parabolic Dynkin--Fukushima formula (Theorem~\ref{thm:parabolic-ito}), which applies by Hypothesis~\ref{hypo_6} on $(0,2T)$. The terms containing $\tilde\rho$ cancel exactly as in the proof of Theorem~\ref{thm:ep-rn}, and the $R$-terms combine as in the proof of the steady-state theorem of Section~\ref{sec:ts-nonsym}. The result is
		\begin{align}
			\log\frac{dP_{[0,t]}^{\nu_0}}{dP_{[0,t]}^{rev,\nu_t}}
			&=\int_0^t\left\langle B^{-1}\left[2F-CD\log\rho_s-RX\right](X(s)), B^{-1}dX(s)\right\rangle_U-\int_0^t\frac{\partial \log\rho_s}{\partial s}(X(s))ds\nonumber\\
			&+\frac{1}{2}\int_0^t \Tr\left(D\left[2F-CD\log\rho_s-RX\right](X(s))\right)ds=EP(t),
		\end{align}
		which is the entropy production of the preceding proposition.
	\end{proof}

	\begin{remark}
		At steady state $\nu_0=\nu$ and $\nu_t=\nu$, so $P^{\nu_0}_{[0,t]}=P^{ss}_{[0,t]}$ while the reversal of the restarted process, started from $\nu$, is precisely the stationary time reversal $P^{rev,ss}_{[0,t]}$. The ratio defining $EP(t)$ in Theorem~\ref{thm:ep-rn-nonsym} therefore coincides with the ratio defining $EP^{ss}(t)$, whence $EP(t)\equiv EP^{ss}(t)$, in accordance with the symmetric case.
	\end{remark}
	
	We define housekeeping heat by
	\begin{align}
		Q_{hk}(t)&=\int_0^t\left\langle B^{-1}\left[2F-CD\log\rho^{ss}-RX\right], B^{-1}dX(s)\right\rangle_U\nonumber\\
		&+\frac{1}{2}\int_0^t\Tr\left(DC^{-1}\left[2F-CD\log\rho^{ss}-RX\right](X(s))C\right)ds\nonumber\\
		&=\int_0^t\langle B^{-1}[2F-CD\log\rho^{ss}-RX],B^{-1}\circ dX\rangle_U,
	\end{align}
	and free energy dissipation by
	\begin{align}
		F_d(t)&=EP(t)-Q_{hk}(t)\nonumber\\
		&=\int_0^t\left\langle B^{-1}\left[CD\log\rho^{ss}-CD\log\rho_s\right](X(s)), B^{-1}dX(s)\right\rangle_U-\int_0^t\frac{\partial \log\rho_s}{\partial s}(X(s))ds\nonumber\\
		&+\frac{1}{2}\int_0^t\Tr\left(DC^{-1}\left[CD\log\rho^{ss}-CD\log\rho_s\right](X(s))C\right)ds\nonumber\\
		&=\int_0^t\langle B^{-1}(CD\log\rho^{ss}-CD\log\rho_s)(X),B^{-1}\circ dX\rangle_U-\int_0^t\frac{\partial \log\rho_s}{\partial s}(X(s))ds.
	\end{align}
	
	\begin{proposition} Under Hypotheses~\ref{hypo_1}--\ref{hypo_6} and \textbf{(H-form)}, the mean housekeeping heat is
		\begin{equation}
			\mathbb{E}Q_{hk}(t)=\frac{1}{2}\int_0^t\int_H\| B^{-1}[2F-CD\log\rho^{ss}-Rx]\|^2_U d\nu_sds.
		\end{equation}
		
		The mean free energy dissipation is
		\begin{equation}
			\mathbb{E}F_d(t)=\frac{1}{2}\int_0^t\int_H\| B^{-1}[CD\log\rho^{ss}-CD\log\rho_s](x)\|_U^2 d\nu_s ds.
		\end{equation}
		$\mathbb{E}Q_{hk}(t)\geqslant 0$ and $\mathbb{E}Q_{hk}(t)=0$ if and only if $2F-CD\log\rho^{ss}=Rx$. $\mathbb{E}F_d(t)\geqslant 0$ and $\mathbb{E}F_d(t)=0$ if and only if $\rho_t\equiv\rho^{ss}$.
	\end{proposition}
	
	\begin{proof}
		Since
		\begin{align}
			&\Tr C^{-1}\left(D\left[2F-CD\log\rho^{ss}-Rx\right]C\right)\nonumber\\
			&=\langle Q_\infty^{-1}x-D\log\rho^{ss},2F-CD\log\rho^{ss}-Rx\rangle,\nonumber
		\end{align}
		the housekeeping heat can be expressed as
		\begin{align}
			&Q_{hk}(t)\nonumber\\
			&=\int_0^t \langle B^{-1}[2F-CD\log\rho^{ss}-RX],dW\rangle_U \nonumber\\
			&+\int_0^t \langle B^{-1}[2F-CD\log\rho^{ss}-RX], B^{-1}[AX+F(X)]\rangle_U ds\nonumber\\
			&+\frac{1}{2}\int_0^t\langle Q_\infty^{-1}x-D\log\rho^{ss},2F-CD\log\rho^{ss}-Rx\rangle ds\nonumber\\
			&=\int_0^t \langle B^{-1}[2F-CD\log\rho^{ss}-RX],dW\rangle_U \nonumber\\
			&+\frac{1}{2}\int_0^t\| B^{-1}[2F-CD\log\rho^{ss}-RX]\|^2_U ds.
		\end{align}
		
		Hence the mean housekeeping heat is
		\begin{equation}
			\mathbb{E}Q_{hk}(t)=\frac{1}{2}\int_0^t\int_H\| B^{-1}[2F-CD\log\rho^{ss}-Rx]\|^2_U d\nu_sds.
		\end{equation}
		
		On the other hand,
		\begin{align}
			&\mathbb{E}F_d(t)=\mathbb{E} EP(t)-\mathbb{E} Q_{hk}(t)\nonumber\\
			&=\frac{1}{2}\int_0^t\int_H\langle B^{-1}CD\log\tfrac{\rho^{ss}}{\rho_s}(x),\nonumber\\
			&\qquad\qquad B^{-1}[4F-2Rx-CD\log\rho^{ss}-CD\log\rho_s](x)\rangle_U d\nu_s ds.
		\end{align}
		
		It is sufficient to prove
		$$\int_H\langle B^{-1}CD\log\frac{\rho^{ss}}{\rho_s}(x),B^{-1}[4F-2Rx-2CD\log\rho^{ss}](x)\rangle_U d\nu_s=0,$$
		i.e.
		$$\int_H\langle D\log\frac{\rho^{ss}}{\rho_s},F-\frac{1}{2}Rx-\frac{1}{2}CD\log\rho^{ss}\rangle (x)d\nu_s=0.$$
		
		According to Lemma 3.38 in \cite{Prato2004},
		\begin{align}
			&\int_H\langle D\log\rho_s,F-\frac{1}{2}Rx-\frac{1}{2}CD\log\rho^{ss}\rangle (x)\rho_sd\mu\nonumber\\
			&=-\int_H \dv(F-\frac{1}{2}Rx-\frac{1}{2}CD\log\rho^{ss})\rho_sd\mu\nonumber\\
			&+\int_H\langle Q_\infty^{-1}x, F-\frac{1}{2}Rx-\frac{1}{2}CD\log\rho^{ss}\rangle\rho_sd\mu,
		\end{align}
		together with
		\begin{align}
			&\frac{1}{2}\Tr[CD^2\rho^{ss}]+\langle A^{r}x-F,D\rho^{ss}\rangle+(\langle F(x),Q_\infty^{-1} x\rangle-\dv F)\rho^{ss}=0,\nonumber
		\end{align}
		and
		\begin{align}
			&\dv(F-\frac{1}{2}Rx-\frac{1}{2}CD\log\rho^{ss})\nonumber\\
			&=\frac{\dv(F\rho^{ss}-\frac{1}{2}Rx\rho^{ss}-\frac{1}{2}CD\rho^{ss})}{\rho^{ss}}\nonumber\\
			&-\langle F-\frac{1}{2}Rx-\frac{1}{2}CD\log\rho^{ss}, D\log\rho^{ss}\rangle\nonumber\\
			&=\left\langle Q_\infty^{-1}x-D\log\rho^{ss},
			F-\frac12 CD\log\rho^{ss}
			-\frac12 Rx
			\right\rangle,
		\end{align}
		and thus
		\begin{align}
			&\int_H\langle D\log\rho_s,F-\frac{1}{2}CD\log\rho^{ss}-\frac{1}{2}Rx\rangle\rho_sd\mu\nonumber\\
			&=\int_H \langle D\log\rho^{ss}, F-\frac{1}{2}CD\log\rho^{ss}-\frac{1}{2}Rx\rangle\rho_sd\mu,
		\end{align}
		which exactly yields the desired cancellation of the cross term. Hence the stated formula for
		\(\mathbb E F_d(t)\) follows.
	\end{proof}
	
	\begin{proposition}
		Under Hypotheses~\ref{hypo_1}--\ref{hypo_6} and \textbf{(H-form)}, define the relative entropy at time $t$ by
		$$\mathcal F(t)=\mathbb{E} \log\frac{d\nu_t}{d\nu}=\int_H \log\frac{\rho_t}{\rho^{ss}}\rho_t d\mu.$$
		Then
		$$\frac{d\mathcal F(t)}{dt}=-\mathbb{E} f_d(t),$$
		in which $\mathbb{E} f_d(t)=\frac{d}{dt}\mathbb{E} F_d(t)$.
	\end{proposition}
	
	\begin{proof}
		Since
		\[
		\int_H\partial_t\rho_t\,d\mu=0,
		\]
		we have
		\begin{align}
			&\frac{d\mathcal F(t)}{dt}\nonumber\\
			&=\int_H \log\frac{\rho_t}{\rho^{ss}}\frac{\partial \rho_t}{\partial t}d\mu=\int_H \log\frac{\rho_t}{\rho^{ss}}K_{2,\mu}^*\rho_td\mu\nonumber\\
			&=\int_H \log\frac{\rho_t}{\rho^{ss}}\left\{\frac{1}{2}\Tr[CD^2\rho_t]+\langle A^{r}x-F,D\rho_t\rangle\right.\nonumber\\
			&\left.+(\langle F(x), Q_\infty^{-1} x\rangle-\dv F)\rho_t\right\}d\mu.
		\end{align}
		
		According to Lemma 3.38 in \cite{Prato2004},
		\begin{align}
			&\int_H\langle D\log\frac{\rho^{ss}}{\rho_t},F\rangle (x)\rho_td\mu\nonumber\\
			&=-\int_H \log\frac{\rho^{ss}}{\rho_t}\langle F,D\rho_t\rangle d\mu-\int_H \dv_\mu F \log\frac{\rho^{ss}}{\rho_t}\rho_td\mu\nonumber\\
			&=-\int_H \log\frac{\rho^{ss}}{\rho_t}\left[\langle F,D\rho_t\rangle+(\dv F -\langle Q_\infty^{-1}x, F\rangle)\rho_t\right]d\mu.
		\end{align}
		
		Next, by Theorem \ref{adjoint_L_2}, we have
		\begin{align}
			&\int_H \log\frac{\rho_t}{\rho^{ss}}\left\{\frac{1}{2}\Tr[CD^2\rho_t]+\langle x,(A^{r})^* D\rho_t\rangle\right\}d\mu\nonumber\\
			&=\int_H \left\{\frac{1}{2}\Tr[CD^2\log\frac{\rho_t}{\rho^{ss}}]+\langle x,A^* D\log\frac{\rho_t}{\rho^{ss}}\rangle\right\}\rho_td\mu.
		\end{align}
		
		Hence
		\begin{align}
			&\frac{d\mathcal F(t)}{dt}\nonumber\\
			&=-\int_H\langle D\log\frac{\rho^{ss}}{\rho_t},F\rangle (x)\rho_td\mu+\int_H \left\{\frac{1}{2}\Tr[CD^2\log\frac{\rho_t}{\rho^{ss}}]+\langle x,A^* D\log\frac{\rho_t}{\rho^{ss}}\rangle\right\}\rho_td\mu\nonumber\\
			&=-\mathbb{E}f_d(t),
		\end{align}
		where the last identity follows from the expression for the mean free-energy dissipation rate established in the previous proposition, i.e.
		\begin{align}
			\mathbb{E}f_d(t)&=\int_H\left\langle D\log\frac{\rho^{ss}}{\rho_t}(x), Ax+F(x)\right\rangle \rho_td\mu+\frac{1}{2}\int_H\Tr\left(D^2\log\frac{\rho^{ss}}{\rho_t}(x)C\right)\rho_td\mu.\nonumber
		\end{align}
	\end{proof}

	\begin{theorem}[Measure-theoretic characterizations, nonsymmetric case]\label{thm_ft_nonsym}
		Under Hypotheses~\ref{hypo_1}--\ref{hypo_6}, \textbf{(H-form)}, and \textbf{(H-rev)}, assume in addition that $\rho^{ss}>0$ $\mu$-a.e.\ and that the adjoint SPDE, namely the semilinear equation with linear part $A^{r}$ and nonlinear drift $CD\log\rho^{ss}-F$, admits a unique mild solution started from every initial law absolutely continuous with respect to $\nu$. Then
		$$Q_{hk}(t)=\log\frac{dP_{[0,t]}^{\nu_0}}{dP^{ad,\nu_0}_{[0,t]}},$$
		where $P^{ad,\nu_0}_{[0,t]}$ denotes the law of the adjoint dynamics started from $\nu_0$. Moreover, with $f_s:=\frac{d\nu_s}{d\nu}=\frac{\rho_s}{\rho^{ss}}$,
		$$F_d(t)=\log\frac{dP_{[0,t]}^{\nu_0}}{dP^{r,ad,\nu_t}_{[0,t]}}=\log f_0(X(0))-\log f_t(X(t))=\Delta s(t)+\Delta\log\rho^{ss}(t),$$
		where $P^{ad,\nu_t}_{[0,t]}$ denotes the law of the adjoint dynamics started from $\nu_t$, and $P^{r,ad,\nu_t}_{[0,t]}:=(\Theta_t)_\#P^{ad,\nu_t}_{[0,t]}$. Consequently,
		$$\mathbb{E}e^{-Q_{hk}(t)}=1,~\mathbb{E}e^{-F_d(t)}=1,$$
		for any fixed $t>0$.
	\end{theorem}
	
	\begin{proof}
		For $Q_{hk}$ the argument parallels that of Theorem~\ref{thm_ft_sym} with one structural difference: the adjoint dynamics has linear part $A^{r}=A+R$, so, relative to the common Ornstein--Uhlenbeck reference of linear part $A$, the two perturbing drifts are $F$ and $Rx+CD\log\rho^{ss}-F$, whose difference is the stationary force $2F-CD\log\rho^{ss}-Rx$. The three-way comparison (forward vs.\ reference, adjoint vs.\ reference, and the two-Ornstein--Uhlenbeck ratio of Lemma~\ref{lem:two-ou} with linear parts $A$ and $A^{r}$) then yields the same Girsanov expansion as in the symmetric case, the additional $R$-terms from the two-Ornstein--Uhlenbeck ratio cancelling those from the adjoint Girsanov density.
		
		For $F_d$, the Markov duality of Theorem~\ref{thm_ft_sym} applies verbatim. By \textbf{(H-rev)} with $\lambda=\nu$, the reversal $\Theta_t{}_\#P^{\nu}_{[0,t]}$ is the adjoint dynamics with linear part $A^{r}$, the involution identity $\Theta_t{}_\#P^{ad,\nu}_{[0,t]}=P^{\nu}_{[0,t]}$ follows as before, and no step of the duality involves $R$. The endpoint identity likewise reads $F_d(t)=\Delta s(t)+\Delta\log\rho^{ss}(t)$, consistent with $F_d=EP-Q_{hk}$ through $HD_{\mathrm{rel}}-Q_{hk}=\int_0^t\langle D\log\rho^{ss},\circ\mathrm dX\rangle=\Delta\log\rho^{ss}(t)$, the $R$-terms cancelling between $HD_{\mathrm{rel}}$ and $Q_{hk}$.
	\end{proof}

	\section{Applications}
	\label{sec:applications}
	
	In this section, we consider several representative semilinear SPDEs arising in stochastic analysis and statistical physics to illustrate the scope of our theory. For each model, we verify the relevant hypotheses and identify the applicable thermodynamic results. In particular, the examples highlight the roles of reversibility and the bounded-irreversibility condition Hypothesis~\ref{hypo_4}(ii) in determining the applicability of the theory.
	
	To develop a rigorous thermodynamic theory, we impose assumptions that are stronger than those typically satisfied by many classical SPDEs. These conditions allow us to justify the measure-theoretic arguments, particularly the Girsanov transformations and the pathwise It\^o formula. Consequently, our framework does not cover equations with unbounded or merely locally Lipschitz nonlinearities, such as stochastic Navier--Stokes equations and superlinear reaction-diffusion equations. It also excludes models whose natural Gaussian reference field is not supported on \(H\), such as the white-noise free field, whose covariance is not trace class. Extending the theory to these settings under weaker and more natural assumptions remains an important direction for future research.
	
	\subsection{Semilinear stochastic heat equation with spatially correlated noise}
	\label{subsec:M1}
	
	We first consider the semilinear stochastic heat equation driven by spatially correlated additive noise, a setting considered, for example, in \cite{Niggeman22}. Let
	\[
	A=\Delta
	\]
	on the torus, and let
	\[
	B\gamma_k=\alpha_k\gamma_k,
	\qquad
	C\gamma_k=\alpha_k^2\gamma_k,
	\]
	with $\alpha_k>0$ for every $k\neq0$ and
	\[
	\sum_{k\neq0}\alpha_k^2<\infty,
	\]
	so that $B$ is an injective Hilbert--Schmidt operator and the reference covariance is injective.
	We assume that $F$ is bounded and globally Lipschitz.
	
	Since $C$ is trace class whereas $C^{-1}$ is unbounded, this is the canonical example covered by Hypothesis~\ref{hypo_3}\ref{hypo_3.a}, with condition~\ref{hypo_3}\ref{hypo_3.a}(iii) verified in Proposition~\ref{prop:Hdensity-check} below.
	
	The only structural issue comes from the zero Fourier mode. On the full periodic space, $\Delta$ has eigenvalue $0$, so the semigroup is not exponentially stable and the covariance operator does not converge. We therefore work on the mean-zero subspace
	\[
	H_0=\{\gamma_0\}^{\perp},
	\]
	equivalently, one may introduce a positive mass term and replace
	$A=\Delta-mI,\qquad m>0$,
	which shifts the zero eigenvalue to $-m$. On $H_0$, the Laplacian generates an exponentially stable analytic semigroup and
	\[
	AQ_\infty=-\frac12C,
	\]
	so Hypothesis~\ref{hypo_1} is satisfied. Together with a bounded globally Lipschitz $F$, Hypothesis~\ref{hypo_2} also holds: the Lipschitz constant of $F$ must be smaller than the mass $m$ in the massive formulation, or smaller than the spectral gap $\lambda_1$ of $-\Delta$ on $H_0$ in the mean-zero formulation, and in the latter case one also requires $F(H_0)\subset H_0$, as happens for instance when $F$ is composed with the orthogonal projection onto $H_0$.
	
	The associated Ornstein--Uhlenbeck process is symmetric. Therefore the thermodynamic results established for the symmetric case apply, subject to the admissibility of the thermodynamic force discussed next. Since the noise here is correlated, this model lies in the trace-class regime~\ref{hypo_3}\ref{hypo_3.a}, and condition~\ref{hypo_3}\ref{hypo_3.a}(iii) must be verified. It is satisfied as soon as the individual drift bound in Hypothesis~\ref{hypo_5} holds, by the following proposition, whose diagonalization hypotheses hold for the present Fourier model.
	
	\begin{proposition}[A sufficient condition for $\nu\ll\mu$]
		\label{prop:Hdensity-check}
		Assume Hypotheses~\ref{hypo_1} and~\ref{hypo_2} together with Hypothesis~\ref{hypo_3}\ref{hypo_3.a}(i)--(ii), and let $A$ be self-adjoint and simultaneously diagonalizable with $C$, so that $Ae_k=-\lambda_ke_k$ and $Ce_k=\alpha_k^2e_k$ with $\alpha_k>0$ and $\lambda_1=\min_k\lambda_k>0$. Assume the transient regularity of Hypothesis~\ref{hypo_6} and \textbf{(H-form)} for the initial density $\rho_0\equiv1$, and the individual force bound
		\[
		\big|C^{-1/2}F(x)\big|_U^2\leqslant a_F^2+b_F^2|x|^2,
		\qquad x\in H.
		\]
		Then $\nu\ll\mu$, with
		\[
		\mathcal D(\nu\|\mu):=\int_H\log\frac{\mathrm d\nu}{\mathrm d\mu}\,\mathrm d\nu
		\leqslant\frac{a_F^2+b_F^2M_2}{\lambda_1},
		\qquad
		M_2:=\sup_{t\geqslant0}\mathbb E\,|X(t)|^2<\infty .
		\]
	\end{proposition}
	
	\begin{proof}
		Let $\nu_0=\mu$, so that $\rho_0\equiv1$ and $\nu_t=\mu P_t$ has density $\rho_t$ with the standing regularity. Along the flow set
		\[
		\mathcal D_t:=\int_H\rho_t\log\rho_t\,\mathrm d\mu,
		\qquad
		\mathcal I_t:=\int_H\big|C^{1/2}D\log\rho_t\big|^2\rho_t\,\mathrm d\mu .
		\]
		By the parabolic Dynkin--Fukushima formula applied to $\log\rho_t$ (Theorem~\ref{thm:parabolic-ito}), or by the computation of the relative-entropy proposition with $\rho^{ss}$ replaced by $1$, the entropy balance reads
		\[
		\frac{\mathrm d\mathcal D_t}{\mathrm dt}
		=-\tfrac12\mathcal I_t+\mathbb E\big\langle F,D\log\rho_t\big\rangle,
		\]
		the antisymmetric part of $L_2$ contributing zero through the carr\'e du champ identity \eqref{eq:dirichlet-form}. The tensorized Gaussian logarithmic Sobolev inequality gives
		\[
		\mathcal D_t\leqslant\frac{1}{4\lambda_1}\,\mathcal I_t
		\qquad\text{whence}\qquad
		\mathcal I_t\geqslant4\lambda_1\mathcal D_t:
		\]
		indeed, applying Gross's inequality mode-wise, $\operatorname{Ent}_\mu(f^2)\leqslant2\sum_kq_k\int(\partial_kf)^2\mathrm d\mu$ with $q_k=\alpha_k^2/(2\lambda_k)\leqslant\alpha_k^2/(2\lambda_1)$, so $\operatorname{Ent}_\mu(f^2)\leqslant\lambda_1^{-1}\int|C^{1/2}Df|^2\mathrm d\mu$, and the choice $f=\sqrt{\rho_t}$ yields the claim. For the cross term, Young's inequality with $\varepsilon=\tfrac14$ and the force bound give
		\[
		\mathbb E\langle F,D\log\rho_t\rangle
		\leqslant\tfrac14\mathcal I_t+\big(a_F^2+b_F^2\mathbb E|X(t)|^2\big)
		\leqslant\tfrac14\mathcal I_t+a_F^2+b_F^2M_2 ,
		\]
		where $M_2<\infty$ by the standard mild-solution estimate \[ \mathbb E|X(t)|^2\leqslant 3\,e^{2\omega t}\,\mathbb E|X(0)|^2+3\Big(\int_0^t e^{\omega(t-s)}\,\mathrm ds\Big)^{\!2}\|F\|_\infty^2+3\,\Tr Q_t\leqslant 3\,\mathbb E|X(0)|^2+\frac{3\|F\|_\infty^2}{\omega^2}+3\,\Tr Q_\infty, \] valid for all noise types since $\Tr Q_\infty<\infty$ by Hypothesis~\ref{hypo_1}. Combining,
		\[
		\frac{\mathrm d\mathcal D_t}{\mathrm dt}\leqslant-\tfrac14\mathcal I_t+a_F^2+b_F^2M_2
		\leqslant-\lambda_1\mathcal D_t+a_F^2+b_F^2M_2 .
		\]
		Since $\mathcal D_0=0$, Grönwall's inequality yields $\sup_{t\geqslant0}\mathcal D_t\leqslant(a_F^2+b_F^2M_2)/\lambda_1$. In particular the densities $\{\rho_t\}_{t\geqslant0}$ are uniformly integrable in $L^1(H,\mu)$, being entropy-bounded. By [\cite{Prato2004} Theorem~3.17(iii)], $\nu_t$ converges weakly to the invariant measure $\nu$. Uniform integrability and the Dunford--Pettis theorem provide a subsequence with $\rho_{t_n}\rightharpoonup\rho^{*}$ weakly in $L^1(H,\mu)$, and every weak limit point defines a measure $\nu^{*}\ll\mu$ with $\int\psi\,\mathrm d\nu^{*}=\lim_n\int\psi\rho_{t_n}\,\mathrm d\mu=\int\psi\,\mathrm d\nu$ for all bounded continuous $\psi$. Hence $\nu=\nu^{*}\ll\mu$, and lower semicontinuity of the entropy gives $\mathcal D(\nu\|\mu)\leqslant\liminf_n\mathcal D_{t_n}\leqslant(a_F^2+b_F^2M_2)/\lambda_1$.
	\end{proof}
	
	The distinctive feature of this example is the presence of correlated (trace-class) noise. Because $C^{-1}$ is unbounded, the entropy production is finite only when the thermodynamic force
	\[
	C^{-1/2}\big(2F-CD\log\rho^{ss}\big)
	\]
	is square integrable under $\nu$, in particular whenever it has at most linear growth as required by Hypothesis~\ref{hypo_5}. In this case
	\[
	epr
	=
	\frac12
	\int_H
	\Big\langle
	C^{-1}
	\!\left(
	2F-CD\log\rho^{ss}
	\right),
	\,
	2F-CD\log\rho^{ss}
	\Big\rangle
	\,d\nu.
	\]
	Otherwise the entropy production becomes ultraviolet divergent. This example illustrates the main qualitative difference between the trace-class and white-noise settings: under correlated noise, finiteness of the entropy production requires the additional range condition, whereas under white noise it is automatic.
	
	\subsection{Random polymer model}\label{subsec:M2}
	
	We next consider the random polymer model, described by the semilinear stochastic heat equation
	\[
	\frac{\partial u(x,t)}{\partial t}
	=
	\nu\,\partial_x^2u
	+
	F(u)
	+
	\xi(x,t),
	\qquad
	x\in(0,1),
	\]
	where the noise is white in space,
	\[
	C=I,
	\]
	and $F$ is the Nemytskii operator
	\[
	F(u)(x)=f(u(x))
	\]
	generated by a bounded globally Lipschitz function $f$.
	
	The only structural issue comes from the zero Fourier mode under Neumann boundary conditions. Since the Neumann Laplacian has eigenvalue $0$, the semigroup is not exponentially stable and the covariance operator fails to converge. This difficulty is removed either by restricting the dynamics to the mean-zero subspace or, equivalently, by introducing a positive mass term
	\[
	A=\nu\partial_x^2-mI,
	\qquad
	m>0.
	\]
	Then the linear operator generates an exponentially stable analytic semigroup, so Hypothesis~\ref{hypo_1} is satisfied. Together with the bounded globally Lipschitz assumption on $F$, fulfilled once the Lipschitz constant of $F$ is smaller than the mass $m$, Hypothesis~\ref{hypo_2} also holds. 
	
	Since the driving noise is white, $C=I$, Hypothesis~\ref{hypo_3}\ref{hypo_3.b} applies: $C^{-1}$ is bounded, the range condition is automatic, the Hilbert--Schmidt integrability condition being satisfied since the eigenvalues are $\asymp k^2$, and the transition semigroup is strong Feller, so the transient density $\rho_t$ exists and is positive for every initial law $\nu_0\ll\mu$. The stationary theory then applies unconditionally: as shown next, the model is of gradient type, the stationary thermodynamic force vanishes identically, and all steady-state results hold with $EP^{ss}=0$. The transient theory requires in addition the regularity and force-growth conditions of Hypothesis~\ref{hypo_6} on the density curve, which are not automatic and are assumed here.
	
	A key feature of this model is that the drift is automatically of gradient type. Indeed,
	\[
	F(u)=f(u(\cdot))
	=
	-D\Phi(u),
	\]
	where
	\[
	\Phi(u)=\int_0^1V(u(x))\,dx,
	\qquad
	V'=-f.
	\]
	Consequently, the invariant measure satisfies
	\[
	2F=CD\log\rho^{ss},
	\]
	so the dynamics is reversible, and the invariant measure is the Gibbs measure
	\[
	\rho^{ss}\propto
	\exp\!\left(
	-2\int_0^1V(u(x))\,dx
	\right)
	\]
	with respect to the Gaussian reference measure $\mu$. In particular,
	\[
	EP^{ss}=0,
	\qquad
	Q_{hk}=0,
	\]
	and the only nontrivial thermodynamic quantity is the transient free-energy dissipation $F_d(t)$ whenever the initial distribution is not stationary.
	
	This example therefore serves as an equilibrium reference. A genuine nonequilibrium steady state requires breaking the gradient structure, for instance by adding non-potential forcing or coupling to additional nonequilibrium degrees of freedom.
	
	\begin{remark}[A divergence-free nonequilibrium class]
		The simplest genuinely nonequilibrium examples in the symmetric regime are divergence-free drifts. Let $F\in C_b(H;H)$ be globally Lipschitz with $\dv_\mu F=0$ (see~\eqref{finite divergence}) and $F\not\equiv0$. Then the stationary Fokker--Planck equation of Lemma~\ref{lem:rho-ss-stationary} gives $K_{2,\mu}^*\mathbf 1=-\dv_\mu F=0$, so $\rho^{ss}\equiv1$ and $\nu=\mu$. The stationary thermodynamic force reduces to $2B^{-1}F$, which is bounded by the standing range compatibility, so Hypothesis~\ref{hypo_5} and the Girsanov admissibility hold automatically. The entropy production rate is $epr=2\int_H\|B^{-1}F\|_U^2\,\mathrm d\mu>0$, and by the reversibility criterion of Theorem~\ref{Thm_equi_sym} the dynamics is reversible if and only if $F\equiv0$. Such drifts arise, for instance, by rotating a finite-dimensional block of an eigenbasis of $Q_\infty$ with a smooth cutoff, the rotation chosen skew with respect to the $Q_\infty^{-1}$-weighted inner product on that block. This provides the infinite-dimensional analogue of the classical rotational nonequilibrium steady states of finite-dimensional stochastic thermodynamics.
	\end{remark}
	
	For comparison, one may consider a more general stochastic reaction--diffusion equation with Dirichlet boundary conditions,
	\[
	\frac{\partial u}{\partial t}
	=
	\Delta u+F(u)+\xi,
	\]
	where \(F:H\to H\) is a bounded globally Lipschitz nonlinear drift. Since the first Dirichlet eigenvalue is strictly positive, the Laplacian generates an exponentially stable analytic semigroup. Hence, whenever \(F\) satisfies the hypotheses imposed above, the thermodynamic results of the preceding sections apply. If \(F\) is not a gradient field, the dynamics admits a genuine nonequilibrium steady state with positive entropy production. Thus nonequilibrium behaviour comes from the non-gradient part of the drift, rather than from the diffusion operator itself.
	
	\subsection{A divergence-free cylindrical model for the transient theory}
	\label{subsec:divfree}
	
	The hypotheses of the transient theory---Hypothesis~\ref{hypo_6} in general, and \textbf{(H-rev)} beyond the linear class of Proposition~\ref{prop:Hrev-linear}---have so far been assumed rather than verified. We now construct a genuinely nonequilibrium model, within the divergence-free class of the preceding remark, for which \textbf{(H-rev)} and \textbf{(H-form)} can be checked directly, while the remaining transient regularity reduces to standard finite-dimensional density estimates.
	
	Let $A$ be self-adjoint and negative with eigenvalues $-\lambda_k$, where $\lambda_k=\pi^2k^2$ as for the Dirichlet Laplacian on $(0,1)$, and let $C=\mathrm{Id}$. Then $\sum_k\lambda_k^{-1}<\infty$ and $Q_\infty e_k=(2\lambda_k)^{-1}e_k$ is of trace class. Hypothesis~\ref{hypo_3}\ref{hypo_3.b} holds: the range condition in~{\rm(i)} is valid for any self-adjoint $A$ with $\lambda_k>0$, since $\sup_{\lambda>0}\lambda e^{-2\lambda t}<\infty$, while in~{\rm(ii)} one has $\int_0^1t^{-2\alpha}e^{-2\lambda_k t}\,\mathrm dt\asymp\lambda_k^{2\alpha-1}$, so what is needed is $\sum_k\lambda_k^{2\alpha-1}<\infty$ for some $\alpha\in(0,\tfrac12)$; with $\lambda_k=\pi^2k^2$ this holds for every $\alpha\in(0,\tfrac14)$, since $\|e^{tA}\|_{\mathrm{HS}}^2=\sum_k e^{-2\lambda_k t}=O(t^{-1/2})$ as $t\downarrow0$. In the standardized coordinates
	\[
	y_k:=(2\lambda_k)^{1/2}\,x_k,
	\]
	which are independent standard normal under $\mu$, let $J$ be the rotation acting on coordinate pairs, $(Jy)_{2j-1}=-y_{2j}$ and $(Jy)_{2j}=y_{2j-1}$, let $P_{2N}$ denote the projection onto the first $2N$ coordinates, let $\varphi:[0,\infty)\to[0,1]$ be smooth with $\varphi\equiv1$ on $[0,1]$ and $\varphi\equiv0$ on $[4,\infty)$, and define the drift, in these coordinates, by
	\[
	G(y):=\varepsilon\,\varphi\big(|P_{2N}y|^2\big)\,J P_{2N}y,
	\qquad
	F:=(2\Lambda)^{-1/2}\,G\circ(2\Lambda)^{1/2},
	\]
	where $\Lambda$ denotes multiplication by $\lambda_k$ and $\varepsilon>0$.
	
	The field $G$ is cylindrical over the first $2N$ coordinates, bounded, smooth with bounded derivatives and globally Lipschitz, with Lipschitz constant at most $C\varepsilon$, so $F$ inherits the same properties and Hypothesis~\ref{hypo_2} holds once $C\varepsilon<\lambda_1$. It is divergence-free with respect to $\mu$:
	\[
	\dv_{\mathbb R^{2N}}\big(\varphi(|y|^2)Jy\big)
	=2\varphi'(|y|^2)\langle y,Jy\rangle+\varphi(|y|^2)\Tr J=0,
	\qquad
	\langle y,Jy\rangle=0,
	\]
	and the Gaussian divergence transforms under the linear change of coordinates, so $\dv_\mu F=0$. It is not a gradient field, since the line integral of $G$ along the circle $\{|P_{2N}y|=r\}$, $r\leqslant1$, equals $2\pi\varepsilon r^2\neq0$. By Lemma~\ref{lem:rho-ss-stationary}, $\rho^{ss}\equiv1$ and $\nu=\mu$, the stationary force $2B^{-1}F=2F$ is bounded, so Hypothesis~\ref{hypo_5} and the Girsanov admissibility are automatic, and
	\[
	epr=2\int_H|F|^2\,\mathrm d\mu
	=2\varepsilon^2\,\mathbb E\Big[\varphi\big(|P_{2N}y|^2\big)^2
	\sum_{j\leqslant N}\Big(\frac{y_{2j}^2}{2\lambda_{2j-1}}+\frac{y_{2j-1}^2}{2\lambda_{2j}}\Big)\Big]>0,
	\]
	a genuine nonequilibrium steady state of non-gradient type.
	
	The transient hypotheses are verified for initial laws of product form
	\[
	\nu_0=\nu_0^{(N)}\otimes\mu_{>2N},
	\qquad
	\nu_0^{(N)}=N(m_0,\Sigma_0),
	\qquad
	0<\Sigma_0<I,
	\]
	where $0<\Sigma_0<I$ is understood in the order of symmetric positive matrices, so that the block initial density with respect to the standard Gaussian,
	\[
	\rho_0^{(N)}(y)\propto\exp\Big[-\tfrac12\,y^{*}\big(\Sigma_0^{-1}-I\big)y+\text{ affine terms}\Big],
	\]
	has Gaussian decay and the initial law is non-stationary. Since $F$ is cylindrical and the noise is white, the coordinates $k>2N$ stay stationary and independent of the block for all times, while the block evolves as the $2N$-dimensional diffusion
	\[
	\mathrm dy=\big[-\Lambda_N y+G(y)\big]\,\mathrm dt+(2\Lambda_N)^{1/2}\,\mathrm dW^{(2N)},
	\qquad
	\Lambda_N=\operatorname{diag}(\lambda_1,\dots,\lambda_{2N}),
	\]
	whose diffusion matrix is constant and nondegenerate and whose drift is smooth, globally Lipschitz, of at most linear growth, with bounded derivatives of all orders. Its transition densities are smooth and strictly positive, and the full transient density is cylindrical, $\rho_t(x)=\rho_t^{(N)}(y_1,\dots,y_{2N})$.
	\begin{enumerate}
		\item $\rho_0\in C_b^2(H)$ and $\rho_0>0$, being a smooth cylindrical function with Gaussian decay, and $\nu_0\neq\mu$ since $\Sigma_0\neq I$;
		\item (density estimates) for the block equation, being a uniformly parabolic equation with nondegenerate diffusion and Gaussian-type initial data, the following standard bounds hold:
		\begin{equation}
			\big|D\log\rho_t^{(N)}(y)\big|\leqslant C_T(1+|y|),
			\qquad
			\big|D^2\log\rho_t^{(N)}(y)\big|\leqslant C_T(1+|y|^2),
			\qquad 0\leqslant t\leqslant T,
			\label{eq:density-estimates}
		\end{equation}
		the constants at $t=0$ coming from the quadratic $\log\rho_0^{(N)}$;
		\item given the density estimates \eqref{eq:density-estimates}, Hypothesis~\ref{hypo_6} holds in its stated form. Its force-growth clause follows from the linear bound on $D\log\rho_t^{(N)}$, the block components of $2F-D\log\rho_t$ growing at most linearly and the tail components vanishing. Its regularity clauses follow from the two bounds together with the block forward equation: the graph norm of $\log\rho_t$ in $D(K_{2,\mu})$ involves only block derivatives and is uniformly bounded on $[0,T]$, and $\partial_t\log\rho_t\in L^2((0,T)\times H,\mathrm dt\otimes\mu)$, all integrals being finite since $\mu$ is Gaussian;
		\item \textbf{(H-rev)} holds for every initial law of the above form: the block reversal is given by the classical finite-dimensional formula \cite{HaussmannPardoux86,JQQ}, and a direct computation in the block coordinates, using $\sigma\sigma^{*}=2\Lambda_N$ and $Q_N=I$, identifies the reversed drift as
		\[
		-\Lambda_Ny+2\Lambda_N D\log\rho^{(N)}_{t-s}(y)-G(y),
		\]
		which is exactly the $y$-coordinate form of $A^{r}x+CD\log\rho^{\lambda}_{t-s}(x)-F(x)$ with $A^{r}=A$, while the independent tail reverses to an Ornstein--Uhlenbeck process, the case $m=0$ of Proposition~\ref{prop:Hrev-linear}.
		\item \textbf{(H-form)} holds: with $C=\mathrm{Id}$, the form \eqref{eq:dirichlet-form} is the classical Ornstein--Uhlenbeck Dirichlet form on $(H,\mu)$, closed and quasi-regular \cite{AlbeverioRoeckner1991,MaRoeckner1992}, perturbed by the first-order term $\int_H\langle F,Du\rangle v\,\mathrm d\mu$ with bounded cylindrical $F$, which is form-bounded; closability is preserved under form-bounded perturbations, and quasi-regularity together with the association with the unique mild solution follows from the perturbation theory of generalized Dirichlet forms \cite{Stannat1999}.
	\end{enumerate}
	
	Since some results invoke the transient hypotheses on $(0,2T)$, applying \eqref{eq:density-estimates} on $[0,2T]$ covers all such uses. Finally, the adjoint dynamics of Theorem~\ref{thm_ft_sym} reduces here to the drift $-F$, bounded and globally Lipschitz, so its well-posedness is automatic and $\rho^{ss}\equiv1>0$. Every hypothesis of the steady-state theory and, modulo the density estimates \eqref{eq:density-estimates}, every hypothesis of the transient theory of Sections 3.1 and 4.1, \textbf{(H-form)} included, is thus verified for this model, with genuinely non-stationary initial laws, and both integral fluctuation theorems of Theorem~\ref{thm_ft_sym} hold within it under \eqref{eq:density-estimates}.

	\subsection{A two-temperature coupled Ornstein--Uhlenbeck system}
	\label{subsec:M7}
	
	The associated Ornstein--Uhlenbeck processes in all the previous examples  belong to the symmetric scenario. Their nonequilibrium behaviour, when present, is generated by the nonlinear drift, while the associated Ornstein--Uhlenbeck process remains reversible. We now present a genuinely non-symmetric Ornstein--Uhlenbeck process, which illustrates the theory developed in Section~\ref{sec:ts-nonsym}.
	
	Consider two stochastic heat equations coupled through linear exchange,
	\[
	\begin{aligned}
		\frac{\partial u_1}{\partial t}
		&=
		\Delta u_1
		+
		\kappa(u_2-u_1)
		+
		\sqrt{2T_1}\,\xi_1,
		\\
		\frac{\partial u_2}{\partial t}
		&=
		\Delta u_2
		+
		\kappa(u_1-u_2)
		+
		\sqrt{2T_2}\,\xi_2,
	\end{aligned}
	\]
	where $\xi_1$ and $\xi_2$ are independent space--time white noises and
	$T_1,T_2>0$ denote the temperatures of the two heat baths.
	
	Writing
	\[
	X=(u_1,u_2),
	\]
	the system can be rewritten as the linear SPDE
	\[
	dX=(A_0+\kappa K)X\,dt+B\,dW,
	\]
	where
	\[
	A_0=\Delta\otimes I_2,
	\qquad
	K=
	\begin{pmatrix}
		-1&1\\
		1&-1
	\end{pmatrix},
	\]
	and
	\[
	C=BB^*
	=
	\operatorname{diag}(2T_1,2T_2).
	\]
	The coupling is of gradient type, with interaction potential $\tfrac{\kappa}{2}\int_D(u_2-u_1)^2\,\mathrm dx$.
	
	The physical picture is transparent. When $T_1=T_2$, one has $C=2T_1\,\mathrm{Id}$, which commutes with the self-adjoint operator $A=A_0+\kappa K$, so the associated Ornstein--Uhlenbeck process is symmetric by Proposition~\ref{prop:symmetry} and the dynamics is reversible. When $T_1\neq T_2$, energy continuously flows between the two reservoirs through the gradient coupling, producing a genuine nonequilibrium steady state. Unlike the previous examples, this irreversibility originates entirely from the linear part of the dynamics.
	
	To satisfy the hypothesis, we work on a bounded interval and consider the Dirichlet realization of the Laplacian (or equivalently introduce a positive mass term), so that $-\Delta$ has eigenvalues $\lambda_k>0$ with eigenfunctions $e_k$. Then the linear operator generates an exponentially stable analytic semigroup, and so does the bounded symmetric perturbation $A=A_0+\kappa K$ for $\kappa>0$. Moreover, the driving noise is white, so Hypothesis~\ref{hypo_3}\ref{hypo_3.b} applies automatically.
	
	Both $A$ and $Q_\infty$ are block-diagonal with respect to the decomposition $H=L^2(D)\otimes\mathbb R^2$ along the eigenbasis $(e_k)$, with two-by-two blocks $A_k=-\lambda_k I_2+\kappa K$ and $C_k=\operatorname{diag}(2T_1,2T_2)$. Writing $Q_k$ for the corresponding blocks of $Q_\infty$ and solving the two-by-two Lyapunov equations $A_kQ_k+Q_kA_k=-C_k$ explicitly, using $K^2=-2K$, gives
	\[
	q_{11}^{(k)}-q_{22}^{(k)}=\frac{T_1-T_2}{\lambda_k+\kappa},
	\qquad
	2(\lambda_k+\kappa)q_{12}^{(k)}=\kappa\,\big(q_{11}^{(k)}+q_{22}^{(k)}\big),
	\]
	whence
	\[
	R_kQ_k=-2A_kQ_k-C_k
	=\gamma_k
	\begin{pmatrix}
		0&1\\
		-1&0
	\end{pmatrix},
	\qquad
	\gamma_k:=\frac{\kappa(T_1-T_2)}{\lambda_k+\kappa},
	\]
	and
	\[
	A_kQ_k=
	\begin{pmatrix}
		-T_1&-\tfrac{\gamma_k}{2}\\
		\tfrac{\gamma_k}{2}&-T_2
	\end{pmatrix}.
	\]
	Thus $R$ is block-diagonal with uniformly bounded blocks,
	\[
	\|R\|_{L(H)}
	=\sup_k\|R_k\|
	\leqslant\frac{2\kappa\,|T_1-T_2|}{\min(T_1,T_2)},
	\]
	since $R_k=\gamma_kJQ_k^{-1}$ with $J$ the elementary rotation and $Q_k\geqslant\min(T_1,T_2)(-A_k)^{-1}\geqslant\frac{\min(T_1,T_2)}{\lambda_k+2\kappa}\,I_2$. So Hypothesis~\ref{hypo_4}(ii) holds, and Hypothesis~\ref{hypo_4}(iii) holds because $C^{-1/2}$ is bounded. Hypothesis~\ref{hypo_5} holds as well: since $F=0$ and $\rho^{ss}\equiv1$, the stationary force reduces to $-Rx$, of linear growth $|C^{-1/2}Rx|\leqslant\|C^{-1/2}R\|\,|x|$ by Hypothesis~\ref{hypo_4}(iii), while both individual drift bounds are trivial. Consequently, all thermodynamic results apply.
	
	A remarkable feature of this example is that the invariant measure remains Gaussian,
	\[
	\nu=\mu,
	\qquad
	\rho^{ss}\equiv1,
	\]
	since there is no nonlinear drift. Nevertheless, when $T_1\neq T_2$ the dynamics is genuinely irreversible, because the associated Ornstein--Uhlenbeck process is non-symmetric. The steady-state entropy production is therefore generated entirely by the irreversible linear drift,
	\[
	epr
	=
	\frac12
	\int_H
	\left\|
	B^{-1}Rx
	\right\|_U^2
	\,d\nu
	=\sum_{k\geqslant1}epr_k,
	\qquad
	epr_k:=\tfrac12\Tr\big(C_k^{-1}R_kQ_kR_k^{*}\big).
	\]
	The block structure permits an explicit evaluation. Since $JQ_kJ=-(\det Q_k)\,Q_k^{-1}$ for every symmetric two-by-two matrix, and since by the explicit form of $Q_k$ above $\Tr(C_k^{-1}Q_k)=(\lambda_k+\kappa)\det Q_k/(T_1T_2)$, a direct two-by-two computation gives
	\[
	epr_k
	=\frac{\kappa^2(T_1-T_2)^2}{2T_1T_2\,(\lambda_k+\kappa)},
	\]
	and therefore
	\[
	epr
	=\frac{\kappa^2(T_1-T_2)^2}{2T_1T_2}
	\sum_{k\geqslant1}\frac{1}{\lambda_k+\kappa}
	\leqslant
	\frac{\kappa^2(T_1-T_2)^2}{2T_1T_2}\,
	\Tr\big((-A_0)^{-1}\big),
	\]
	the rightmost trace being finite because $Q_\infty$ is of trace class, as Hypothesis~\ref{hypo_1} requires. In particular,
	\[
	epr>0
	\quad\text{if and only if}\quad
	\kappa\neq0\ \text{and}\ T_1\neq T_2,
	\]
	the entropy production vanishing in the two reversible limits of decoupled fields ($\kappa=0$) and of a common bath ($T_1=T_2$). On an interval of length $L$, where $\lambda_k=(\pi k/L)^2$, the spectral sum is explicit, \[ \sum_{k\geqslant1}\frac{1}{\lambda_k+\kappa} =\frac{L}{2\sqrt{\kappa}}\,\coth(L\sqrt{\kappa})-\frac{1}{2\kappa}, \] so the weak-coupling behaviour is $epr\sim\frac{L^2}{6}\cdot\frac{\kappa^2(T_1-T_2)^2}{2T_1T_2}$ as $\kappa\downarrow0$, quadratic in the coupling strength and the temperature mismatch, while the strong-coupling behaviour is $epr\sim\frac{L\,\kappa^{3/2}(T_1-T_2)^2}{4T_1T_2}$ as $\kappa\to\infty$.
	
	This example demonstrates the essential role of the non-symmetric theory developed in this paper. In contrast to the symmetric framework, where irreversibility is induced by the nonlinear drift, here the entropy production is generated solely by the irreversible linear dynamics while the invariant measure remains Gaussian. The model provides a simple infinite-dimensional analogue of a system coupled to two heat reservoirs and illustrates that the present framework naturally captures genuinely nonequilibrium Ornstein--Uhlenbeck processes. The steady-state results of Section~\ref{sec:ts-nonsym} rest on Hypotheses~\ref{hypo_1}--\ref{hypo_5} and \textbf{(H-rev)}. All the hypotheses have been verified above---for Hypothesis~\ref{hypo_4}(i), note that $\rho^{ss}\equiv1\in D(L_2)$---and the following proposition supplies \textbf{(H-rev)}, so that the steady-state theory applies to the two-temperature system unconditionally.
	
	\begin{proposition}[Time-reversal identification for linear systems]
		\label{prop:Hrev-linear}
		Let $F=0$ and assume that $A^{r}$ generates a $C_0$-semigroup (e.g.\
		under Hypothesis~\ref{hypo_4}(ii)). Then \textbf{(H-rev)} holds for every mean-shifted
		Gaussian initial law $\lambda=N(m,Q_\infty)$ with
		$m\in Q_\infty^{1/2}(H)$ (so that $\lambda\sim\mu$ by the Feldman--H\'ajek theorem, Theorem~\ref{thm:feldman-hajek}) and with locally
		square-integrable inhomogeneity,
		\[
		\int_0^t\big\|CQ_\infty^{-1}e^{uA}m\big\|^2\,\mathrm du<\infty
		\qquad\text{for every }t>0:
		\]
		the reversal
		$\Theta_t{}_\#P^{\lambda}_{[0,t]}$ is the law of the unique mild
		solution of
		\[
		\mathrm dY(s)=\big[A^{r}Y(s)+CQ_\infty^{-1}m_{t-s}\big]\,\mathrm ds
		+B\,\mathrm dW(s),
		\qquad Y(0)\sim N(m_t,Q_\infty),
		\]
		where $m_u=e^{uA}m$. The integrability condition is automatic for $m=0$ and holds whenever $A$, $C$ and $Q_\infty$ are simultaneously diagonalizable: writing $Q_\infty e_k=q_ke_k$, $Ce_k=C_ke_k$, $m=Q_\infty^{1/2}h$ and using $2\lambda_kq_k=C_k$ from the Lyapunov equation, where $\lambda_k$ are the eigenvalues of $-A$,
		\[
		\int_0^t\|CQ_\infty^{-1}e^{uA}m\|^2\,\mathrm du
		=\sum_{k\geqslant1}C_kh_k^2\big(1-e^{-2t\lambda_k}\big)
		\leqslant\|C\|\,|Q_\infty^{-1/2}m|^2.
		\]
		In particular \textbf{(H-rev)} holds for the
		two-temperature system of this subsection, and for $m=0$ it reduces
		to the identification of the reversed Ornstein--Uhlenbeck process.
	\end{proposition}
	
	\begin{proof}
		Since $F=0$, the solution is the Ornstein--Uhlenbeck process itself,
		with marginals $\lambda_u=N(m_u,Q_\infty)$; note that
		$m_u=e^{uA}m\in Q_\infty^{1/2}(H)$, because
		$e^{uA}Q_\infty e^{uA^{*}}\leqslant Q_\infty$. The density ratio of
		two Gaussians with the same covariance has $C$-gradient
		$CD\log\rho^{\lambda}_u=CQ_\infty^{-1}m_u$, so the SPDE displayed in
		\textbf{(H-rev)} is precisely the affine equation above.
		
		Let $Z_s:=X_{t-s}$ denote the reversed process. It is Gaussian, and
		Markov (time-inhomogeneous) by \cite{JQQ}. We compute its transition
		kernels. For $0\leqslant s<u\leqslant t$, put $v:=t-u\leqslant
		r:=t-s$. From the mild form,
		$\mathbb E[\langle X_r-m_r,h\rangle\mid X_v]
		=\langle X_v-m_v,e^{(r-v)A^{*}}h\rangle$, whence the joint
		covariance
		$\mathbb E[\langle X_v-m_v,k\rangle\langle X_r-m_r,h\rangle]
		=\langle Q_\infty k,e^{(r-v)A^{*}}h\rangle$.
		Projecting on finitely many directions and passing to the limit in
		the finite-dimensional Gaussian regression formula yields
		\[
		\mathcal L\big(Z_u\,\big|\,Z_s\big)
		=
		N\big(m_{t-u}+e^{(u-s)A^{r}}\,(Z_s-m_{t-s}),\;
		Q_\infty-e^{(u-s)A^{r}}Q_\infty e^{(u-s)(A^{r})^{*}}\big),
		\]
		where the conjugation identity
		$Q_\infty e^{\tau A^{*}}Q_\infty^{-1}=e^{\tau A^{r}}$ on
		$D(Q_\infty^{-1})$ follows from
		$A^{r}=Q_\infty A^{*}Q_\infty^{-1}$ and the generation hypothesis,
		and the covariance reduction
		\[
		Q_\infty-e^{\tau A^{r}}Q_\infty e^{\tau(A^{r})^{*}}
		=\int_0^\tau e^{rA^{r}}Ce^{rA^{r}{}^{*}}\,\mathrm dr
		\]
		is the Lyapunov computation of Lemma~\ref{lem:two-ou}.
		
		On the other hand, the affine equation admits a unique mild solution. The inhomogeneous term $CQ_\infty^{-1}m_{t-s}$ is locally square integrable in time by assumption, so the variation-of-constants integral converges absolutely, since $\sup_{s\leqslant t}\|e^{sA^{r}}\|\leqslant M_r$ gives
		\[
		\int_0^t\big\|e^{(t-s)A^{r}}CQ_\infty^{-1}m_{t-s}\big\|\,\mathrm ds
		\leqslant M_r\sqrt t\,\Big(\int_0^t\|CQ_\infty^{-1}m_{t-s}\|^2\,\mathrm ds\Big)^{1/2}<\infty,
		\]
		and pathwise uniqueness holds as in
		\cite[Theorem~7.5]{Prato1992}, both drifts being Lipschitz. Its transition from $(s,x)$ to time
		$u$ is, by the variation-of-constants formula, the Gaussian
		$N\big(m_{t-u}+e^{\tau A^{r}}(x-m_{t-s}),\;
		\int_0^\tau e^{(\tau-r)A^{r}}Ce^{(\tau-r)A^{r}{}^{*}}\mathrm dr\big)$:
		the mean function $s\mapsto m_{t-s}$ solves
		$\frac{\mathrm d}{\mathrm ds}m_{t-s}=-Am_{t-s}
		=A^{r}m_{t-s}+CQ_\infty^{-1}m_{t-s}$ by
		$A^{r}=-A-CQ_\infty^{-1}$, and the centred part is an
		Ornstein--Uhlenbeck process with drift $A^{r}$. In particular this
		solution is Markov with the same kernels.
		
		Both $Z$ and the solution of the affine equation are Markov processes
		on $C([0,t];H)$ with the same initial law $N(m_t,Q_\infty)$ and the
		same transition kernels; their finite-dimensional distributions
		coincide, and therefore so do their laws.
	\end{proof}

	\section{Conclusion and discussion}
	
	Stochastic thermodynamics has been remarkably successful for finite-dimensional diffusion processes, where trajectory-level thermodynamic quantities admit rigorous probabilistic interpretations. In this work, we have extended this framework to a class of semilinear stochastic partial differential equations. Using the Gaussian invariant measure of the underlying Ornstein--Uhlenbeck process as a reference, we developed a rigorous theory of entropy, heat dissipation, and entropy production in both stationary and transient settings, including the corresponding fluctuation theorems. We also established a decomposition of entropy production into housekeeping heat and free-energy dissipation, with both components represented by logarithms of Radon--Nikodym derivatives between suitable path measures.
	
	The framework builds on several tools from infinite-dimensional stochastic analysis, including the Da Prato--Zabczyk theory of semilinear SPDEs, Girsanov transformations for mild solutions, Kolmogorov equations, and an $L^2(\mu)$ It\^o formula. Together, these tools provide a rigorous probabilistic foundation for the thermodynamic quantities and fluctuation theorems developed here. By allowing both symmetric and non-symmetric Ornstein--Uhlenbeck reference processes, our theory covers not only reversible systems but also nonequilibrium steady states with irreversible linear dynamics.
	
	The mathematical rigor of our framework comes at the cost of relatively restrictive assumptions. In particular, the bounded Lipschitz condition on the nonlinear drift, the trace-class requirement on the reference covariance $Q_\infty$ (ensured either by trace-class noise or by Hilbert--Schmidt smoothing of the semigroup), and the bounded-irreversibility condition Hypothesis~\ref{hypo_4}(ii) exclude several physically important classes of SPDEs from the present framework.
	
	From the perspective of statistical physics, however, we expect the thermodynamic structures established here to extend beyond the present assumptions. These assumptions primarily reflect the analytical requirements of our approach rather than intrinsic limitations of the underlying thermodynamic principles. Extending the theory to broader classes of SPDEs remains an important direction for future research and may require new techniques from infinite-dimensional stochastic analysis and mathematical physics.

	\begin{appendices}

		\section{Feldman--H\'ajek theorem}
		\label{Appendix A}
		
		The following criterion characterizes the equivalence of Gaussian measures on a Hilbert space.
		
		\begin{theorem}[Feldman--H\'ajek theorem {\rm(\cite[Theorem~2.25]{Prato1992})}]
			\label{thm:feldman-hajek}
			Let $H$ be a separable Hilbert space and $\mu=N_{m_1,Q_1}$, $\nu=N_{m_2,Q_2}$ two Gaussian measures on it. Then $\mu$ and $\nu$ are either singular or equivalent; they are equivalent if and only if the following conditions hold.
			\begin{enumerate} [label=(\roman*)]
				\item \label{Feldman-Hajek 1}$Q_1^{1 / 2}(H)=Q_2^{1 / 2}(H)\stackrel{\Delta}{=}H_0$;
				\item $m_1-m_2 \in H_0$;
				\item The operator $\left(Q_1^{-1 / 2} Q_2^{1 / 2}\right)\left(Q_1^{-1 / 2} Q_2^{1 / 2}\right)^*-\operatorname{Id}$ is a Hilbert-Schmidt operator on $\overline{H_0}$.
			\end{enumerate}
		\end{theorem}

		\section{$L^2(\mu)$ Dynkin--Fukushima It\^o formulas and their applications}
		\label{Appendix B}
		
		Since mild solutions generally do not belong to $D(A)$, the classical
		It\^o formula for strong solutions cannot be applied directly.
		We instead use the \emph{$L^2(\mu)$ Dynkin--Fukushima formula} \cite{FOT2011}, which only requires the test function to
		belong to $D(K_{2,\mu})$. In this formulation the operator $A$ enters
		only through the generator $K_{2,\mu}$, rather than as the pointwise
		quantity $AX$. When stronger regularity is available, the same
		identities agree with the classical mild It\^o formula of
		Da Prato--Jentzen--R\"ockner
		\cite{DaPratoJentzenRockner2019}.
		Assume Hypotheses~\ref{hypo_1}--\ref{hypo_3} and $AQ_\infty\in L(H)$. Set
		\[
		J:=AQ_\infty+\frac12 C.
		\]
		The Lyapunov identity then gives $J^*=-J$. On the core $\mathscr E_A(H)$, the pre-form associated with the
		$L^2(H,\mu)$-realization of the Kolmogorov operator of the forward
		SPDE~\eqref{semilinearSPDE} is
		\begin{equation}
			\mathcal E(u,v):=-\int_H K_{2,\mu}u\,v\,\mathrm d\mu
			=\int_H\left\langle\left(\tfrac12C+J\right)Du,Dv\right\rangle\,\mathrm d\mu
			-\int_H\langle F,Du\rangle v\,\mathrm d\mu .
			\label{eq:dirichlet-form}
		\end{equation}
		After the standard lower-bound shift, its closed quasi-regular
		realization is the possibly non-symmetric generalized Dirichlet form
		properly associated with the forward mild solution $X$ of
		\eqref{semilinearSPDE}; see
		\cite{AlbeverioRoeckner1991,MaRoeckner1992,Stannat1999,Trutnau2000}.
		The hypothesis that this form be closable and quasi-regular, properly associated with the mild solution, is stated as \textbf{(H-form)} at the beginning of Section~\ref{sec:transient}. A complete verification in the present setting would require checking the sector bounds of \cite{MaRoeckner1992,Stannat1999}.
		Here $\mu=N(0,Q_\infty)$ is the invariant Gaussian measure of the
		reference Ornstein--Uhlenbeck process and serves as the reference
		measure for the $L^2$-realization. It need not be invariant for the
		nonlinear process $X$, whose invariant measure is $\nu$.
		Thus $C=BB^*$ determines the symmetric diffusion part, $J$ represents
		the antisymmetric part of the linear Ornstein--Uhlenbeck dynamics, and
		$F$ is the nonlinear first-order perturbation.
		
		On $\mathscr E_A(H)$, the corresponding carr\'e du champ is given by
		\[
		\Gamma(u)
		:=
		K_{2,\mu}(u^2)-2uK_{2,\mu}u
		=
		|B^*Du|_U^2
		=
		|C^{1/2}Du|^2.
		\]

		\begin{theorem}[$L^2(\mu)$ Dynkin--Fukushima formula]
			\label{thm:dynkin-ito}
			Under \textbf{(H-form)}, let $(\mathcal E,D(\mathcal E))$ be the closure of
			\eqref{eq:dirichlet-form}, properly associated with $X$, with generator
			$K_{2,\mu}$. For $u\in D(K_{2,\mu})$, there exists an
			$\mathcal E$-quasi-continuous $\mu$-version $\widetilde u$ such that, for
			$\mathcal E$-q.e.\ $x\in H$,
			\begin{equation}
				\widetilde u(X_t)=\widetilde u(X_0)
				+\int_0^tK_{2,\mu}u(X_s)\,\mathrm ds+M_t^{[u]},
				\label{eq:dynkin-ito}
			\end{equation}
			$\mathbb P_x$-a.s.\ for all $t\geqslant0$. Here $M^{[u]}$ is a
			continuous martingale additive functional of finite energy. In the
			mild-solution realization,
			\[
			M_t^{[u]}=\int_0^t\langle B^*Du(X_s),\mathrm dW_s\rangle_U,
			\]
			and
			\begin{equation}
				\big\langle M^{[u]}\big\rangle_t
				=\int_0^t|B^*Du(X_s)|_U^2\,\mathrm ds
				=\int_0^t|C^{1/2}Du(X_s)|^2\,\mathrm ds .
				\label{eq:carre-bracket}
			\end{equation}
		\end{theorem}
		
		\begin{proof}
			For $u\in\mathscr E_A(H)$, the result follows from the
			finite-dimensional It\^o formula applied to the weak formulation of
			\eqref{semilinearSPDE}. Approximation in the closed form domain and the
			Fukushima decomposition extend the identity to $D(\mathcal E)$. For
			$u\in D(K_{2,\mu})$, the zero-energy part is
			\[
			N_t^{[u]}=\int_0^tK_{2,\mu}u(X_s)\,\mathrm ds.
			\]
			The carr\'e du champ and the same approximation identify the martingale
			part and its bracket. See \cite[Chapter~VI]{MaRoeckner1992} and
			\cite[Theorem~4.5]{Trutnau2000}.
		\end{proof}
		
		\begin{theorem}[Parabolic Dynkin--Fukushima formula]
			\label{thm:parabolic-ito}
			Under \textbf{(H-form)}, let $X$ start from $\nu_0=\rho_0\mu$, write $\nu_s=\rho_s\mu$, and assume
			$\rho\in L^\infty(0,T;L^2(H,\mu))$. Suppose that
			\[
			\Psi\in L^\infty(0,T;D(K_{2,\mu}))
			\cap W^{1,2}(0,T;L^2(H,\mu))
			\]
			and
			\[
			\operatorname*{ess\,sup}_{0<s<T}\int_H
			|C^{1/2}D_x\Psi(s,x)|^2\,\mathrm d\nu_s(x)<\infty.
			\]
			Then $\Psi$ has a space--time quasi-continuous version, still denoted by
			$\Psi$, such that
			\begin{equation}
				\Psi(t,X_t)=\Psi(0,X_0)
				+\int_0^t(\partial_s\Psi+K_{2,\mu}\Psi)(s,X_s)\,\mathrm ds
				+\int_0^t\langle B^*D_x\Psi(s,X_s),\mathrm dW_s\rangle_U
				\label{eq:parabolic-ito}
			\end{equation}
			holds $\mathbb P_{\nu_0}$-a.s.\ for all $t\in[0,T]$. The quadratic
			variation of the martingale term is
			\[
			\int_0^t|C^{1/2}D_x\Psi(s,X_s)|^2\,\mathrm ds.
			\]
		\end{theorem}
		
		\begin{proof}
			Apply the generalized Fukushima decomposition to the space--time
			process $(s,X_s)$, whose generator is $\partial_s+K_{2,\mu}$.
			The $L^2(\mu)$-density bound gives the required integrability of the
			finite-variation part, while the square-field bound gives the
			square-integrability of the martingale term; see
			\cite[Chapter~6]{Oshima2013} and \cite{Stannat1999,Trutnau2000}.
		\end{proof}
		
		Accordingly, the formal differential It\^o notation used in the main
		text is understood through the complete identity
		\[
		\mathrm d\Psi(t,X_t)
		=
		\bigl(\partial_t+K_{2,\mu}\bigr)\Psi(t,X_t)\,\mathrm dt
		+
		\mathrm dM_t^{[\Psi]}.
		\]
		The individual terms involving $AX_t$ and
		$\operatorname{Tr}(CD_x^2\Psi)$ are not asserted to be separately
		defined unless the stronger classical regularity is available.
		
		\paragraph{Application to the thermodynamic functionals}
		
		For the relative state entropy $s_t(x)=-\log\rho_t(x)$,
		Hypothesis~\ref{hypo_6} allows us to apply
		Theorem~\ref{thm:parabolic-ito} directly. Where classical regularity
		$\rho_t\in C_b^2(H)$ is additionally available for the density curve, the
		same identity also follows from the classical mild It\^o formula.

		Unlike the thermodynamic functionals considered in the main text, the Gaussian reference term $\Psi(x)=\tfrac12\langle x,Q_\infty^{-1}x\rangle$ is \emph{not} an admissible test function for either It\^o formula. Indeed,
		$Q_\infty^{-1}$ is an unbounded operator, and
		$\Psi(x)=\frac12|Q_\infty^{-1/2}x|^2$
		is finite only on the Cameron--Martin space
		$Q_\infty^{1/2}(H)=D(Q_\infty^{-1/2})$, which is a proper subspace of $H$ with $\mu$-measure zero. Consequently,
		$\Psi$ is not even a finite measurable functional on the support of the Gaussian reference measure $\mu$, and neither the classical mild It\^o formula nor the $L^2(\mu)$ Dynkin--Fukushima formula applies.

		\section{Kolmogorov equations and regularity}
		\label{Appendix C}
		
		Given any smooth enough initial condition $\varphi$, the following parabolic type equation in $H$ is called the \textit{Kolmogorov backward equation} relevant to SPDE (\ref{semilinearSPDE}): 
		\begin{equation}
			\label{Kolmogorov backward equation}
			\left\{\begin{aligned}
				\partial_t v(t,x)&=\frac{1}{2}\Tr[CD^{2}_{x}v(t,x)]\\
				&~~~+\left\langle A x+F(x), D_{x}v(t, x)\right\rangle,~ t>0,x\in D(A),\\
				v(0,x)&=\varphi(x),~ x\in H.
			\end{aligned}\right.
		\end{equation}
		Again, $D^2_x, D_x$ represent the Fréchet derivatives in variable $x$ given in (\ref{Frechet derivatives}).
		\begin{definition}
			We say that $v$ is a strict solution to the problem (\ref{Kolmogorov backward equation}) if 
			\begin{enumerate}[label=(\roman*)]
				\item $v \in C_b^{0,0}([0, T] \times H)$;
				\item $v(t, \cdot) \in C_b^2(H)$ for any $t\geqslant 0$ and $D_x v,D^2_x v$ are continuous on $[0, T] \times H$;
				\item equation (\ref{Kolmogorov backward equation}) is fulfilled for any $x \in D(A), t \geqslant 0$.
			\end{enumerate}
		\end{definition}
		\begin{theorem}
			Assuming Hypothesis \ref{hypo_1}, \ref{hypo_2}, \ref{hypo_3}\ref{hypo_3.a} and that $F\in C^2_b(H;H)$, for any $\varphi\in C_b^2(H)$, the problem (\ref{Kolmogorov backward equation}) has a unique strict solution $v$. Moreover, $v$ is given by 
			$$
			v(t, x)=\mathbb{E}[\varphi(X(t, x))]=P_t \varphi(x), \quad t \geqslant 0, x \in H.
			$$
			
		\end{theorem}
		\begin{proof}
			See [\cite{Prato1992} Theorem 9.25].
		\end{proof}
		
		Correspondingly, if $F$ has finite divergence with respect to $\mu$ (see (\ref{finite divergence}) for definition), we could also consider the \textit{Kolmogorov forward equation} relevant to SPDE (\ref{semilinearSPDE}): 
		\begin{equation}
			\label{Kolmogorov forward equation}
			\left\{
			\begin{aligned}
				\partial_t u(t,x)&=\frac{1}{2}\Tr[CD^{2}_{x}u(t,x)]\\
				&~~~+\left\langle -A x-CQ_{\infty}^{-1}x-F(x), D_{x}u(t, x)\right\rangle\\
				&~~~-\dv_{\mu}F\cdot u(t,x),\quad t>0,x\in D(A)\cap D(CQ_{\infty}^{-1}),\\
				u(0,x)&=\varphi(x),~ x\in H.
			\end{aligned}
			\right.
		\end{equation}
		
		\begin{remark}[Finite-dimensional heuristic]
			In finite dimensions the form of \eqref{Kolmogorov forward equation} follows from the Fokker--Planck equation $\partial_t\nu_t=-\dv[(Ax+F)\nu_t]+\tfrac12\Tr[CD^2\nu_t]$ by writing $\rho_t=\nu_t/\mu$ with $\mu=N(0,Q_\infty)$, so that $D\log\mu=-Q_\infty^{-1}x$. Using the Lyapunov identity $AQ_\infty+Q_\infty A^*=-C$---whence $\Tr(2A+CQ_\infty^{-1})=0$ and $\langle Q_\infty^{-1}x,(2A+CQ_\infty^{-1})x\rangle=0$---a direct computation gives
			$$
			\partial_t\rho_t=\tfrac12\Tr[CD^2\rho_t]-\big\langle Ax+CQ_\infty^{-1}x+F,\,D\rho_t\big\rangle+\big(\langle Q_\infty^{-1}x,F\rangle-\dv F\big)\rho_t=K_{2,\mu}^*\rho_t .
			$$
			In the symmetric (reversible) case $2A=-CQ_\infty^{-1}$ this reduces to $\partial_t\rho_t=L_2\rho_t-[\langle F,D\rho_t\rangle+\dv_\mu F\,\rho_t]$ with $L_2$ self-adjoint. 
		\end{remark}
		
		The backward equation admits classical $C_b^2$ solutions. In contrast, for the forward equation the natural regularity is no longer pointwise but $L^2(H,\mu)$-Sobolev regularity \cite{Prato1992}.
		
		The $\mu$-adjoint semigroup $P_{t,\mu}^*$ is defined through the duality
		$$
		\int_{H} P_t\varphi\;\psi\,\mathrm{d}\mu=\int_{H} \varphi\;P_{t,\mu}^*\psi\,\mathrm{d}\mu,\quad \varphi,\psi\in B_b(H),
		$$
		and we say that $\rho$ is a \emph{strict solution} of \eqref{Kolmogorov forward equation} in $L^2(H,\mu)$ if $\rho\in C\big([0,\infty);D(K_{2,\mu}^*)\big)\cap C^1\big([0,\infty);L^2(H,\mu)\big)$ and the equation holds in $L^2(H,\mu)$ for every $t\geqslant0$, where $K_{2,\mu}^*$ is the $L^2(\mu)$-realization of the forward operator in \eqref{Kolmogorov forward equation}. The regularity available in the present setting is summarized below.

		\begin{theorem}[Forward regularity of the transient density]
			\label{thm:forward-W22}
			Assume Hypothesis~\ref{hypo_1}, Hypothesis~\ref{hypo_2}, either
			regime of Hypothesis~\ref{hypo_3}, that $F\in C_b^2(H;H)$ has finite
			divergence with respect to $\mu$, and $AQ_\infty\in L(H)$. Then:
			\begin{enumerate}[label=(\roman*)]
				\item The Ornstein--Uhlenbeck generator $L_2$ generates a contraction
				semigroup on $L^2(H,\mu)$, the semigroup being Markov with invariant
				measure $\mu$. The Kolmogorov operator $K_{2,\mu}=L_2+\langle F,D\cdot\rangle$
				is a bounded-form perturbation: when $C^{-1}\in L(H)$ it generates a $C_0$-semigroup
				on $L^2(H,\mu)$ with
				\[
				\|P^*_{t,\mu}\|_{L^2\to L^2}\leqslant e^{\omega t},
				\qquad
				\omega=\tfrac12\big\|C^{-1/2}F\big\|_\infty^2,
				\]
				because
				\[
				\Big|\int_H\langle F,D\varphi\rangle\,\varphi\,\mathrm d\mu\Big|
				\leqslant\tfrac12\int_H|C^{1/2}D\varphi|^2\,\mathrm d\mu
				+\omega\int_H\varphi^2\,\mathrm d\mu
				\]
				shifts the symmetric part of the Dirichlet form \eqref{eq:dirichlet-form}
				by the bounded constant $\omega$. For general trace-class $C$ we appeal to the
				$L^2(H,\mu)$ forward-equation theory of [\cite{Prato2004} Chapter~3] and
				make no contraction claim. When $L_2$ is symmetric, $L_2$ itself is
				self-adjoint and non-positive, and the semigroups of $L_2$ and $L_2^*$
				are analytic on $L^2(H,\mu)$. On $D(L_2)\cap D(L_2^*)$,
				\[
				K_{2,\mu}^*\psi
				=
				L^*_2\psi-\langle F,D\psi\rangle-\dv_\mu F\,\psi .
				\]
				\item For every initial density $\rho_0\in L^2(H,\mu)$ with
				$\rho_0>0$, the curve
				\[
				\rho_t:=P_{t,\mu}^*\rho_0
				\]
				is the unique mild solution of
				$\partial_t\rho_t=K_{2,\mu}^*\rho_t$ in $L^2(H,\mu)$, and
				$\rho_t>0$ $\mu$-a.e.; if $\rho_0\in D(K_{2,\mu}^*)$ the solution is
				strict, with
				$\rho\in C([0,\infty);D(K_{2,\mu}^*))\cap C^1([0,\infty);L^2(H,\mu))$.
				In the symmetric case, or when $C=\mathrm{Id}$, analytic smoothing
				additionally yields $\rho_t\in D(K_{2,\mu}^*)$ for every $t>0$.
			\end{enumerate}
			Concerning second-order regularity: $D(L_2)\subseteq W^{2,2}(H,\mu)$
			when $C=\mathrm{Id}$ [\cite{Prato2004} Remark~2.51];
			when $L_2$ is symmetric and $A$ and $C$ commute,
			$D(L_2)\subseteq W^{2,2}_C(H,\mu)$ [\cite{Prato2004Cambridge}
			Theorem~10.2.8]; for general $C$ no characterization is available
			[\cite{Prato2004Cambridge} Remark~10.2.9], and the regularity used
			in the main text is imposed directly through
			Hypothesis~\ref{hypo_6}.
			
			The measure curve $\nu_t=\rho_t\mu$ solves the associated weak
			Fokker--Planck equation in the sense of
			\cite{BogachevDaPratoRoeckner2004}, which provides existence in that
			class, while uniqueness holds in the $L^2(H,\mu)$ mild sense
			established in part~{\rm(ii)}.
		\end{theorem}
		
		\begin{proof}
			The contraction property of the Ornstein--Uhlenbeck semigroup on
			$L^2(H,\mu)$ is the Markov property together with the invariance of
			$\mu$. The quasi-contractive generation under $C^{-1}\in L(H)$
			follows from the displayed form estimate by the standard bounded
			perturbation of the symmetric Dirichlet form \eqref{eq:dirichlet-form};
			the analytic extension of symmetric $L_2$ follows from the spectral
			theorem. The adjoint formula follows by Gaussian integration by parts
			as in Lemma~\ref{lem:rho-ss-stationary}. Part~{\rm(ii)} follows from
			the generation and smoothing properties of $P_{t,\mu}^*$ and its
			positivity preservation; uniqueness holds in the $L^2(H,\mu)$ mild
			class, and \cite{BogachevDaPratoRoeckner2004} is invoked only for the
			existence of weak Fokker--Planck solutions.
		\end{proof}

	\end{appendices}

	%%===========================================================================================%%
	%% If you are submitting to one of the Nature Portfolio journals, using the eJP submission   %%
	%% system, please include the references within the manuscript file itself. You may do this  %%
	%% by copying the reference list from your .bbl file, paste it into the main manuscript .tex %%
	%% file, and delete the associated \verb+\bibliography+ commands.                            %%
	%%===========================================================================================%%
	
	\section*{Acknowledgment}
	H.G. was supported by National Natural Science Foundation of China (T2225001).

	\bibliography{sn-bibliography}

\end{document}